\documentclass[11pt]{article}

\usepackage[usenames,dvipsnames]{xcolor}
\definecolor{Gred}{RGB}{219, 50, 54}
\definecolor{ToCgreen}{RGB}{0, 128, 0}

\usepackage[margin=1in]{geometry}
\usepackage[T1]{fontenc}

\usepackage[scale=0.97]{XCharter} 

\usepackage[libertine,bigdelims,vvarbb,scaled=1.05]{newtxmath} 

\usepackage{babel}
\usepackage[spacing=true,kerning=true,babel=true,tracking=true]{microtype}

\DeclareMathAlphabet{\pazocal}{OMS}{zplm}{m}{n} 
\renewcommand{\mathcal}[1]{\pazocal{#1}}

\usepackage{makecell}
\usepackage{dsfont}

\usepackage{multirow}
\usepackage{amsmath,amsthm}
\usepackage{bm}
\usepackage{bbm}
\usepackage{textgreek}
\usepackage{mathtools}
\usepackage{enumitem}
\usepackage[numbers,comma,sort&compress]{natbib}
\usepackage{authblk}
\usepackage{graphicx}
\usepackage[font=small]{caption}
\usepackage[labelformat=simple]{subcaption}

\usepackage{float}
\usepackage[ruled,vlined]{algorithm2e}
\usepackage{algorithmic}

\usepackage{physics}
\usepackage{footnote}
\usepackage{xcolor}
\usepackage{mathrsfs}
\usepackage{bbm}
\usepackage{braket}
\usepackage{tikz} \usetikzlibrary{quantikz2}
\usepackage[colorlinks]{hyperref}
\usepackage{cleveref}
\hypersetup{
      colorlinks=true,
  citecolor=ToCgreen,
  linkcolor=Sepia,
  filecolor=Gred,
  urlcolor=Gred
  }

\newtheorem{theorem}{Theorem}
\newtheorem{definition}{Definition}
\newtheorem{lemma}{Lemma}

\newtheorem{corollary}{Corollary}
\newtheorem{remark}{Remark}
\newtheorem{problem}{Problem}

\newcommand{\subexp}{\mathrm{subexp}}

\newcommand{\thm}[1]{\hyperref[thm:#1]{Theorem~\ref*{thm:#1}}}
\newcommand{\cor}[1]{\hyperref[cor:#1]{Corollary~\ref*{cor:#1}}}
\newcommand{\defn}[1]{\hyperref[defn:#1]{Definition~\ref*{defn:#1}}}
\newcommand{\lem}[1]{\hyperref[lem:#1]{Lemma~\ref*{lem:#1}}}
\newcommand{\prop}[1]{\hyperref[prop:#1]{Proposition~\ref*{prop:#1}}}
\newcommand{\assum}[1]{\hyperref[assum:#1]{Assumption~\ref*{assum:#1}}}
\newcommand{\fig}[1]{\hyperref[fig:#1]{Figure~\ref*{fig:#1}}}
\newcommand{\tab}[1]{\hyperref[tab:#1]{Table~\ref*{tab:#1}}}
\newcommand{\algo}[1]{\hyperref[algo:#1]{Algorithm~\ref*{algo:#1}}}
\renewcommand{\sec}[1]{\hyperref[sec:#1]{Section~\ref*{sec:#1}}}
\newcommand{\append}[1]{\hyperref[append:#1]{Appendix~\ref*{append:#1}}}
\newcommand{\fac}[1]{\hyperref[fac:#1]{Fact~\ref*{fac:#1}}}
\newcommand{\lin}[1]{\hyperref[lin:#1]{Line~\ref*{lin:#1}}}

\def\>{\rangle}
\def\<{\langle}

\newcommand{\N}{\mathbb{N}}
\newcommand{\Z}{\mathbb{Z}}
\newcommand{\R}{\mathbb{R}}

\newcommand{\E}{\mathbb{E}}
\newcommand{\cM}{\mathcal{M}}
\newcommand{\cH}{\mathcal{H}}
\newcommand{\cT}{\mathcal{T}}

\newcommand{\cP}{\mathcal{P}}
\newcommand{\cD}{\mathcal{D}}
\newcommand{\cI}{\mathcal{I}}

\newcommand{\polylog}{{\rm polylog}}

\DeclareMathOperator{\poly}{poly}

\newcommand{\SWAP}{\textnormal{SWAP}}

\newcommand{\0}{\mathbf{0}}
\renewcommand{\emptyset}{\varnothing}
\def\Tr{\operatorname{tr}}\def\:{\hbox{\bf:}}

\newcommand{\dist}{\text{dist}}

\let\oldnl\nl
\newcommand{\nonl}{\renewcommand{\nl}{\let\nl\oldnl}}

\renewcommand{\epsilon}{\varepsilon}
\begin{document}
%%%%%%%%%%%%%%%%%%%%%%%%%%%%%%%%%%%%%%%%%%%%%%%%%%%%%%%%%%%%%

%%%%%%%%%%%%%%%%%%%%%%%%%%%%%%%%%%%%%%%%%%%%%%%%%%%%%%%%%%%%%
% Make title
\title{Optimal tradeoffs for estimating Pauli observables}
% \author{Anonymous Authors}
\author{Sitan Chen
\thanks{SEAS, Harvard University. Email: \href{mailto:sitan@seas.harvard.edu}{sitan@seas.harvard.edu};}
\qquad\qquad
Weiyuan Gong
\thanks{SEAS, Harvard University. Email: \href{mailto:wgong@g.harvard.edu}{wgong@g.harvard.edu};}
\qquad\qquad
Qi Ye
\thanks{IIIS, Tsinghua University. Email: \href{mailto:yeq22@mails.tsinghua.edu.cn}{yeq22@mails.tsinghua.edu.cn}.}
}
\date{}
\maketitle

\begin{abstract}
    We revisit the problem of \emph{Pauli shadow tomography}: given copies of an unknown $n$-qubit quantum state $\rho$, estimate $\Tr(P\rho)$ for some set of Pauli operators $P$ to within additive error $\epsilon$. This has been a popular testbed for exploring the advantage of protocols with quantum memory over those without: with enough memory to measure two copies at a time, one can use Bell sampling to estimate $|\Tr(P\rho)|$ for all $P$ using $O(n/\epsilon^4)$ copies, but with $k\le n$ qubits of memory, $\Omega(2^{(n-k)/3})$ copies are needed.

    These results leave open several natural questions. How does this picture change in the physically relevant setting where one only needs to estimate a certain \emph{subset} of Paulis? What is the optimal dependence on $\epsilon$? What is the optimal tradeoff between quantum memory and sample complexity? We answer all of these questions:
    \begin{itemize}
        \item For any subset $A$ of Paulis and any family of measurement strategies, we completely characterize the optimal sample complexity, up to $\log |A|$ factors.
        \item We show any protocol that makes $\poly(n)$-copy measurements must make $\Omega(1/\epsilon^4)$ measurements.
        \item For any protocol that makes $\poly(n)$-copy measurements and only has $k < n$ qubits of memory, we show that $\widetilde{\Theta}(\min\{2^n/\epsilon^2, 2^{n-k}/\epsilon^4\})$ copies are necessary and sufficient.
    \end{itemize}

    \noindent The protocols we propose can also estimate the actual values $\Tr(P\rho)$, rather than just their absolute values as in prior work. Additionally, as a byproduct of our techniques, we establish tight bounds for the task of \emph{purity testing} and show that it exhibits an intriguing phase transition not present in the memory-sample tradeoff for Pauli shadow tomography.
\end{abstract}

\clearpage

\tableofcontents

\clearpage

%%%%%%%%%%%%%%%%%%%%%%%%%%%%%%%%%%%%%%%%%%%%%%%%%%%%%%%%%%%%%

\newpage

\setcounter{page}{1}
\section{Introduction}

Learning properties of unknown quantum systems is of fundamental importance for a wide range of applications, from benchmarking quantum devices to learning about the physical world. In this work we focus on a prototypical task in this vein:

\begin{problem}[Pauli shadow tomography]\label{prob:PauliShadow}
    Let $A$ be a known subset of the Pauli group $\cP_n$. Given access to copies of an $n$-qubit unknown state $\rho$, learn the expectation values $\{\tr(P\rho)\}_{P\in A}$ to within additive error $\epsilon$.
\end{problem}

\noindent This problem first arose with the introduction of the variational quantum eigensolver~\cite{peruzzo2014variational}, where a key step is to estimate expectation values $\Tr(H\rho)$ for a given Hamiltonian $H$ and trial state $\rho$ by decomposing $H$ into a linear combination of Pauli operators $P$ and estimating each $\Tr(P\rho)$. This problem has also naturally emerged in other contexts like quantum resource theory, process tomography, and stabilizer state learning (see Section~\ref{sec:RelatedWorks}).

In recent years, this task has played a central role in investigations into statistical overheads incurred by quantum learning protocols for states~\cite{huang2021information,chen2022exponential,huang2022quantum,chen2023tight} and channels~\cite{chen2022quantum,chen2023efficientPauli,caro2022learning} due to near-term constraints. Prior work in this direction showed that for protocols with enough quantum memory to measure two copies at a time, $O(n/\epsilon^4)$ copies of $\rho$ suffice to estimate $|\tr(P\rho)|$ for all $P\in\cP_n$ to within additive error $\epsilon$, whereas for protocols that can only measure one copy at a time, $\Theta(2^n/\epsilon^2)$ copies are necessary and sufficient. In between these extremes, it is also known that with $k \le n$ qubits of quantum memory, $\Omega(2^{(n-k)/3})$ copies are necessary. While these results shed significant light on the striking tradeoffs between sample complexity and quantum memory, they leave open a number of fundamental questions.

\vspace{0.3em}

\noindent \textbf{(1) Bounds depending on $A$.} The above bounds paint a nearly complete picture if one is interested in estimating \emph{all} Pauli observables of a state. In physically relevant settings however, this is not the case. In the variational quantum eigensolver, the Pauli decomposition of the Hamiltonian $H$ often has very particular structure. For example, in quantum chemistry, the second-quantized electronic structure Hamiltonian decomposes into a particular collection of $O(n^4)$ global Pauli operators. More generally, if one wants to find the ground state of a given molecular Hamiltonian $H$, it suffices to be able to estimate the Pauli observables that appear in the decomposition of $H$. 

An extensive line of work has revolved around finding efficient heuristics for partitioning the set $A$ of Pauli operators that arise for physically relevant Hamiltonians into ``compatible'' sets that can be simultaneously measured (see Table 1 in~\cite{huggins2021efficient} for an overview). While these works have obtained increasingly refined upper bounds for certain choices of $A$, it remains open to what extent these bounds are tight. In this work we take a step back and ask a more general question:
\begin{center}
    {\em Can we tightly characterize the sample complexity of Pauli shadow tomography for \emph{any} subset $A$?}
\end{center}
In this work we give a complete answer in the affirmative to this, up to $\log|A|$ factors\--- see Theorem~\ref{thm:general_theorem}.

\vspace{0.3em}

\noindent \textbf{(2) Optimal $\epsilon$ dependence.} While the upper bound showing that $O(n/\epsilon^4)$ copies suffice for two-copy measurements achieves optimal scaling in $n$, the $\epsilon$ dependence is far from ideal: to estimate all $|\Tr(P\rho)|$ for a $20$-qubit state to within error $0.1$ would require hundreds of thousands of measurements. In contrast, the optimal \emph{single-copy} measurement protocol does achieve $1/\epsilon^2$ scaling, albeit at the cost of exponential dependence in $n$. We thus ask:
\begin{center}
    {\em Is there a protocol that achieves the best of both worlds, with sample complexity scaling as $\mathrm{poly}(n)/\epsilon^2$?}
\end{center}
This is a long-standing mystery even for \emph{general} shadow tomography, where the best upper bounds either scale with $\mathrm{polylog}(|A|)/\epsilon^4$~\cite{buadescu2019quantum,aaronson2019gentle,watts2024quantum} or with $|A|/\epsilon^2$ via the naive estimator, and yet no lower bound is known showing that one cannot get the best of both worlds. A similar gap also appears in the context of quantum process learning~\cite{caro2022learning}.

In this work, we answer the above in the negative and show that any protocol that makes $\mathrm{poly}(n)$-copy measurements (e.g., the protocols of~\cite{buadescu2019quantum,aaronson2019gentle,watts2024quantum}) must incur $1/\epsilon^4$ scaling\--- see Theorem~\ref{thm:Pauliepsdep} below.

\vspace{0.3em}

\noindent \textbf{(3) Optimal memory-sample tradeoffs.} When $A = \cP_n$, prior work resolved the dependence on $n$ for single-copy and two-copy measurements. In the intermediate regime where one has $k$ qubits of quantum memory however, the optimal dependence on $n$ and $k$ remained open. In fact, to the best of our knowledge, it wasn't known how to use this quantum memory to give a protocol improving upon the single-copy rate for any value of $0 < k < n$. This motivates us to ask:
\begin{center}
    {\em With bounded quantum memory, what is the optimal dependence on $n$ and $k$ for Pauli shadow tomography?}
\end{center}
We completely pin down the sample complexity in this setting for all bounded-memory protocols that can make $\poly(n)$-copy measurements (see Definition~\ref{def:ModelCCopyKMem})\--- see Theorem~\ref{thm:PauliShadowAllK}.

Our results stem from a fairly general characterization of the sample complexity for Pauli shadow tomography in a setting where we can not only plug in different choices of $A$ but also different families of measurements based on the resources available to the protocol. This framework and the techniques we develop to reason about it are flexible enough to provide tight characterizations for other learning problems as well, e.g. purity testing\---see Theorem~\ref{thm:Purity}.

\section{Our results}

Here we provide a more detailed treatment of the results that we prove in this work.

\subsection{Master theorem} Our starting point is to prove the following generic characterization of the optimal sample complexity achieved by all protocols within a given class, formalized as follows: Given $c\in\mathbb{N}$ and a set $\cM$ of POVMs on $cn$ qubits, define a \emph{$(c,\cM)$ protocol} to be one which in each round gets $c$ copies of $\rho$ and applies a (possibly adaptively chosen) measurement from $\cM$ to $\rho^{\otimes c}$. For example, if $c = 1$ and $\cM$ is all POVMs, then this corresponds to protocols that use arbitrary incoherent measurements. Throughout, we will focus on the natural setting where $c$ is at most polynomial in $n$, which captures all known quantum learning algorithms employed in practice~\cite{huang2021information,huangPredictingManyProperties2020,grier2022sample,huang2022quantum} and in theory~\cite{aaronson2018shadow,aaronson2019gentle,buadescu2019quantum,watts2024quantum}.

To state our master theorem, we first consider a somewhat different learning task. Suppose we want to distinguish between two states $\rho_0, \rho_1$, and all we can do is perform a \emph{single} measurement $M = \{F_s\}_s \in \cM$ on a batch of $c$ copies. To quantify the power of $M$, we can consider some notion of statistical distance between the distributions over outcomes when $\rho^{\otimes c}_1$ is measured versus when $\rho^{\otimes c}_0$ is measured. Concretely, we can consider the \emph{chi-squared divergence} between these distributions, which we denote by $\chi^2_M(\rho_1^{\otimes c} \| \rho_0^{\otimes c})$.\footnote{Given distributions $p, q$ over $\mathcal{X}$, the chi-squared divergence between them is $\chi^2(p\|q) \triangleq \sum_{x\in\mathcal{X}} q(x)\cdot (p(x)/q(x) - 1)^2$.} The smaller $\chi^2_M$ is, the harder it is to distinguish. While $\chi^2_M$ is seemingly tailored to this distinguishing task, our first key finding is that it can be used to tightly characterize \emph{both the upper and lower bounds} on the sample complexity of Pauli shadow tomography.

\begin{theorem}[Master theorem]\label{thm:general_theorem}
The sample complexity of Pauli shadow tomography for any $A\subseteq\cP_n$ using $(c, \cM)$-protocols is characterized by the following quantity: \begin{equation}\label{eq:general_delta}
\delta_{c, \cM}(A)\coloneqq \min_{\pi\in\cD(A)}\max_{M \in \cM}\, \mathbb{E}_{P\sim \pi}\, \chi^2_M\Bigl(\Bigl(\frac{I + 3\epsilon P}{2^n}\Bigr)^{\otimes c} \Big\| \Bigl(\frac{I}{2^n}\Bigr)^{\otimes c}\Bigr)\,, \footnote{Throughout the paper, we will not explicitly exclude $P=I$ even though $(I+3\epsilon I)/2^n$ is not a state. Indeed, $\chi_M^2((I+3\epsilon I)/2^n\|I/2^n)$ is larger than other terms, so the optimal distribution $\pi$ will never put weight on $I$.}
\end{equation}
where $\cD(A)$ denotes the set of probability distributions over $A$. Specifically, for any $A\subseteq\cP_n$, we have that:
\begin{itemize}[leftmargin=*, itemsep=0pt,topsep=0.3em]
    \item Any $(c, \cM)$-protocol requires at least $\Omega(c/(\delta_{c, \cM}(A)))$ copies of $\rho$.
    \item If $c = 1$ and $\cM$ is Pauli-closed,\footnote{We say that a family of POVMs is \emph{Pauli-closed} if it is closed under simultaneous conjugation of all of the POVM elements by the same Pauli operator $P$, for any $P\in\cP_n$.} then there is a $(c,\cM)$-protocol using $O(\log(|A|)/\delta_{1,\cM}(A))$ copies of $\rho$.
    \item For general $c$, if $\cM$ is Pauli-closed and contains all single-copy measurements, then there is a $(c, \cM)$-protocol using $2^{O(c)} \log(\abs{A})/\delta_{c, \cM}(A) + O(\log(|A|) / \epsilon^4)$ copies of $\rho$.
\end{itemize}
\end{theorem}

\noindent To interpret this result, let us first consider the special case of constant $\epsilon$, deferring discussion of the $\epsilon$ dependence to the next section. In this case, when $c = O(1)$ (resp. $c = O(\log n)$), Theorem~\ref{thm:general_theorem} completely characterizes the sample complexity of Pauli shadow tomography up to a $\log(|A|)$ factor (resp. and additionally a $\poly(n)$ factor), for \emph{any} subset of Pauli observables $A$. Previously, such a characterization was only known for the special cases when $A = \cP_n$ and when $A$ is a stabilizer group.

Similar $\chi^2$-divergence terms have appeared in prior works establishing \emph{lower bounds} on the sample complexity of Pauli shadow tomography for the special case of $A = \cP_n$~\cite{huang2021information,chen2022exponential,huang2022quantum}. Indeed, one can interpret these lower bounds as bounding the value of the minimax problem in the special case where the min player only plays the uniform distribution over $\cP_n$. The primary conceptual novelty in our \Cref{thm:general_theorem} is that it implies the minimax problem governs not just the lower, but also the upper bound.

Another feature unique to \Cref{thm:general_theorem} is that we consider the task of estimating the actual expectation values $\Tr(P\rho)$. In contrast, prior works~\cite{huang2021information,chen2022exponential,huang2022quantum} on the special case of $A = \cP_n$ and $\cM$ consisting of all two-copy measurements only obtained tight bounds for the task of estimating the \emph{absolute value} of $\Tr(P\rho)$. Fortunately, as we show, these two tasks are almost equivalent in terms of sample complexity:

\begin{theorem}\label{thm:learn_sign_from_absolute}
    Let $\rho$ be an unknown state. Suppose we are given estimates $\{f_P\}_{P\in A}$ of the absolute values $\{\abs{\tr({P}\rho)}\}_{P \in A}$ which are accurate to additive error $\epsilon$ for every  $P\in A$. Then there is a protocol that takes $\{f_P\}_{P\in A}$, performs single-copy measurements on $O(\log(\abs{A})/\epsilon^4)$ copies of $\rho$, and with probability $9/10$ estimates the expectation values $\{\tr(P\rho)\}_{P\in A}$ to additive error $3\epsilon$.
\end{theorem}

\noindent 
In contrast, the best known previous algorithm required $O(n/\epsilon^2)$-copy measurements~\cite{huang2021information}. We note that \Cref{thm:learn_sign_from_absolute} is just one of several components in the proof of \Cref{thm:general_theorem}; see \Cref{sec:overview} for an overview of the other components.

Finally, we remark that for any $A$, the protocol we give depends on the choice of $A$. It is natural to ask whether this is necessary, or whether one can achieve the optimal rate even using a protocol that is \emph{oblivious} to $A$. This obliviousness is a defining feature of classical shadow protocols, which perform some measurements on copies of the unknown state $\rho$ and output a classical description of $\rho$ that can then be used to estimate arbitrary expectation values~\cite{huangPredictingManyProperties2020,grier2022sample}. In \Cref{sec:Oblivious}, we prove that no such classical shadows-like protocol can achieve optimal sample complexity: if a learning algorithm is oblivious to $A$, then the sample complexity must be the same as learning all Pauli strings in $\cP_n$ even if it just needs to estimate a single expectation value, i.e. any single-copy protocol requires sample complexity exponentially large in $n$  (see Theorem~\ref{thm:Oblivious}).

\subsection{Applying the master theorem.} Next, we turn to the question of how to apply \Cref{thm:general_theorem} to obtain explicit bounds for specific choices of $A$ and $\cM$, summarized in Table~\ref{tab:summary}. These bounds fall under two categories: those for protocols without quantum memory ($c = 1$), and those for protocols with quantum memory ($c > 1$).

\begin{table}[t]
    \centering
    \begin{tabular}{|c|p{0.5cm}|c|c|c|}
        \hline
        Measurements&\multicolumn{2}{c|}{Set $A\subseteq \cP_n$} & Sample complexity & Pointer\\
        \hline
        \multirow{4}{*}{Memory-free ($c = 1$)}  & \multicolumn{2}{c|}{General $A\subseteq \cP_n$} & $\tilde{\Theta}(1/(\epsilon^2\delta_A))$ & \Cref{thm:PauliShadowNoMem}\\
        \cline{2-5}
        & & Union of $m$ Pauli families & $\tilde{\Theta}(m/\epsilon^2)$ & \Cref{thm:UnionDisjoint}\\
        \cline{3-5}
        & & $m$ noncommuting strings & $\tilde{\Theta}(m/\epsilon^2)$ & \Cref{thm:Noncommuting}\\
        \cline{3-5}
        & & $\{X, Y, Z\}^{\otimes n}$ & $\tilde{\Theta}((3/2)^n/\epsilon^2)$ & \Cref{thm:PauliShadowXYZ}\\
        \hline
        Memory-free, Clifford & \multicolumn{2}{c|}{General $A\subseteq \cP_n$} & $\tilde{\Theta}(\zeta_f(G_A)/\epsilon^2)$ & \Cref{thm:Clifford} \\
        \hline
        $c$-copy, $k$-qubit memory & \multicolumn{2}{c|}{All Pauli strings $\cP_n$} & $\tilde{\Theta}(\min\{2^n/\epsilon^2,2^{n-k}/\epsilon^4\})$ & \Cref{thm:PauliShadowAllK}\\
        \hline
        $c$-copy & \multicolumn{2}{c|}{All Pauli strings $\cP_n$} & $\tilde{\Theta}(\min\{2^n/\epsilon^2,1/\epsilon^4\})$ & \Cref{thm:ccopyLower}\\
        \hline
    \end{tabular}
    \caption{Summary of our bounds with different sets $A$ and measurements ($\tilde{\Theta}(\cdot)$ hides (poly)-logarithmic factors). Here, $\delta_A$ is defined in \eqref{eq:DeltaNoMem} and $\zeta_f(G_A)$ is the fractional coloring number of the anti-commutation graph of $A$. The last two rows hold when $c=O(\log n)$, or when $c=\poly(n)$ and $\epsilon<O(1/c)$.}
    \label{tab:summary}
\end{table}

\subsubsection{Pauli shadow tomography without quantum memory}
\label{sec:ourresults_without}

We first focus on protocols without quantum memory, i.e. in which one can only perform incoherent (single-copy) measurements. Prior works~\cite{huang2021information,chen2022exponential} showed that any such protocol for Pauli shadow tomography with $A = \cP_n$ requires at least $\Omega(2^n/\epsilon^2)$ copies of $\rho$. Furthermore, this bound is nearly tight: the classical shadows protocol of~\cite{huangPredictingManyProperties2020} can solve this task using $O(n2^n/\epsilon^2)$ copies. Unfortunately, classical shadows do not achieve the optimal rate for general subsets $A$. As an extreme example, suppose $A$ is a stabilizer group. In this case, $O(n/\epsilon^2)$ copies of $\rho$ are sufficient as we can simply measure in the joint eigenbasis, whereas classical shadows would still require exponentially many copies.

\vspace{0.3em}

\noindent \textbf{A simpler minimax problem.} By specializing \Cref{thm:general_theorem} to $\cM$ consisting of single-copy measurements, we prove the following optimal bound for protocols without quantum memory:

\begin{theorem}\label{thm:PauliShadowNoMem}
    The sample complexity for Pauli shadow tomography of $A$ without quantum memory is characterized by the following quantity:
    \begin{equation}\label{eq:DeltaNoMem}
        \delta_A \coloneqq \min_{\pi \in \cD(A)}\max_{\ket{\psi}} \, \E_{P\sim \pi} \braket{\psi|P|\psi}^2\,,
    \end{equation}
    where $\ket{\psi}$ ranges over all pure states.
    Specifically, for any $A\subseteq\cP_n$, we have that
    \begin{itemize}[leftmargin=*,itemsep=0pt,topsep=0.3em]
        \item Any protocol without quantum memory requires at least $\Omega(1/(\epsilon^2\delta_A))$ copies of $\rho$.
        \item There is a protocol without quantum memory using $O(\log\abs{A}/(\epsilon^2\delta_A))$ copies of $\rho$.
    \end{itemize}
\end{theorem}

\noindent This completely characterizes the sample complexity of Pauli shadow tomography up to a $\log(|A|)$ factor. In contrast to the more general Eq.~\eqref{eq:general_delta}, the maximization over POVMs from $\cM$ has been replaced with a simpler maximization over pure states in the above \Cref{thm:PauliShadowNoMem}.

We note that $\delta_A$ also has a nice dual formulation as an optimization over the set $\Sigma_{\rm sep}$ of separable states, i.e. the convex hull of all two-fold states of the form $\rho\otimes\rho'$ for some $n$-qubit states $\rho, \rho'$. Specifically, we have: 
\begin{align}
\delta_A = \max_{\tau\in\Sigma_{\rm sep}}\min_{P \in A}\tr(\tau P^{\otimes 2})\,. \label{eq:delta_no_mem_dual}
\end{align}
We refer the reader to \Cref{thm:maximin} for details.

Additionally, armed with the simpler characterization from \Cref{thm:PauliShadowNoMem}, we explicitly compute the sample complexity for various special choices of $A$, including $\{X, Y, Z\}^{\otimes n}$ in \Cref{sec:PauliShadowNoXYZ}, union of $m$ disjoint Pauli families in \Cref{sec:PauliShadowNoFamily}, and $m$ noncommuting Pauli strings in \Cref{sec:PauliShadowNoNonCommute}.

\vspace{0.3em}

\noindent \textbf{Combinatorial characterization for Clifford measurements.} As another application of \Cref{thm:PauliShadowNoMem}, we prove guarantees in the practical setting where the single-copy measurements are further restricted to be efficiently implementable, specifically \emph{Clifford} measurements. We prove the following \emph{combinatorial} characterization of the sample complexity of Pauli shadow tomography with such measurements. Given $A\subseteq\cP_n$, define the \emph{anti-commutation graph} $G_A$ to be the graph with vertices indexed by $A$ and edges connecting any pair of observables which anti-commute.

\begin{theorem}\label{thm:Clifford}
    For any $A\subseteq\cP_n$, let $\zeta_f(G_A)$ denote the fractional coloring number of the anti-commutation graph $G_A$ (see \Cref{sec:PauliShadowNoClifford} for definitions). The sample complexity for Pauli shadow tomography of $A$ using Clifford measurements is lower bounded by $\Omega(\zeta_f(G_A)/\epsilon^2)$ and upper bounded by $O(\log|A| \zeta_f(G_A)/\epsilon^2)$.
\end{theorem}

\subsubsection{Pauli shadow tomography with quantum memory}\label{sec:results_with_quantum_memory}
Next, we consider protocols that have access to a small amount of quantum memory. This is the setting in which problems (2) and (3) about the optimal $\epsilon$ dependence and the optimal tradeoff between quantum memory and sample complexity are relevant.
 
Let us first formulate the family of protocols for which we prove sample complexity bounds. Given $c\in\mathbb{N}$ and $k\le n$, we say that a learning algorithm is a \emph{$c$-copy protocol with $k$ qubits of quantum memory} if it proceeds in rounds of the following form (informal, see \Cref{def:ModelCCopyKMem} for formal definition and well-defined post-measurement states for POVMs):
\begin{enumerate}[leftmargin=*,itemsep=0pt]
    \item At the beginning of each round, prepare a $k$-qubit quantum state $\sigma$, initialized to maximally mixed.
    \item Repeat the following $c-1$ times:
    \begin{enumerate}[itemsep=0pt]
        \item Perform arbitrary POVM $\{M_s^\dagger M_s\}_s$ on $\sigma\otimes \rho$, where $M_s \in \mathbb{C}^{2^{n+k}\times 2^k}$. Record the classical readout.
        \item Set $\sigma$ to be the post-measurement state.
    \end{enumerate}
    \item Perform a destructive (rank-1) POVM on $\sigma\otimes \rho$ and record the classical readout.
\end{enumerate}

\noindent The sample complexity of such a protocol is given by the number of times Steps 2(a) and 3 are performed across all rounds \--- see \Cref{def:ModelCCopyKMem} for a formal definition. Note that without the constraint on $c$, this model is equivalent to the model of learning with bounded quantum memory studied in~\cite{chen2022exponential}.

All of the measurements performed in Step 2 can be organized into a single POVM on $c$ copies of $\rho$. We refer to the family of POVMs that can arise in this way as $\cM^k_{c,n}$. By specializing \Cref{thm:general_theorem} to $(c,\cM^k_{c,n})$-protocols, we immediately get an optimal characterization of the sample complexity of protocols in the above model (see Theorem~\ref{thm:PauliShadowK}).  As in Section~\ref{sec:ourresults_without}, it remains to extract a usable bound from this abstract result. For this, we specialize to $A = \cP_n$. 

We now answer (2) and show that if one wants to avoid exponential scaling in $n$, then $1/\epsilon^4$ dependence is necessary for protocols that make arbitrary $\subexp(n)$-copy measurements, for any $k$:

\begin{theorem}[Informal, see Theorem~\ref{thm:ccopyLower}]\label{thm:Pauliepsdep}
    For Pauli shadow tomography of $A = \cP_n$, if $\epsilon < 1/c$, then any protocol that makes arbitrary $c$-copy measurements requires $\Omega(\min\{2^n/(c\epsilon^2), 1/(c^3\epsilon^4)\})$ copies of $\rho$.
\end{theorem}

\noindent As mentioned previously, state-of-the-art protocols for general shadow tomography~\cite{aaronson2019gentle,aaronson2018shadow,buadescu2019quantum,watts2024quantum} all make $\poly(n)$-copy measurements, so Theorem~\ref{thm:Pauliepsdep} provides a fundamental reason for why these methods have all fallen short of achieving the ``parametric'' $1/\epsilon^2$ rate. 

Furthermore, Theorem~\ref{thm:Pauliepsdep} shows that for Pauli shadow tomography, the rate achieved by the union of the Bell basis protocol (2-copy measurements) and the classical shadows protocol (1-copy measurements) is optimal up to $\poly(n)$ factors among all protocols that make arbitrary $\poly(n)$-copy measurements. This implies that no nontrivial interpolation between these two protocols is possible in this regime.

Lastly, it is interesting to contrast the $1/\epsilon^4$ dependence with what is known in the seemingly related setting of \emph{Pauli channel estimation}~\cite{chen2022quantum,chen2023efficientPauli,chen2022tight2,chen2023tight}, the natural quantum channel analogue of Pauli shadow tomography. There, it is known that with $n$ qubits of memory, one can solve the task with $O(n/\epsilon^2)$ sample complexity. \Cref{thm:PauliShadowAllK} implies that such a favorable scaling is not possible for state learning.

Next, we turn to (3) and show an optimal tradeoff between quantum memory and sample complexity.

\begin{theorem}\label{thm:PauliShadowAllK}
For Pauli shadow tomography of $A = \cP_n$,
\begin{itemize}[leftmargin=*,itemsep=0pt]
    \item There is a $2$-copy protocol with $k\leq n$ qubits of memory that uses $O(n\min\{2^n/\epsilon^2,2^{n-k}/\epsilon^4\})$ copies of $\rho$.
    \item Any $c$-copy protocol with $k$ qubits of memory requires $\Omega(\min\{2^n/(c\epsilon^2),e^{-6c\epsilon} 2^{n-k}/(c^3\epsilon^4)\})$ copies of $\rho$.
\end{itemize}
\end{theorem}

\noindent It is helpful to interpret this first in the regime where $\epsilon < 1/c$, similar to Theorem~\ref{thm:Pauliepsdep}, so that $e^{-6c\epsilon} = \Omega(1)$. Then the protocol from the first part of \Cref{thm:PauliShadowAllK} thus simultaneously achieves optimal scaling in $n, k, \epsilon$, up to a single factor of $n$ (which we conjecture to be necessary). This yields a smooth and nearly optimal tradeoff between sample complexity and quantum memory, see Figure~\ref{fig:transition} for an illustration. We note that recently a smooth tradeoff of this nature was achieved for the task of \emph{state tomography}~\cite{chen2024optimal}, but in the regime where $k$ can be much larger than $n$. In contrast, we operate in a regime where there is not enough quantum memory to perform arbitrary $c$-copy measurements for any $c > 2$.

We also note that Theorem~\ref{thm:Pauliepsdep} is interesting even in the regime where $\epsilon$ is large, e.g. constant. In that case, as long as $c = O(\log n)$, the tradeoff achieved in the Theorem is optimal up to $\poly(n)$ factors.

It is an interesting open question whether the picture suggested by Theorems~\ref{thm:Pauliepsdep} and~\ref{thm:PauliShadowAllK} changes when $c$ is unconstrained. In that regime, the lower bound of $\Omega(2^{(n-k)/3})$ from~\cite{chen2022exponential} remains the best known bound.

\subsubsection{Purity testing with quantum memory}

While we primarily focused on shadow tomography in this work, our techniques are flexible and apply to other quantum learning problems as well. In \Cref{sec:Purity}, we consider the \emph{purity testing} problem~\cite{ekert2002direct}. In the simplest incarnation of this task, one has access to copies of an unknown $n$-qubit quantum state $\rho$, and we need to distinguish whether $\rho$ is a pure state or a maximal mixed state. 

Ref.~\cite{chen2022exponential} proved that this task also exhibits striking separations in sample complexity with versus without quantum memory. With $n$ qubits of memory, one can solve this task simply by performing the swap test on $O(1)$ pairs of copies. Without quantum memory, the optimal sample complexity is $\Theta(2^{n/2})$.

Unlike Pauli shadow tomography, however, for purity testing it was an open question to establish \emph{any} upper or lower bound on sample complexity in the intermediate regime where one has $k$ qubits of quantum memory for $0 < k < n$. In this work, we resolve this by establishing an optimal memory-sample tradeoff which is qualitatively different from the one for Pauli shadow tomography:

\begin{theorem}\label{thm:Purity}
    The sample complexity of purity testing for protocols that make $2$-copy measurements and have $k$ qubits of quantum memory is $\Theta(\min\{2^{n-k}, 2^{n/2}\})$.
\end{theorem}

\begin{figure}[htbp]
    \centering
    \includegraphics[width=0.45\textwidth]{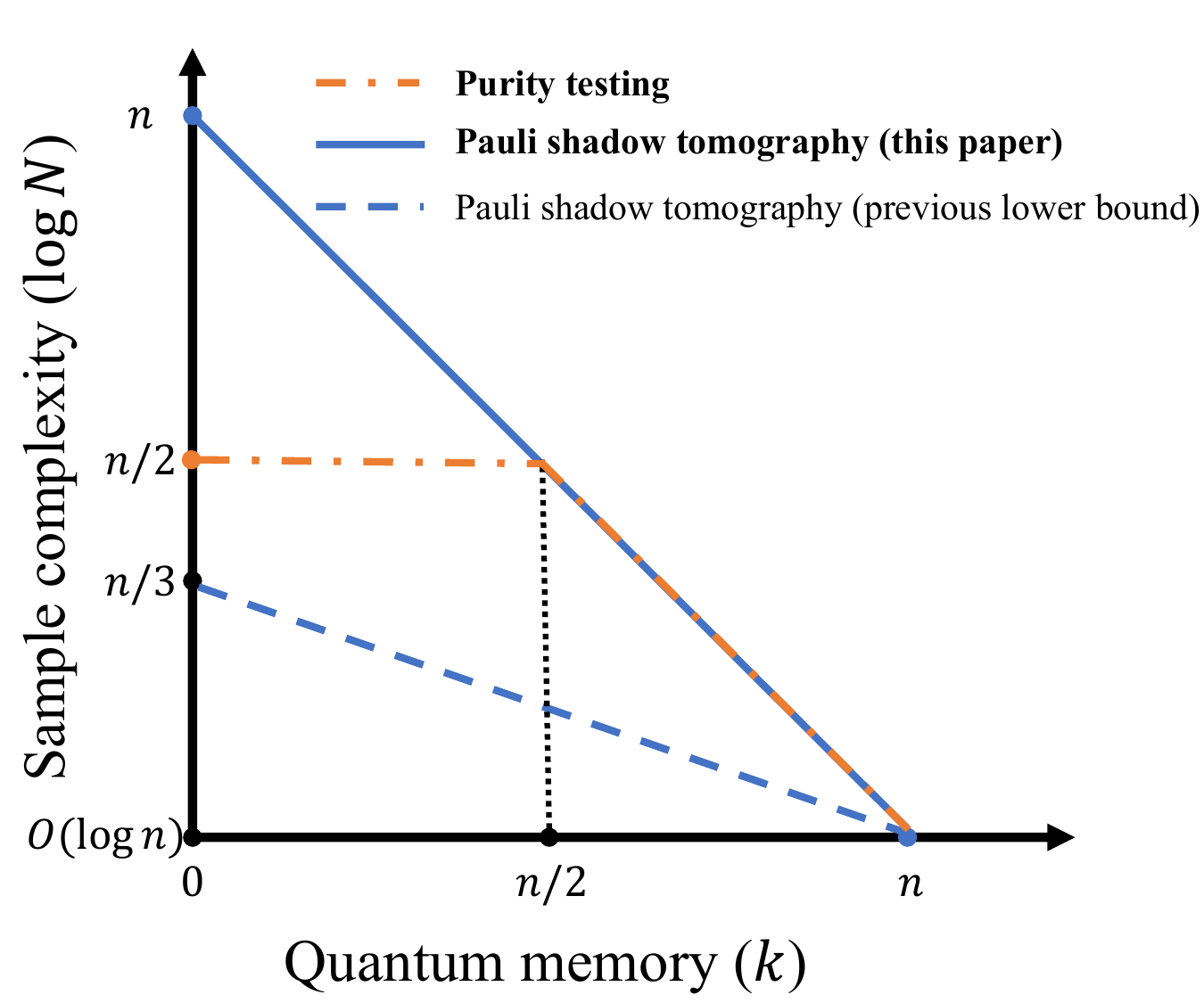}
    \caption{Comparison of bounds in~\Cref{thm:PauliShadowAllK} and~\ref{thm:Purity}, and~\cite{chen2022exponential} for shadow tomography and purity testing.}
    \label{fig:transition}
\end{figure}

\noindent Both upper and lower bounds in \Cref{thm:Purity} are new. Interestingly, our result implies that with $k \le n/2$ qubits of quantum memory, no protocol can improve upon the sample complexity achieved by protocols with no quantum memory. This is in stark contrast to Pauli shadow tomography where, at least for $\epsilon = \Theta(1)$, the sample complexity continuously improves as $k$ increases from $0$ to $n$. Instead, the memory-sample tradeoff for purity testing exhibits a ``second-order phase transition'' at $k = n/2$ (see \Cref{fig:transition}).

\section{Overview of techniques}
\label{sec:overview}
Here, we provide a high-level discussion of the techniques for all theorems mentioned so far.

\paragraph{Minimax lower bound.} Our lower bounds for Pauli shadow tomography throughout this paper exploit the ``learning tree" framework introduced in~\cite{aharonov1997fault,bubeck2020entanglement,chen2022exponential,huang2021information,chen2022complexity,chen2023efficientPauli}. Roughly speaking, in this approach one models the learning protocol as a decision tree, and any choice of unknown state $\rho$ will induce a distribution on the leaves. The goal is to define a pair of ensembles of $\rho$ such that the resulting leaf distributions are statistically indistinguishable unless the tree is sufficiently deep. For Pauli shadow tomography, these two ensembles can naturally be taken to be $\rho_P \triangleq (I + 3\epsilon P)/2^n$ for $P$ sampled from some distribution $\pi$ over $A$, and $\rho_m \triangleq I/2^n$. As this general technique is by now standard, here we only highlight the features of our proof that are new to this work.

Our starting point is a recurring shortcoming with existing sample complexity lower bounds for bounded-memory protocols~\cite{chen2022exponential,chen2022quantum,chen2023efficient}, namely that they all lose a factor of $3$ in the exponent. This stems from a specific technical reason: these arguments all rely on arguing that after one measurement of $\rho$ sampled from either ensemble, the quantum memory does not differ too much in trace distance  under either ensemble, in expectation over $P$. In the language of our work, this can be thought of as a bound on a certain \emph{matrix-valued} version of the $\delta_{c,\cM}$ quantity in Eq.~\eqref{eq:general_delta} (see e.g. Eq. (133) in~\cite{chen2022exponential}). Unfortunately, it is unclear how to optimally ``chain together'' bounds in expectation for this matrix term; that is, given that little information is gained in expectation from one round of the protocol, at least as measured by this matrix term, how do we argue that little information is gained over multiple rounds? For this, the aforementioned works resort to a certain lossy pruning argument, and this is where they lose the factor of $3$. In contrast, we eschew such an argument by identifying a clean \emph{scalar-valued} quantity $\delta_{c,\cM}$ that correctly quantifies the amount of information gained in a single measurement, and which can be chained together in a lossless fashion via a martingale analysis like in~\cite{chen2022complexity,chen2022tight}. The minimax problem in~\Cref{thm:general_theorem} falls out quite naturally from this proof technique, and we defer the details to Section~\ref{sec:GeneralThm}.

\vspace{0.3em}

\paragraph{Optimal learning protocol via Pauli conjugation.} What is perhaps more surprising is that the lower bound that naturally emerges from the martingale analysis is actually tight. Here we outline the main ingredients in the protocol we design to give a matching upper bound. 

As a first step, reminiscent of Yao's principle \cite{yaoProbabilisticComputationsUnified1977}, 
we can use the minimax theorem to rewrite $\delta_{c, \cM}(A)$ as 
\begin{equation*}
    \delta_{c, \cM}(A) = \sup_{\tau\in \cD(\cM)}\min_{P\in A}\E_{M\sim \tau}\chi^2_M(\rho_P^{\otimes c}\|\rho_m^{\otimes c}),
\end{equation*}
where $\cD(\cM)$ is the set of discrete distributions over $\cM$. One can think of this as guaranteeing the existence of a ``locally good distinguisher'' between $\rho_P$ and $\rho_m$: specifically, it ensures there exists a discrete distribution $\tau$ such that in expectation, a POVM from $\tau$ has ``advantage'' at least $\delta_{c,\cM}$ (as measured by chi-squared divergence) in distinguishing between $\rho_P^{\otimes c}$ from $\rho_m^{\otimes c}$, for any $P \in A$. How do we go from this locally good distinguisher to a protocol for Pauli shadow tomography?

To get intuition, let us first consider the case of $c = 1$. In this case the dual simplifies further to Eq.~\eqref{eq:delta_no_mem_dual}: in place of a distribution $\tau$ over POVMs, we get a separable state $\tau$ for which $\tr(\tau P^{\otimes 2}) \ge \delta_A$ for every $P\in A$. For simplicity, suppose $\tau = \sigma\otimes \sigma$ so that $\tr(P\sigma)^2 \ge \delta_A$. Here is the key idea for how to leverage $\sigma$: even though the $c = 1$ setting precludes Bell basis measurements on $\rho^{\otimes 2}$, we can perform Bell basis measurements on $\rho\otimes\sigma$ to estimate $\Tr(P\sigma)\Tr(P\rho)$ using only single-copy measurements of $\rho$. The fact that $\Tr(P\sigma)^2 \ge \delta_A$ for all $P\in A$ ensures that we can safely divide our estimates for $\Tr(P\sigma)\Tr(P\rho)$ by $\Tr(P\sigma)$ to get the desired estimates for $\Tr(P\rho)$.

While this argument is clean, it heavily exploits the structure of the dual in the $c = 1$ case. Fortunately, there is an equivalent interpretation of the above argument that does extend beyond $c = 1$. Let $\epsilon_{P, Q}\in \{-1, +1\}$ be the indicator of commutation of two Pauli strings, i.e., $PQ=\epsilon_{P, Q}QP$. The essential observation is that measuring $\sigma\otimes\rho$ with a Bell state is the same as measuring $\rho$ with some element of the POVM $\{Q\sigma Q\}_{Q\in\cP_n}$ (see~\Cref{remark:Bell_measurement_single_copy}). It turns out that we can thus re-frame the optimal protocol above as measuring copies of $\rho$ using this POVM and averaging together $\epsilon_{P,Q}$ for the $Q$'s that are observed.

This now lets us extend to the $c$-copy case. Suppose for simplicity that instead of a distribution $\tau$ over POVMs, there exists a single POVM $M=\{F_s\}\in \cM$ such that $\chi_{M}^2(\rho_P^{\otimes c}\|\rho_m^{\otimes c})\ge \delta_{c, \cM}(A)$ for every $P\in A$. We conjugate each $F_s$ by all possible Paulis to produce an augmented POVM $\{Q F_s Q\}_{s,Q}$. We can then construct an unbiased estimator for $\Tr(P\rho)^c= \Tr(P^{\otimes c} \rho^{\otimes c})$ by averaging together commutation indicators $\epsilon_{P,Q}$ as in the $c = 1$ case. 

There are a few crucial subtleties that remain. For starters, depending on $P$ and the structure of the augmented POVM, it might be sub-optimal to construct an estimator for $\Tr(P^{\otimes c} \rho^{\otimes c})$ in this fashion rather than, say, $\Tr((P^{\otimes c-1}\otimes I) \rho^{\otimes c})$ or more generally $\Tr((P\otimes_S I)\rho^{\otimes c})$ for some subset $S\subseteq[c]$. 

Secondly, even putting this issue aside, it turns out that the naive estimator based on the commutation indicators does not quite have the right variance to serve as the optimal estimator. Instead, we need to define a careful \emph{reweighting} of the estimator to ensure that the term $\chi^2_M(\rho^{\otimes c}\| \rho^{\otimes c}_m)$ can appear in the variance (see Eq.~\eqref{eq:general_theorem_estimator}). 

Thirdly, if the optimal choice of $S$ above is even, then this only allows us to estimate $\Tr(P\rho)$ \emph{up to sign}. Fortuitously, the same logic of applying Bell basis measurements to $\sigma\otimes \rho$ that allowed us to handle the $c = 1$ case above can be leveraged here. First, it suffices to learn the signs of $\Tr(P\rho)$ for $P\in A$ such that $|\Tr(P\rho)|$ is not too small. We then search for a state $\sigma$ for which $|\Tr(P\sigma)|$ is close to $|\Tr(P\rho)|$ and thus not too small, after which we can apply Bell basis measurements to $\sigma\otimes \rho$ and safely divide out our estimates for $\Tr(P\sigma)\Tr(P\rho)$ by $\Tr(P\sigma)$ as before.

\paragraph{Optimal rates for Clifford measurements.}

Here we remark briefly on how we specialize our bounds to the case of protocols with Clifford measurements. Recall that when $\cM$ consists of all single-copy measurements, the sample complexity is $\widetilde{\Theta}(1/\epsilon^2\delta_A))$ (see \Cref{sec:PauliShadowNo} for the proof of \Cref{thm:PauliShadowNoMem}). Similar analysis shows that when $\cM$ consists of all $n$-qubit Clifford measurements, then $\delta_A$ gets replaced with the quantity
\begin{equation*}
    \delta_{A,\mathrm{Clifford}} \coloneqq \min_{\pi\in\cD(A)}\max_{\ket{\text{stab}}} \E_{P\sim \pi}\braket{\text{stab}|P|\text{stab}}^2\,,
\end{equation*}
where $\ket{\text{stab}}$ ranges over all stabilizer states. By definition, $\ket{\text{stab}}$ is a common eigenstate of a stabilizer group $S$ of size $2^n$, so $\ket{\text{stab}}$, and a basic property is that $\braket{\text{stab}|P|\text{stab}}^2$ is 1 if $P\in S$ and is 0 otherwise. Furthermore, every commuting set of Pauli strings is contained in a stabilizer group and equivalently can be regarded as an independent set in the anti-commutation graph $G_A$.

This equips $\delta_{A, \text{Clifford}}$ with a clear combinatorial interpretation: the min player assigns a distribution $\pi$ over $A$ and the max player chooses an independent set of $G_A$ which maximizes the probability that $P\sim \pi$ lands in the independent set. By writing this as a linear program and taking the dual, we find that $\delta_{A, \text{Clifford}}$ is precisely the inverse of the \emph{fractional coloring number} of $G_A$. See \Cref{sec:PauliShadowNoClifford} for the definition and the proof of \Cref{thm:Clifford} in detail.

\paragraph{Optimal \texorpdfstring{$\epsilon$}{eps} dependence.}
We turn to the proof of our next main result, Theorem~\ref{thm:Pauliepsdep}, which focuses on protocols that have sufficient memory to implement any $c$-copy measurements, with no restriction on quantum memory. For this, our master theorem specialized to $(c, \cM_{cn})$ protocols implies that it suffices to upper bound $\delta_{c, \cM_{cn}}(\cP_n)$. 

The intuition for the bound on this quantity is quite clean: the $2^n/\epsilon^2$ term in the final lower bound will come from the \emph{first-order} terms in the expansion of $\chi^2_M(\rho^{\otimes c}_P\|\rho^{\otimes c}_m)$, while the $1/\epsilon^4$ part will come from the \emph{higher-order} terms. This also naturally explains why single-copy protocols like classical shadows do not incur $1/\epsilon^4$ dependence.

As a first attempt, we can take a binomial expansion of $\chi^2_M(\rho^{\otimes c}_P\|\rho^{\otimes c}_m)$ in the variable $P$ and bound $\delta_{c,\cM_{cn}}(\cP_n)$ in terms of the quantities
\begin{equation*}
    a(S, \{F_s\}) = \E_P \E_s \frac{\Tr(F_s P^S)^2}{\Tr(F_s)^2}\,,
\end{equation*}
where $P^S$ is the $cn$-qubit Pauli operator which acts via $P$ on the sites indexed by $S$ and trivially elsewhere, and the expectations are over a uniformly random $P$ and $s$ coming from measuring the maximally mixed state with $\{F_s\}$. Concretely, we can bound $\delta_{c,\cM_{cn}}(\cP_n)$ by $2^{O(c)}\sum_{\emptyset\neq S\subseteq[c]} \epsilon^{2|S|} \max_{c\text{-copy POVMs} \ M} a(S, M)$.  

Obviously $a(S, M)\leq 1$ for any $S$, and furthermore, using the two-design property of the Clifford group, we can show that $a(S, M)\leq 1/2^n$ when $\abs{S}=1$ (see Parts (1) and (2) of~\Cref{lem:Bounded_memory_lower_bound}). Therefore $\delta_{c, \cM_{cn}}(\cP_n)\leq 2^{O(c)}(c\epsilon^2/2^n + 2^c\epsilon^4)$. This already implies a $\Omega(\min\{2^n/\epsilon^2, 1/\epsilon^4\})$ lower bound when $c=O(1)$. With a more refined argument, we can remove the $2^{O(c)}$ factors above and arrive at \Cref{thm:Pauliepsdep} and \Cref{thm:ccopyLower}. We defer the details to Section~\ref{sec:lowerboundck}.

\paragraph{Optimal memory-sample tradeoff.}
We now turn to $c$-copy protocols with $k$ qubits of quantum memory, where $0\leq k\leq n$. To show the lower bound of \Cref{thm:PauliShadowAllK}, we will apply the master theorem to the POVM set $\cM_{c,n}^k$ corresponding to the POVMs that arise in such protocols (see~\Cref{sec:Protocol_bounded_memory} for the formal definition). Similar to above, we will consider a binomial expansion of $\chi^2_M(\rho^{\otimes c}_P\|\rho^{\otimes c}_m)$, but we will use finer-grained information about the POVMs. The key property we will exploit about $\cM_{c,n}^k$ is that every POVM element is a \emph{matrix product state} (see-\Cref{def:MPS}) with low bond dimension. This will allow us to refine the naive upper bound of $a(S,M) \le 1$ above to an upper bound of $a(S, M)\leq 1/2^{n-k}$ for $\abs{S}\ge 2$ and $M\in \cM_{c,n}^k$. This estimate is the main workhorse behind our proof. The proof is rather technical so we defer the details to Part (3) of~\Cref{lem:Bounded_memory_lower_bound}. Following the same argument as in the proof of \Cref{thm:Pauliepsdep}, we can obtain the lower bound of \Cref{thm:PauliShadowAllK}.

We now describe our proof of the matching upper bound. To our knowledge this is the first protocol to exploit a small amount of quantum memory to get a nontrivial sample complexity for Pauli shadow tomography. First note that the $O(n2^{n}/\epsilon^2)$ term in the upper bound is just the sample complexity when the quantum memory is not used. Here we sketch the protocol achieving sample complexity $O(n2^{n-k}/\epsilon^4)$. Recall that when $k=n$, we can apply two-copy Bell measurements to estimate all Pauli strings simultaneously. A natural intuition for extending this is to apply a Bell measurement on a \emph{subsystem} of $2$ copies of $\rho$. For instance, given two copies of unknown state $\rho$, we divide the $n$ qubits of each $\rho$ into two disjoint subsystems containing $k$ and $(n-k)$ qubits. We measure these two subsystems differently: For the first subsystem, we jointly measure the first $k$ qubits of two copies of $\rho$ using a Bell measurement. For the second subsystem, we apply a stabilizer basis measurement on the remaining $(n-k)$ qubits for some choice of stabilizer group $S$.

\paragraph{Purity testing.}
Finally, we describe an interesting further application of our techniques, to obtaining an optimal memory-sample tradeoff for purity testing~\Cref{thm:Purity}. Here we sketch the proof.

We first briefly describe the upper bound. The $O(2^{n/2})$ term just comes from the memory-free protocol of~\cite{chen2022exponential}. The $O(2^{n-k})$ term comes from a partial swap test. Specifically, given two copies of $\rho$, we perform computational basis measurement on the first $n-k$ qubits of two copies respectively, and perform swap test (i.e., the measurement $\{F_1=\frac{I+\SWAP}{2}, F_2=\frac{I-SWAP}{2}\}$) on the remaining $2k$ qubits, which can be done using $k$ qubits of quantum memory. Denote the outcomes by $x_1,x_2\in \{0, 1\}^{n-k}$ and $y\in \{1, 2\}$. We consider the event ``$x_1=x_2,y=2$''. The event will not happen when $\rho$ is pure, but it will happen with probability $\Theta(2^{k-n})$ when $\rho$ is maximally mixed (see \Cref{sec:PurityUpper}). Therefore we can distinguish the two cases using $O(2^{n-k})$ samples.

For the lower bound, unfortunately we cannot directly use the martingale analysis that was prevalent in our Pauli shadow tomography proofs. The reason is that in any learning tree for purity testing, the chi-squared divergence in a single step can actually be as large as $\Omega(1)$. This was previously observed in~\cite{chen2022exponential} and motivated their development of a so-called ``path-based analysis'' which allowed them to directly prove that the likelihood ratio between observing a particular path in the tree when $\rho$ is pure versus when it is maximally mixed is not too small.

Unfortunately, this path-based analysis is heavily tailored to the memoryless setting, and it was unclear how to get any handle on the problem in the presence of a small amount of quantum memory. Our key idea is to \emph{engineer a martingale-like analysis} out of tools from the path-based analysis.

Suppose that one encounters POVM elements $F_1,\ldots,F_t$ along some path of the tree. The expected likelihood ratio between seeing this path when $\rho$ is Haar-random pure versus when $\rho$ is maximally mixed can be explicitly computed and lower bounded by 
\begin{equation}
    \Bigl(1- \frac{4T^2}{2^n}\Bigr)\cdot \Tr\Bigl(\bigotimes_{t=1}^T F_t ~S_{2T}\Bigr)/\prod_{t=1}^T\Tr(F_t)\,, \label{eq:haar}
\end{equation}
where $S_{2T}$ is the sum of all permutation operators on $2T$ $2^n$-dimensional qudits. The factor $(1-4T^2/2^n)$ comes from Haar integration and is close to 1 when $T<O(2^{n/2})$.

We deal with Eq.~\eqref{eq:haar} in two steps. First, in~\Cref{lem:permutation} we show that
\begin{equation}
    \Tr(\bigotimes_{t=1}^T F_t ~S_{2T})\ge \prod_{t=1}^T \Tr(F_t S_2)\,. \label{eq:keyperm}
\end{equation} 
This can be thought of as a generalization of the diagrammatic calculation of \cite[Lemma 5.12]{chen2022exponential}. Whereas \cite{chen2022exponential} uses a version of this estimate to bound the likelihood ratio along the path in one shot, we have to be more careful to account for the quantum memory. The intuition for what Eq.~\eqref{eq:keyperm} buys is that it lets us consider the quantity $\prod^T_{t=1} \Tr(F_t S_2)/\Tr(F_t)$ instead. 

This brings us to the second step of the proof. The crucial conceptual twist is that $\prod_{t=1}^T\tr(F_t S_2)/\tr(F_t)$ can be interpreted roughly as the likelihood ratio for a new distinguishing problem on $2n$-qubit states, namely $\rho_1=(I_{2n}+\SWAP)/2^n(2^n+1)$ and $\rho_0=I_{2n}/2^{2n}$, using measurements in $\cM_{2, n}^k$. It turns out that this new likelihood ratio is actually amenable to the martingale approach. The main estimate we prove (see \Cref{lem:permutation_SWAP} for details) leverages the low-rank structure in the POVM elements in $\cM_{2, n}^k$ to bound the chi-squared divergence for a single measurement in this case by $2^{k-n}$. So the likelihood ratio for this new distinguishing task is not too small for most paths when $T<2^{n-k}$. This, together with the $1-4T^2/2^n$ factor in Eq.~\eqref{eq:haar}, gives the lower bound of $\Omega(\min\{2^{n/2},2^{n-k}\})$.

\section*{Future directions}

\paragraph{Role of adaptivity.} The protocol we proposed in this paper for learning the signs of $\tr(P\rho)$ employs an adaptive procedure as the measurement in each round depends on the previous history. In contrast, there are no sample-efficient non-adaptive algorithms known even using multi-copy measurements. It is thus natural to ask if this is inherent or whether we can find a non-adaptive alternative. In an upcoming manuscript, we provide a partial negative answer to this question by showing that at least for $2$-copy measurements, an exponential number of samples is necessary for any nonadaptive protocol to learn the signs even when the absolute values are known.

\paragraph{Larger $c$.} Our lower bounds do not apply beyond $c=O(\poly(n))$, so we are not able to say anything about \emph{general} protocols with $k\leq n$ qubits of quantum memory. It is interesting to see to what extent the landscape changes if, e.g., $c$ is allowed to scale with $\epsilon$.

\paragraph{Hard subsets of Pauli strings.} What are subsets $A$ for which the sample complexity as characterized by the various $\delta$ quantities in this paper especially large? We conjecture that the ``hardest'' such subsets are given by randomly chosen Pauli strings.

\paragraph{Computational bounds?} Our upper bounds are purely information-theoretic in nature, and it remains an outstanding open question to identify interesting settings where \emph{computationally efficient} Pauli shadow tomography is possible. Even more tantalizing, it would be interesting to identify \emph{computational-statistical gaps} in the form of subsets $A$ for which the optimal sample complexity cannot be achieved by any polynomial-time algorithm.

\section{Related work}\label{sec:RelatedWorks}
Our work is part of a larger body of recent results exploring how near-term constraints on quantum devices for various quantum learning tasks affect the underlying statistical complexity, see e.g.~\cite{bubeck2020entanglement,chen2022toward,chen2022tight,chen2022tight2,chen2022exponential,chen2021hierarchy,aharonov2022quantum,fawzi2023quantum,huang2021information,huang2022quantum,fawzi2023lower,liu2023memory,chen2023efficientPauli}. A review of this literature is beyond the scope of this work and the reader can refer to the survey~\cite{anshu2023survey} for a more detailed overview. For conciseness, we only review the most relevant lines of work here.

\paragraph{Learning properties of quantum states. }
The most canonical task for learning quantum states is to completely learn the density matrix to high accuracy, usually in trace norm or fidelity, which is referred to as quantum state tomography~\cite{hradil1997quantum,gross2010quantum,blume2010optimal,banaszek2013focus}. Unfortunately, it is well-known that the sample complexity for this task is exponential in the system size $n$~\cite{haah2016sample,odonnell2016efficient,chen2022tight}.

A well-studied alternative that circumvents this exponential barrier is shadow tomography~\cite{aaronson2018shadow}, in which one is only required to estimate the expectation values of a given set of $M$ observables or measurements. A line of works~\cite{aaronson2018shadow,aaronson2018online,brandao2019quantum,aaronson2019gentle,buadescu2021improved,gong2023learning} gives algorithms for this task that use $\poly(\log M,n,1/\epsilon)$ copies of samples. These algorithms heavily rely on joint measurement samples and thus require a large amount of quantum memory, rendering them impractical to implement on near-term devices. \cite{huangPredictingManyProperties2020} gave an algorithm, classical shadows, that only uses single-copy measurements and uses $O(2^n\log M/\epsilon^2)$ copies of $\rho$, which is known to be tight in the single-copy setting~\cite{chen2022exponential}; roughly speaking, the algorithm uses random basis measurements to produce a classical, unbiased estimator of the state which can then be used to predict arbitrary observables. 

We have already discussed prior theoretical work on Pauli shadow tomography and mention here the plethora of works on applications of this to finding ground states, e.g. in quantum chemistry~\cite{mcclean2016theory,kandala2017hardware,paini2019approximate,gokhale2019minimizing,huggins2021efficient,cotler2020quantum,crawford2021efficient,huang2021efficient,huang2022provably}.

\paragraph{Quantum memory trade-offs for learning.}

Understanding the quantum memory trade-offs for learning without quantum memory and with bounded quantum memory has been the subject of a long line of recent work. In~\cite{bubeck2020entanglement}, a polynomial gap for mixedness testing with and without quantum memory was demonstrated; see also the follow-up works of~\cite{chen2022toward,chen2022tight2}. Subsequent works~\cite{huangPredictingManyProperties2020,aharonov2022quantum,chen2022exponential,anshu2022distributed} demonstrated exponential separations without and with enough quantum memory for shadow tomography, unitary learning, purity testing, distributed inner product estimation, and quantum dynamics learning. Similar separations have also been developed for the task of learning Pauli channels~\cite{chen2022quantum}, where an even sharper exponential separation was discovered using $O(\log n)$ qubits of quantum memory~\cite{chen2023efficientPauli}. Recently, a classical learning problem that enjoys a quantum memory trade-off was proposed~\cite{liu2023memory}. Finally, we remark that memory trade-offs have also been demonstrated for classical learning
problems~\cite{steinhardt2016memory,raz2018fast,raz2017time,kol2017time,moshkovitz2017mixing,moshkovitz2018entropy,sharan2019memory}. It would be interesting to establish potential connections between these two lines of research.

\paragraph{Other applications of Pauli analysis.}

Properties of the Pauli spectrum $\{\Tr(P\rho)\}_{P\in\cP_n}$ also figure heavily among various proposals in resource theory for quantifying how ``magic'' or ``non-stabilizer'' a quantum state is, e.g. the stabilizer rank~\cite{qassim2021improved}, stabilizer nullity~\cite{beverland2020lower}, and various notions of stabilizer entropy~\cite{leone2022stabilizer,haug2023efficient,turkeshi2023pauli}. One such notion~\cite{haug2023scalable} was used to quantify entanglement properties of fast scrambling circuits in the recent neutral atoms breakthrough~\cite{bluvstein2024logical}.

In quantum learning theory, Bell sampling, a simple but powerful algorithm for estimating Pauli observables, has been used to give new algorithms for various tasks involving learning from or testing stabilizer or near-stabilizer states~\cite{grewal2023improved,grewal2022low,grewal2023efficient}. Additionally, various recent positive results in quantum process tomography~\cite{huang2022learning,nadimpalli2023pauli} effectively reduced to Pauli shadow tomography by arguing that the relevant operators were well-approximated by the low-degree truncation of their Pauli decomposition.

As mentioned in the introduction, the Pauli shadow tomography first arose with the variational quantum eigensolver as one is required to estimate the expectation value $\tr(H\rho)$ for a given Hamiltonian $H$ as a linear combination of Pauli observables and a state $\rho$. Pauli shadow tomography itself, however, can also be applied to an inverse direction: one can recover the Hamiltonian assumed to be a linear combination of Pauli observables using the expectation values $\tr(P\rho)$ for Pauli observable of some unknown eigenstate $\rho$ of the Hamiltonian~\cite{qi2019determining,bentsen2019integrable}.

\paragraph{Concurrent work.} During the preparation of this manuscript, we were made aware of the independent and concurrent work of Ref.~\cite{King2024Triply}, which also studies the complexity of Pauli shadow tomography. They independently proved the following results present in this work: (i) \Cref{thm:learn_sign_from_absolute} about the algorithm for estimating all $\Tr(P\rho)$ using $O(n/\epsilon^4)$ two-copy measurements, and (ii) \Cref{thm:Clifford} about the algorithm for Pauli shadow tomography using single-copy Clifford measurements. Apart from these overlaps however, their results focus primarily on the \emph{time} complexity of protocols for Pauli shadow tomography, as well as applications to estimating fermionic observables, which are unique to their paper. 

Their results have interesting implications on the time complexity for achieving the sample complexity bounds in our work: (i) they show that the sample complexity in \Cref{thm:learn_sign_from_absolute} can be achieved within $\poly(2^n,1/\epsilon)$ running time using matrix multiplicative weights~\cite{arora2007combinatorial,aaronson2018online}, and (ii) they show that the running time is $O((T+n^3)\zeta_f(G_A)\log(\abs{A})/\epsilon^2))$ for our algorithm using Clifford measurements where $T$ is the classical runtime for sampling the fractional coloring. 

The other results in this paper are specific to our work, namely the Master \Cref{thm:general_theorem} in \Cref{sec:GeneralThm}, the optimal algorithm for estimating an arbitrary subset $A\subseteq\cP_n$ using single-copy measurements (including the general \Cref{thm:PauliShadowNoMem} in \Cref{sec:PauliShadowNo} and its applications to different sets in \Cref{sec:ApplicationNo}), the optimal algorithms and lower bounds with $k$ qubits of quantum memory in \Cref{sec:PauliShadowBoundedMemory}, all the lower bounds including the lower bound for the master theorem, ancilla-free protocols (with or without restricting to Clifford measurements), protocols with $k$ qubits of quantum memory, and the $\Omega(1/\epsilon^4)$ lower bound regarding $\epsilon$ dependence for $\poly(n)$-copy protocols, and our results on purity testing in \Cref{sec:Purity}.

\section{Preliminaries}\label{sec:Prelim}

In this section, we collect the basic concepts and results required throughout this paper. We use $\norm{B}$ to represent the operator norm for matrix $B$, and use $\norm{v}_p$ to denote the $L_p$ norm of vector $v$. We also use $\tilde{O}$ and $\tilde{\Theta}$ to hide the poly-logarithmic dependence in big-O notations.
We will abbreviate $\{1, 2, \cdots, c\}$ as $[c]$, $(x_a, x_{a+1}, \cdots, x_{b})$ as $x_{[a:b]}$, and ``with probability'' as ``w.p.''. When we say ``with high probability'' without specification, we mean with probability at least $2/3$.

\subsection{Basic results in quantum information}\label{sec:Prelim_QI}
We start with some standard definitions and calculations in quantum information. A general $n$-qubit quantum state can be represented as a positive semi-definite matrix $\rho\in\mathbb{C}^{2^n\times 2^n}$ with $\tr(\rho)=1$. When the state is rank-1 and thus $\tr(\rho^2)=1$, it is a pure state and is usually represented as $\ket{\psi}$, $\ket{\phi}$, or $\ket{\varphi}$ throughout this paper. The Bell basis $\{\ket{\psi^\pm}=(\ket{10}\pm\ket{01})/\sqrt{2},\ket{\phi}=(\ket{00}\pm\ket{11})/\sqrt{2}\}$ is a set of maximally entangled two-qubit states. Given a (possibly unnormalized) $n$-qubit quantum state $\rho$ and a set $S\subseteq [n]$, we use $\tr_{S}(\rho)$ to denote the remaining state after tracing out the qubits in $S$. Denote the $n$-qubit Hilbert space by $\mathbb{H}_n$ and the $n$-qubit identity operator by $I_n$. Throughout this paper, $n$ is the number of qubits in the unknown state we aim to learn. We will call an $n$-qubit quantum state a \emph{qudit} for short.

\paragraph{Pauli strings.} Define $\cP_n\coloneqq \{I, X, Y, Z\}^{\otimes n}$ to be the set of $n$-qubit \emph{Pauli strings}, where 
\begin{equation*}
    I=\begin{pmatrix}1 & 0\\0& 1\end{pmatrix}\,, \qquad X=\begin{pmatrix}0 & 1\\1& 0\end{pmatrix}\,, \qquad Y=\begin{pmatrix}0 & -i\\i& 0\end{pmatrix}\,, \qquad Z=\begin{pmatrix}1 & 0\\0& -1\end{pmatrix}
\end{equation*} are single-qubit Pauli operators. We call $\cP_n$ the \emph{Pauli group} although we do not use the group structure in this paper.

For any two Pauli strings $P, Q\in \cP_n$, let $\epsilon_{P, Q}\in \{-1, 1\}$ denote the indicator for whether $P$ and $Q$ commute, so that 
\begin{equation*}
    PQ=\epsilon_{P, Q}QP\,.
\end{equation*}

We will need the following lemmas for the sum of products of pairs of Pauli matrices:
\begin{lemma}[E.g., Lemma 4.10 in~\cite{chen2022exponential}] \label{lem:SumPauli}
    We have
    \begin{align*}
        \sum_{P\in\cP_ n}P\otimes P=2^n\,\mathrm{SWAP}_n\,,    
    \end{align*} 
    where $\mathrm{SWAP}_n$ is the $2n$-qubit SWAP operator.
\end{lemma}

\begin{lemma}\label{lem:SumPauliExp2}
    For every $2^n\times 2^n$ matrix $B$, we have
    \begin{align*}
        \sum_{P\in \cP_n}PBP = 2^n \tr(B) I.
    \end{align*}
\end{lemma}

\paragraph{Quantum measurements.}General quantum measurements are represented as positive operator-valued measures (POVMs). An $n$-qubit POVM is defined by a set of positive-semidefinite matrices $\{F_s\}_s$ satisfying $\sum_s F_s=I$. Here, each $F_s$ is a \emph{POVM element} corresponding to \emph{measurement outcome} $s$. If we measure a given quantum state $\rho$ using the POVM $\{F_s\}_s$, the probability of observing outcome $s$ is $\Tr(F_s\rho)$. 

When $F_s$ is rank-1 for all $s$, we call $\{F_s\}_s$ a \emph{rank-1 POVM}.~\cite[Lemma 4.8]{chen2022exponential} shows that information-theoretically, any POVM can be simulated by a rank-$1$ POVM and classical post-processing. In the sequel, we will thus assume without loss of generality when we prove our lower bounds that all measurements in the learning protocols we consider only use rank-$1$ POVMs. A useful parametrization for rank-1 POVMs is via
\begin{align*}
    \{w_s2^n\ket{\psi_s}\bra{\psi_s}\}\,,
\end{align*}
for pure states $\{\ket{\psi_s}\}$ and nonnegative weights $w_s$ with $\sum_s w_s = 1$.

Sometimes we also want to know the post-measurement quantum states. However, the post-measurement state is not determined by the POVM $\{F_s\}_s$ itself, but by the physical realization of it. To specify the post-measurement state, we need to specify a Cholesky decomposition for every POVM element $\{F_s=M_s^\dagger M_s\}_s$, where $M_s$ is a $2^m\times 2^n$ matrix. Then the post-measurement state upon outcome $s$ is an $m$-qubit state given by
\begin{align*}
    \rho\to\frac{M_s\rho M_s^\dagger}{\Tr(M_s^\dagger M_s\rho)}.
\end{align*}
To emphasize the dimension of the post-measurement state, we will refer to such a (realization of a) POVM as having \emph{$n$-qubit input and $m$-qubit output}, or simply as being an \emph{$(n\to m)$-POVM}. Unless we explicitly specify the number of qubits in the output, when we simply say ``POVM'' in this paper, we always mean a realization with $0$-qubit output.

\paragraph{POVM sets. }To avoid mathematical complications, in this paper, we assume all POVMs have a bounded number of outcomes, say, at most $N_P$ outcomes. Here $N_P$ can be an arbitrarily large number and our bounds will not depend on this number in any way, so this assumption will not affect the generality of our results but simply allow us to avoid unnecessary topological minutiae.

Let $\cM_n$ be the set of $n$-qubit POVMs under this definition. Then $\cM_n=\{(F_1, \cdots, F_{N_P})\}$ is a bounded closed subset of a finite-dimensional Euclidean space and is thus compact. 
We will assume that all sets of POVMs that we consider are closed, and thus compact as well. This compactness is necessary to apply the minimax theorem and to ensure the existence of optimal POVMs in Eq.~\eqref{eq:general_delta}.

We say a $cn$-qubit POVM set $\cM$ is \emph{Pauli-closed} if for every $M=\{F_s\}_s\in \cM$ and every $cn$-qubit Pauli string $P$, the POVM $\{PF_sP\}_s$ is also in $\cM$. We say $\cM$ \emph{contains all single-copy measurements}, if for every $n$-qubit POVM $\{E_s\}_s$, the POVM $\{E_s\otimes (I_{n})^{\otimes (c-1)}\}_s$ is in $\cM$.

\paragraph{Permutation operators and Haar measure. }For $T>0$, let $S_T$ be the symmetric group on $T$ elements. Every permutation $\pi\in S_T$ can be extended to a linear quantum operator acting on $T$ qudits by the following
\begin{equation*}
    \pi:(\mathbb{H}_n)^{\otimes T}\to(\mathbb{H}_n)^{\otimes T},\quad \pi\ket{\psi_1}\otimes\cdots\otimes\ket{\psi_T}=\ket{\psi_{\pi^{-1}(1)}}\otimes\cdots\otimes\ket{\psi_{\pi^{-1}(T)}}.
\end{equation*}
The $\pi$ is called a \emph{permutation operator}. With a little abuse of notation, we also regard $S_T$ as an operator $S_T\coloneqq \sum_{\pi\in S_T}\pi$. 
\begin{lemma}[See, e.g., \cite{meleIntroductionHaarMeasure2023}]\label{lem:permutation_operator}
    $\tr(S_T)=T(T+1)\cdots (T+2^n-1)$. Moreover, the operator $S_T/T!$ is a projector to the symmetric subspace of dimension $\binom{T+2^n-1}{T}$, thus $S_T\ge 0$.
\end{lemma}
\begin{lemma}[See, e.g., \cite{meleIntroductionHaarMeasure2023}]\label{lem:haar_integral}
    Let $\mu_H$ be the uniform measure (i.e., Haar measure) over $n$-qubit pure states. Then
    \begin{equation*}
        \E_{\ket{\psi}\sim \mu_H}[\ketbra{\psi}^{\otimes T}]=\frac{1}{2^n(2^n+1)\cdots (2^n+T-1)}S_T.
    \end{equation*}
\end{lemma}

\subsection{Learning with and without quantum memory}\label{sec:Prelim_Model}
\begin{figure}[t]
    \centering
    \begin{minipage}[c]{0.35\textwidth}
        \centering
        \begin{minipage}[t]{\textwidth}
            \captionsetup{skip=3pt}
            \centering
            \begin{quantikz}
                \gate[1]{\rho}\gategroup[1,steps=2]{Repeat for $T$ rounds}&\meter{}
            \end{quantikz}
            \vspace{0pt}
            \caption*{(a) Learning without quantum memory}
        \end{minipage}

        \vspace{0.1cm}

        \begin{minipage}[t]{\textwidth}
            \captionsetup{skip=3pt}
            \centering
            \begin{quantikz}
                \gate[1]{\rho^{\otimes c}}\gategroup[1,steps=3]{Repeat for T rounds}&\qwbundle{c}&\meter{}
            \end{quantikz}
            \vspace{0pt}
            \caption*{(c) Learning with $c$-copy measurements}
        \end{minipage}

        \vspace{0.1cm}

        \begin{minipage}[t]{\textwidth}
            \captionsetup{skip=3pt}
            \centering
            \begin{quantikz}
                \gate[1]{\rho^{\otimes c}}\gategroup[1,steps=3,style={inner sep=8pt}]{Repeat for T rounds}&\qwbundle{c}&\meter{M\in \cM}
            \end{quantikz}
            \vspace{0pt}
            \caption*{(e) $(c, \cM)$ learning protocols}
        \end{minipage}
    \end{minipage}\hfill
    \begin{minipage}[c]{0.6\textwidth}
        \begin{minipage}[t]{\textwidth}
            \captionsetup{skip=3pt}
            \centering
            \begin{quantikz}
                \gategroup[2,steps=9]{Depth $T$, No repetitions}&\meter[2]{}\setwiretype{n}&\setwiretype{c}&\meter[2]{}&&\meter[2]{}&\ \ldots\ &&\meter[2]{}\\
                \gate[1]{\rho}&&\gate[1]{\rho}\setwiretype{n}&\setwiretype{q}&\gate[1]{\rho}\setwiretype{n}&\setwiretype{q}&\ \ldots\ &\gate[1]{\rho}\setwiretype{n}&\setwiretype{q}
            \end{quantikz}
            \vspace{0pt}
            \caption*{(b) Learning with $k$ qubits of quantum memory}
        \end{minipage}

        \vspace{0.2cm}

        \begin{minipage}[t]{\textwidth}
            \captionsetup{skip=3pt}
            \centering
            \begin{quantikz}
                \gategroup[2,steps=6]{Depth $c$ in each round ($c=3$ here), Repeat for $T$ rounds}&\meter[2]{}\setwiretype{n}&\setwiretype{c}&\meter[2]{}&&\meter[2]{}\\
                \gate[1]{\rho}&&\gate[1]{\rho}\setwiretype{n}&\setwiretype{q}&\gate[1]{\rho}\setwiretype{n}&\setwiretype{q}
            \end{quantikz}
            \vspace{0pt}
            \caption*{(d) Learning with $c$-copy measurements and $k$ qubits of quantum memory}
        \end{minipage}
    \end{minipage}
    \caption{Illustration of five classes of learning protocols. In all diagrams, ``\begin{quantikz} &\end{quantikz}'' means a $n$-qubit wire and ``\begin{quantikz} \setwiretype{c} &\end{quantikz}'' means a $k$-qubit wire ($k$ is the size of quantum memory). Measurements in (a)-(d) are arbitrary POVMs whose input and output dimensions should be clear from the wires. All POVMs are adaptively chosen depending on previous measurement outcomes.}
    \label{fig:learning_protocols}
\end{figure}

We now consider the settings for learning the expectation values, or more generally, the properties of a given quantum state $\rho$ given access to copies thereof. 

In the following, we give formal definitions for learning algorithms that either have no quantum memory or have some bounded quantum memory and/or can make multi-copy measurements. We then introduce $(c, \cM)$ learning protocols, which include all of these classes of protocols. \Cref{fig:learning_protocols} provides an illustration of this taxonomy.

\begin{definition}[Learning without quantum memory (ancilla-free), see \Cref{fig:learning_protocols}(a)]\label{def:ModelNoMem}
The algorithm is allowed to perform arbitrary POVM measurements on one copy of the unknown quantum state $\rho$ at a time. In particular, in each round of the algorithm, it gets a new copy of $\rho$, selects a POVM $\{F_s\}_s$ and measures with it to obtain some classical outcome $s$ with probability $\Tr(F_s \rho)$. The choice of POVM can depend on all previous measurement outcomes. After $T$ measurements, the algorithm predicts the desired properties of $\rho$. The sample complexity of the protocol is $T$.
\end{definition}

\noindent We next define the setting where the algorithm has a bounded amount of quantum memory. Intuitively, the algorithm has an extra $k$-qubit quantum memory register $\sigma$. Every time the algorithm accesses a fresh copy of $\rho$, it performs a joint measurement on the combined state $\sigma\otimes\rho$, collects the classical outcome, and stores the $k$-qubit output to the quantum memory. The formal definition is given below:

\begin{definition}[Learning with $k$ qubits of quantum memory ($k$ ancilla qubits), see \Cref{fig:learning_protocols}(b)]\label{def:ModelKMem}
The algorithm maintains a $k$-qubit quantum memory $\sigma$ (empty at first). In each round of the algorithm, it gets a new copy of $\rho$, selects an $(n+k\to k)$ POVM $\{M_s^\dagger M_s\}_s$ to measure the joint state $\sigma\otimes\rho$, obtains the classical outcome $s$ with probability $\tr((\sigma\otimes \rho) M_s^\dagger M_s)$, and stores the $k$-qubit post-measurement state $M_(\sigma\otimes \rho)M_s^\dagger/\tr((\sigma\otimes \rho) M_s^\dagger M_s)$ to the quantum memory.
The first POVM is $(n\to k)$ since the quantum memory is empty at first. The last POVM is $(n+k\to 0)$ since it does not need quantum memory anymore.
After $T$ measurements, the algorithm predicts the desired properties of $\rho$. The sample complexity of the protocol is $T$.
\end{definition}
In this definition, the quantum memory is empty at first. This is slightly different from the description in \Cref{sec:results_with_quantum_memory}, where the quantum memory is initialized to the maximally mixed state. We remark that the initial memory does not matter because the set of $(n\to k)$ POVMs includes the power of preparation of ancilla qubits. To be concrete, suppose the ancilla state is a $k$-qubit pure state $\ketbra{\psi}$ (mixed ancilla state is just a classical probability mixture of pure states), then the $(n+k\to k)$ POVM $\{M_s^\dagger M_s\}_s$ on $\ketbra{\psi}\otimes \rho$ can be simulated by the $(n\to k)$ POVM $\{(\bra{\psi}\otimes I_n)M_s^{\dagger}M_s(\ket{\psi}\otimes I_n)\}_s$ on $\rho$.

\noindent An alternative model for learning to predict the properties of an unknown quantum state $\rho$ is via $c$-copy measurements. Within this model, the algorithm gets a batch of fresh copies $\rho^{\otimes c}$, performs a joint measurement on $\rho^{\otimes c}$, collects the classical data, and repeats. The number of copies used is $cT$ after $T$ rounds. This model is strictly weaker than learning with $(c-1)n$ qubits of quantum memory as we can always prepare the quantum memory using swap gates and jointly measure $c$ copies of $\rho$.

\begin{definition}[Learning with $c$-copy measurements, see \Cref{fig:learning_protocols}(c)]\label{def:ModelCCopy}
The algorithm is allowed to perform arbitrary POVM measurements on $c$ copies of $\rho$ at a time. In particular, in each round of the algorithm, it gets a new copy of $\rho^{\otimes c}$, selects a $cn$-qubit POVM $\{F_s\}_s$ and measures with it to obtain some classical outcome $s$ with probability $\Tr(F_s \rho^{\otimes c})$. The choice of POVM can depend on all previous measurement outcomes. After $T$ measurements, the algorithm predicts the desired properties of $\rho$. The sample complexity of the protocol is $cT$.
\end{definition}

\noindent We now define the model of learning with $c$-copy measurements and $k$ qubits of quantum memory.
This is a combination of the constraints in~\Cref{def:ModelKMem} and \Cref{def:ModelCCopy} and thus protocols in this class are weaker than protocols in either of the two preceding classes.

Informally, protocols in this class can be thought of as protocols with $k$ qubits of memory, but the quantum memory is reset after every $c$ oracle accesses. While this is technically a restriction,  for the ranges of $c$ that we consider we argue that this model already captures almost all existing techniques for designing state learning protocols. Also note that when $k=(c-1)n$, this model is equivalent to learning with $c$-copy measurements. In particular, learning with two-copy measurements is equivalent to learning with 2-copy measurements and $n$ qubits of quantum memory.

\begin{definition}[Learning with $c$-copy measurements and $k$ qubits of quantum memory, see \Cref{fig:learning_protocols}(d)]\label{def:ModelCCopyKMem}
The algorithm consists of $T$ rounds and each round involves $c$ new copies of $\rho$. In each round, it maintains a $k$-qubit quantum memory $\sigma$ (empty at first). Within the round, in each step it gets a new copy of $\rho$, selects an $(n+k\to k)$ POVM $\{M_s^\dagger M_s\}_s$ to measure the joint state $\sigma\otimes\rho$, obtains the classical outcome $s$ with probability $\tr((\sigma\otimes \rho) M_s^\dagger M_s)$, and stores the $k$-qubit post-measurement state $M_s(\sigma\otimes \rho)M_s^\dagger/\tr((\sigma\otimes \rho)M_s^\dagger M_s)$ to the quantum memory. In the first step of any round, the POVM is $(n\to k)$ and, in the $c$-th step it is $(n+k\to 0)$. After $c$ oracle accesses, the quantum memory is reset to empty, and the next round begins.
After $T$ rounds, the algorithm predicts the desired properties of $\rho$. The sample complexity of the protocol is $cT$.
\end{definition}

\noindent Finally, we introduce $(c, \cM)$ learning protocols to unify all classes of learning protocols above.

\begin{definition}[$(c, \cM)$ learning protocols, see \Cref{fig:learning_protocols}(e)]\label{def:Model_c_M}
    For a positive integer $c$ and a closed set of $cn$-qubit POVMs $\cM$, a $(c, \cM)$ learning protocol is an algorithm that learns with $c$-copy measurements, with the additional restriction that the measurements must all be from the POVM set $\cM$.
\end{definition}

\noindent For completeness, let us verify that this can capture the first four models from this section. Let $\cM_n$ be the set of all POVMs on $n$ qubits. By definition, \Cref{def:ModelNoMem} is equivalent to the class of $(1, \cM_n)$ learning protocols. \Cref{def:ModelCCopy} is equivalent to the class of $(c, \cM_{cn})$ learning protocols. \Cref{def:ModelCCopyKMem} describes a restricted set of POVMs on $cn$ qubits, so it is equivalent to $(c, \cM_{c, n}^k)$ protocols for some $\cM_{c, n}^k\subseteq \cM_{cn}$ \--- we will define $\cM_{c, n}^k$ in greater detail in \Cref{sec:Protocol_bounded_memory}. Lastly, a learning algorithm with $k$ qubits of quantum memory in \Cref{def:ModelKMem} that uses $T$ samples is a $(T, \cM_{T, n}^k)$ learning protocol with only one round. Therefore, the class of $(c, \cM)$ is general enough to illustrate the role of quantum memory. On the other hand, this definition makes the analysis cleaner as we don't need to consider the role of quantum memory explicitly in some of our proofs. 

\subsection{Tree representations and the Le Cam's method}
\label{sec:LeCam}

We use tree representations and Le Cam's method when we prove the lower bounds. We adapt the learning tree formalism of~\cite{bubeck2020entanglement,aharonov1997fault,huang2021information,chen2022exponential} to the Pauli shadow tomography in \Cref{prob:PauliShadow}. Since all mentioned classes of learning protocols are special cases of $(c, \cM)$ learning protocols, we only give the tree representation for $(c, \cM)$ learning protocols here. 

\begin{definition}[Tree representation for $(c, \cM)$ learning protocols]\label{def:TreeRepNonCM}
    Given an unknown $n$-qubit quantum state $\rho$, a $(c, \cM)$ learning protocol can be represented as a rooted tree $\cT$ of depth $T$ with each node on the tree recording the measurement outcomes of the algorithm so far. In particular, it has the following properties:
    \begin{itemize}
        \item We assign a probability $p^\rho(u)$ for any node $u$ on the tree $\cT$.
        \item The probability assigned to the root $r$ of the tree is $p^\rho(r)=1$.
        \item At each non-leaf node $u$, we measure $\rho^{\otimes c}$ using an adaptively chosen POVM $M_u=\{F_{s}^u\}_s\in \cM$, which results in a classical outcome $s$. Each child node $v$ corresponding to the classical outcome $s$ of the node $u$ is connected through the edge $e_{u,s}$.
        \item If $v$ is the child of $u$ through the edge $e_{u,s}$, the probability to traverse this edge is $p^{\rho}(s|u)= \tr(F_s^u\rho^{\otimes c})$. Then
        \begin{align*}
        p^\rho(v)=p^\rho(u)p^\rho(s|u)=p^\rho(u)\tr(F_s^u \rho^{\otimes c}).
        \end{align*}
        \item Each root-to-leaf path is of length $T$. Note that for a leaf node $\ell$, $p^\rho(\ell)$ is the probability for the classical memory to be in state $\ell$ at the end of the learning protocol. Denote the set of leaves of $\cT$ by $\mathrm{leaf}(\cT)$
    \end{itemize}
\end{definition}

\noindent When proving the lower bounds, we will consider a reduction from the following two-hypothesis distinguishing problem.
\begin{problem}[Many-versus-one distinguishing problem between the maximally mixed state $\rho_m$ and a ensemble $\{\rho_a \text{ w.p. } p_a\}$]\label{prob:Distinguish}
    We are given access to copies of an unknown quantum state $\rho$. And we need to distinguish the following two cases
    \begin{itemize}
        \item $\rho$ is the maximal mixed state $\rho_m=I/2^n$.
        \item $\rho$ is sampled from $\{\rho_a\}$ according to the probability distribution $\{p_a\}$.
    \end{itemize}
\end{problem}
\noindent In the Pauli shadow tomography task for Pauli string subset $A$, we focus on the many-versus-one distinguishing problem between $\rho_m$ and the ensemble $\{\rho_P=(I+3\epsilon P)/2^{n},~ P\sim \pi\in \cD(A)\}$. In the purity testing task, we focus on the many-versus-one distinguishing problem between $\rho_m$ and the Haar random ensemble of pure states $\{\ketbra{\psi}\sim \mu_{\text{Haar}}\}$.

The essential idea for proving lower bounds based on \Cref{prob:Distinguish} is the two-point method. In the tree representation, the distribution of the leaves for the two events must be sufficiently distinct in order to distinguish between the two events. Formally, we have the following lemma:
\begin{lemma}[Le Cam's two-point method, see e.g. Lemma 1 in~\cite{yu1997assouad}]\label{lem:LeCam}
The probability that the learning protocol represented by the tree $\cT$ correctly solves the many-versus-one distinguishing task in Problem~\ref{prob:Distinguish} is upper bounded by
\begin{align*}
\frac12\sum_{\ell\in\mathrm{leaf}(\cT)}\abs{\E_a \,p^{\rho_a}(\ell)-p^{\rho_m}(\ell)}\,.
\end{align*}
\end{lemma}

\noindent One way of bounding this is by showing that the \emph{likelihood ratio}, defined below, is not too small with high probability over the leaves:

\begin{definition}[Likelihood ratio]
    Let $\cT$ be the tree representation of a learning protocol for the many-versus-one distinguishing task in Problem~\ref{prob:Distinguish}. For any $\ell\in\mathrm{leaf}(\cT)$, define the \emph{likelihood ratio}
    \begin{equation*}
        L(\ell)\coloneqq \frac{\E_a p^{\rho_a}(\ell)}{p^{\rho_m}(\ell)}.
    \end{equation*}
    Additionally, for every $a$, edge $e_{u, s}$, and leaf $\ell\in \mathrm{leaf}(\cT)$, define 
    \begin{equation*}
        L_a(u, s)\coloneqq\frac{p^{\rho_a}(s|u)}{p^{\rho_m}(s|u)}, \qquad L_a(\ell)\coloneqq\frac{p^{\rho_a}(\ell)}{p^{\rho_m}(\ell)}\,.
    \end{equation*}
\end{definition}
\begin{lemma}[One-sided bound suffices for Le Cam, adapted from \cite{chen2022exponential}, Lemma 5.4]\label{lem:one_side_likelihood}
    The probability $p_\text{succ}$ that the learning protocol represented by the tree $\cT$ correctly solves the many-versus-one distinguishing task in Problem~\ref{prob:Distinguish} is upper bounded by the following two quantities:
    \begin{equation*}
        p_\text{succ}\leq \Pr_{\ell\sim p^{\rho_m}(\ell)}[L(\ell)\leq \beta] + (1-\beta),
    \end{equation*}
    \begin{equation*}
        p_\text{succ}\leq \Pr_{a\sim\{p_a\}, \ell\sim p^{\rho_m}(\ell)}[L_a(\ell)\leq \beta] + (1-\beta),
    \end{equation*}
    where $\beta$ is an arbitrary constant in $(0, 1)$.
\end{lemma}
\begin{proof}By \Cref{lem:LeCam},
    \begin{align*}
        p_{\text{succ}}&\leq \frac12\sum_{\ell\in\mathrm{leaf}(\cT)}\abs{\E_a \,p^{\rho_a}(\ell)-p^{\rho_m}(\ell)} \\
        &= \sum_{\ell\in\mathrm{leaf}(\cT), \E_a \,p^{\rho_a}(\ell)\leq p^{\rho_m}(\ell)}p^{\rho_m}(\ell)-\E_a \,p^{\rho_a}(\ell)\\
        &= \sum_{\ell\in\mathrm{leaf}(\cT)} \max(0, p^{\rho_m}(\ell)-\E_a \,p^{\rho_a}(\ell))\\
        &= \E_{\ell\sim p^{\rho_m}(\ell)}\max(0, 1-L(\ell))\\
        &\leq \Pr_{l\sim p^{\rho_m}(\ell)}[L(\ell)\leq \beta] + (1-\beta).
    \end{align*}
    Since $\max\{0, x\}$ is convex, $\max\{0, 1-L(\ell)\}\leq \E_a\max\{0, 1-L_a(\ell)\}$ for any $\ell$. So we also have
    \begin{equation*}
        p_{\text{succ}}\leq \E_{\ell\sim p^{\rho_m}(\ell)}\max(0, 1-L(\ell))\leq \E_{a, \ell\sim p^{\rho_m}(\ell)}\max(0, 1-L_a(\ell))\leq \Pr_{a\sim \{p_a\}, \ell\sim p^{\rho_m}(\ell)}[L_a(\ell)\leq \beta] + (1-\beta).
    \end{equation*}
\end{proof}

\noindent This idea is further developed in~\cite{chen2022complexity} using martingale concentration inequalities. 
We distill the main idea therein into the following lemma:
\begin{lemma}\label{lem:MartingaleTrick}
    There exists a constant $c$ such that the following statement holds. Suppose $\cT$ is a tree representation of a learning protocol for the many-versus-one distinguishing problem defined in \Cref{prob:Distinguish}. 
    If there is a $\delta>0$ such that for every node $u$ we have
    \begin{equation*}
        \E_{a\sim \{p_a
        \}}\E_{s\sim \{p^{\rho_m}(s|u)\}}[(L_a(u, s)-1)^2]\leq \delta,
    \end{equation*}
    then 
    \begin{equation*}
        \Pr_{a, \ell\sim p^{\rho_m}(\ell)}[L_a(\ell)\leq 0.9] \leq 0.1+c\delta T.
    \end{equation*}
    As a consequence, the success probability of $\cT$ is at most $0.2+c\delta T$.
    In particular, if $\cT$ solves the many-versus-one distinguishing problem with success probability at least $1/3$, then $T\ge \Omega(1/\delta)$. 
\end{lemma}
\noindent We will prove this lemma in~\Cref{sec:Martingale}. Intuitively, the lemma states that if the likelihood ratio is concentrated around 1 for intermediate nodes in the tree, then the tree must have large depth to distinguish between the two events.

\subsection{Tail bounds}
\begin{lemma}[Hoeffding's inequality~\cite{hoeffdingProbabilityInequalitiesSums1963}, see e.g., Theorem 2.2.6 in~\cite{vershyninHighDimensionalProbabilityIntroduction2018}]\label{lem:Hoeffding}
    Let $X_1, \cdots, X_N$ be independent random variables such that $X_i\in [m_i, M_i]$ for every $i$. Then for any $t>0$, we have
    \begin{equation*}
        \Pr[\abs{\sum_{i=1}^N X_i-\E[X_i]}\geq t]\leq 2\exp\left(-\frac{2t^2}{\sum_{i=1}^N(M_i-m_i)^2}\right).
    \end{equation*}
\end{lemma}

\begin{lemma}\label{lem:Chernoff}
    Let $X_1, \cdots, X_N$ be i.i.d. random variables in $[0, 1]$ and $\E[X_i]=\mu$ for every $i$. Then 
    \begin{equation*}
        \Pr\left[\sum_{i=1}^N X_i\leq \frac{1}{2}N\mu\right] \leq e^{-\frac{N\mu}{8}}.
    \end{equation*}
\end{lemma}

\subsection{The minimax theorem}\label{sec:Prelim_Minimax}
We will use the minimax theorem to bridge the lower bounds and the upper bounds. We trim the minimax theorem to the following form for our convenience and to avoid mentioning infinite-dimensional linear spaces.
For any set $B$, let $\cD(B)$ be the set of discrete probability distributions over $B$, i.e.,
\begin{equation*}
    \cD(B)=\{(q_b)_{b\in B}|\text{countable nonzero elements}, q_b\ge 0, \sum_{b\in B} q_b=1\}.
\end{equation*}
Note that although $B$ might be an uncountable set, we can safely write $\sum_{b\in B}$ because there are only countably many nonzero terms in the sum.
\begin{lemma}\label{lem:minimax}
    Let $A$ be a finite set and $B$ be a compact subset of $\R^m$ for some $m$.  Let $f:A\times B\to \mathbb{R}$ be a function such that $f(a,\cdot)$ is continuous for every $a\in A$. Then the following two quantities exist and are equal:
    \begin{equation*}
        \min_{(p_a)\in \cD(A)}\max_{b\in B}\sum_{a\in A}p_a f(a, b) = \sup_{(q_b)\in \cD(B)}\min_{a\in A}\sum_{b\in B}q_b f(a, b).
    \end{equation*}
\end{lemma}
\begin{proof}
    The Heine-Cantor theorem states that a continuous function ($f(a, \cdot)$ here) from a compact metric space ($B$) to a metric space ($\R$) is uniformly continuous. In other words, for any $\eta>0$ and $a\in A$, there exists a $\gamma_a>0$ such that for any $b_1,b_2\in B$, $\dist(b_1, b_2)\leq \gamma_a$ implies $\abs{f(a, b_1)-f(a, b_2)}\leq \eta$ (here $\dist$ is the Euclidean distance). Let $\gamma\coloneqq \min_{a\in A}\gamma_a$.
    Obviously $B$ has a finite $\gamma$-net $B_\gamma$, which means for any $b\in B$, there exists $b'\in B_{\gamma}$ such that $\dist(b, b')\leq \gamma$. 
    Now $A, B_\gamma$ are both finite sets, and $\cD(A), \cD(B_\gamma)$ are both compact subsets of a finite-dimensional linear space. According to the standard Von Neumann's minimax theorem for finite-dimensional linear spaces~\cite{neumannZurTheorieGesellschaftsspiele1928}, 
    \begin{equation*}
        \min_{(p_a)\in \cD(A)}\max_{(q_b)\in \cD(B_\gamma)}\sum_{a\in A}\sum_{b\in B_\gamma}p_a q_b f(a, b) = \max_{(q_b)\in \cD(B_\gamma)}\min_{(p_a)\in \cD(A)}\sum_{a\in A}\sum_{b\in B_\gamma}p_a q_b f(a, b).
    \end{equation*}
    Therefore, 
    \begin{equation}
        \min_{(p_a)\in \cD(A)}\max_{b\in B_{\gamma}}\sum_{a\in A}p_a f(a, b) = \max_{(q_b)\in \cD(B_\gamma)}\min_{a\in A}\sum_{b\in B_\gamma}q_b f(a, b). \label{equ:minimax_proof_1}
    \end{equation}
    We now show that both sides converge to the desired quantities as $\eta\to 0$. Indeed, for any $b\in B$, let $\hat{b}\in B_\gamma$ be its closest point in $B_\gamma$, so $\dist(b, \hat{b})\leq \gamma$. For any $(p_a)\in \cD(A)$,
    \begin{equation*}
        \abs{\sum_{a\in A}p_a f(a, b)-\sum_{a\in A}p_a f(a, \hat{b})}\leq \sum_{a\in A}p_a \abs{f(a, b)-f(a, \hat{b})}\leq \eta.
    \end{equation*}
    So 
    \begin{equation}
        \abs{\min_{(p_a)\in \cD(A)}\max_{b\in B}\sum_{a\in A}p_a f(a, b)-\min_{(p_a)\in \cD(A)}\max_{b\in B_\gamma}\sum_{a\in A}p_a f(a, b)}\leq \eta \label{equ:minimax_proof_2}.
    \end{equation}
    For any $(q_b)\in \cD(B)$, let $(r_{b'}\coloneqq \sum_{b\in B: \hat{b}=b'}q_b)\in \cD(B_\gamma)$ be its approximation. Then
    \begin{equation*}
        \abs{\sum_{b\in B}q_b f(a, b)-\sum_{b'\in B_\gamma}r_{b'} f(a, b')}=\abs{\sum_{b\in B}q_b f(a, b)-\sum_{b\in B}q_b f(a, \hat{b})}\leq \sum_{b\in B}q_b \abs{f(a, b)-f(a, \hat{b})}\leq \eta.
    \end{equation*}
    So 
    \begin{equation}
        \abs{\sup_{(q_b)\in \cD(B)}\min_{a\in A}\sum_{b\in B}q_b f(a, b)-\max_{(q_b)\in \cD(B_\gamma)}\min_{a\in A}\sum_{b\in B_\gamma}q_b f(a, b)}\leq \eta. \label{equ:minimax_proof_3}
    \end{equation}
    Combining \eqref{equ:minimax_proof_1}, \eqref{equ:minimax_proof_2}, and \eqref{equ:minimax_proof_3}, we have 
    \begin{equation*}
        \abs{\min_{(p_a)\in \cD(A)}\max_{b\in B}\sum_{a\in A}p_a f(a, b)-\sup_{(q_b)\in \cD(B)}\min_{a\in A}\sum_{b\in B}q_b f(a, b)}\leq 2\eta.
    \end{equation*}
    Since $\eta>0$ is arbitrary, the lemma follows.
\end{proof}

\subsection{Tensor network diagrams and matrix product states}\label{sec:Prelim_TensorNetwork}
It is sometimes useful to represent quantum states and operators using quantum-circuit-like diagrams. It is the simplest special case of tensor network diagrams, where all legs are horizontal from left to right. A $2^n\times 2^k$ matrix $A$ can be written as 
\begin{equation*}
\begin{quantikz} n~& \gate{A}&\setwiretype{c} ~k\end{quantikz}.
\end{equation*}
Here the integer near a wire indicates the number of qubits on the wire. By default, we will use ``\begin{quantikz} &\end{quantikz}'' to represent a $n$-qubit wire and ``\begin{quantikz} \setwiretype{c} &\end{quantikz}'' to represent a $k$-qubit wire ($k$ is the size of quantum memory). 

A block can have multiple input or output wires to indicate the tensor product structure. For example, 
\begin{equation*}
    \begin{quantikz}
        &\gate[2]{A}&\setwiretype{c}\\
        & \setwiretype{c} &
    \end{quantikz}
\end{equation*}
is a $2^{n+k}\times 2^{2k}$ matrix from $\cH_n\otimes \cH_k$ to $\cH_k\otimes \cH_k$.

A block can also have zero input or output wires to represent a quantum state. For example, a $n$-qubit quantum state $\ket{\psi}$ can be written as
\begin{equation*}
    \ket{\psi} = \begin{quantikz} &\meterD{\psi}
    \end{quantikz},\quad
    \bra{\psi} = \begin{quantikz}
        \inputD{\psi}&
    \end{quantikz}.
\end{equation*}

The identity operator on $n$ qubits and the swap operator on two $n$-qubit systems are expressed as
\begin{equation*}
    I= \begin{quantikz} &&\end{quantikz},\quad \SWAP = \begin{quantikz} &\permute{2,1}&\\&&\end{quantikz}.
\end{equation*}

Matrix multiplication is simple in diagrams. We just connect the corresponding wires of two blocks. For example
\begin{equation*}
    \begin{quantikz}
        \inputD{\psi_1}& & \gate[2]{A}&\setwiretype{c}&\gate[2]{B} &\setwiretype{q} &\meterD{\psi_2}\\
        \inputD{\psi_3}& \gate[1]{C} \setwiretype{c}& & & &&\meterD{\psi_4}
    \end{quantikz}
\end{equation*}
means $\braket{\psi_1\psi_3|(I_n\otimes C)AB|\psi_2\psi_4}$, where the dimensions of these states and matrices should be clear from the wires.

The tensor network diagrams are helpful in defining matrix product states (MPSs).

\begin{definition}[Matrix product states (MPSs)]\label{def:MPS}
    A $(2^n, 2^k, c)$ matrix product state (MPS) $\ket{M}$ is a quantum state defined on $c$ qudits of local dimension $2^n$ with the following form:
    \begin{equation*}
        \ket{M} = \sum_{s_1,\cdots, s_c\in [0:2^n-1]}A_1^{(s_1)}A_2^{(s_2)}\cdots A_c^{(s_c)}\ket{s_1s_2\cdots s_c}.
    \end{equation*}
    Here for any $s_1,\cdots, s_c\in [0:2^n-1]$, $A_1^{(s_1)}$ is a $1\times 2^k$ matrix, $A_i^{(s_i)}$ is a $2^k\times 2^k$ matrix ($i=2,3,\cdots, c-1$), and $A_c^{(s_c)}$ is a $2^k\times 1$ matrix. The set of all normalized $(2^n, 2^k, c)$ MPSs is denoted by $\mathrm{MPS}_c(2^n, 2^k)$. 

    Equivalently, an MPS is a state that can be written as the following diagram (here we only draw the $c=4$ case for simplicity):
    \begin{equation*}
        \ket{M} = \begin{quantikz}
            & \gate[1]{B_1}&\gate[2]{B_2} \setwiretype{c}\\
            & & & \gate[2]{B_3} \setwiretype{c}\\
            & & & & \gate[2]{B_4} \setwiretype{c}\\
            & & & &
        \end{quantikz}.
    \end{equation*}
    The correspondence between the two definitions is given by
    \begin{equation*}
        A_1^{(s_1)} = \begin{quantikz}\inputD{s_1}&\gate[1]{B_1}&\setwiretype{c}
        \end{quantikz},\quad
        A_c^{(s_c)}=\begin{quantikz}&\gate[2]{B_4}\setwiretype{c}\\\inputD{s_c}&
        \end{quantikz},\quad
        A_i^{(s_i)}=\begin{quantikz}
            &\gate[2]{B_i}\setwiretype{c}\\\inputD{s_i}&&\setwiretype{c}
        \end{quantikz}.
    \end{equation*}
    Then (again, we only draw the $c=4$ case in the middle)
    \begin{equation*}
        \braket{s_1s_2\cdots s_c|M}= \begin{quantikz}
            \inputD{s_1}& \gate[1]{B_1}&\gate[2]{B_2} \setwiretype{c}\\
            \inputD{s_2} & & & \gate[2]{B_3} \setwiretype{c}\\
            \inputD{s_3} & & & & \gate[2]{B_4} \setwiretype{c}\\
            \inputD{s_4}& & & &
        \end{quantikz} = A_1^{(s_1)}A_2^{(s_2)}\cdots A_c^{(s_c)}.
    \end{equation*}
\end{definition}

\noindent We can also define MPSs based on entanglement.
\begin{lemma}\label{lem:MPS_Schmidt}
    Let $\ket{M}\in \mathrm{MPS}_c(2^n, 2^k)$ be a MPS. Then the Schmidt rank (i.e., the number of terms in the Schmidt decomposition) of $\ket{M}$ with respect to any cut $[1:r]\cup [r+1:c]$ is at most $2^k$. In other words, for any $r\in [c-1]$, $\ket{M}$ can be written as
    \begin{equation*}
        \ket{M}=\sum_{i=1}^{2^k}\sqrt{\lambda_i}\ket{\alpha_i^{[1:r]}}\otimes\ket{\beta_i^{[r+1:c]}},
    \end{equation*}
    where $\{\ket{\alpha_i^{[1:r]}}\}$ ($\{\ket{\beta_i^{[r+1:c]}}\}$) are some orthonormal states defined on the first $r$ (last $c-r$) qudits, and $\sum_i{\lambda_i}=1$.
\end{lemma}

\section{Pauli shadow tomography with \texorpdfstring{$(c, \cM)$}{Lg} protocols}\label{sec:GeneralThm}
In this section, we prove the master theorem, restated for convenience below:
\begingroup
\def\thetheorem{\ref{thm:general_theorem}}
\begin{theorem}[Master theorem]
    The sample complexity of Pauli shadow tomography for any $A\subseteq\cP_n$ using $(c, \cM)$-protocols is characterized by the following quantity: \begin{equation*}
        \delta_{c, \cM}(A)\coloneqq \min_{\pi\in\cD(A)}\max_{M \in \cM}\, \mathbb{E}_{P\sim \pi}\, \chi^2_M\Bigl(\Bigl(\frac{I + 3\epsilon P}{2^n}\Bigr)^{\otimes c} \Big\| \Bigl(\frac{I}{2^n}\Bigr)^{\otimes c}\Bigr)\,,
        \end{equation*}
        where $\cD(A)$ denotes the set of probability distributions over $A$. Specifically, for any $A\subseteq\cP_n$, we have that:
        \begin{itemize}[leftmargin=*, itemsep=0pt,topsep=0.3em]
            \item Any $(c, \cM)$-protocol requires at least $\Omega(c/(\delta_{c, \cM}(A)))$ copies of $\rho$.
            \item If $c = 1$ and $\cM$ is Pauli-closed, then there is a $(c,\cM)$-protocol using $O(\log(|A|)/\delta_{1,\cM}(A))$ copies of $\rho$.
            \item For general $c$, if $\cM$ is Pauli-closed and contains all single-copy measurements, then there is a $(c, \cM)$-protocol using $2^{O(c)} \log(\abs{A})/\delta_{c, \cM}(A) + O(\log(|A|) / \epsilon^4)$ copies of $\rho$.
        \end{itemize}
\end{theorem}
\addtocounter{theorem}{-1}
\endgroup
\noindent The theorem characterizes the sample complexity of Pauli shadow tomography using $(c, \cM)$-protocols. It is tight up to a $\log(\abs{A})$ factor when $c=O(1)$, and tight up to a $\poly(n)$ factor when $c=\log(n)$. The theorem is very flexible. One can choose an arbitrary subset $A$ of Pauli strings and any reasonable POVM set. We will show various applications of the master theorem. Especially, we use the master theorem to build the optimal tradeoff between the size of the quantum memory and the sample complexity. 

\subsection{Lower bound}
We only need to prove the lower bound for the many-versus-one distinguishing problem between $\rho_m$ and ensemble $\{\rho_{P}\coloneqq \frac{I+3\epsilon P}{2^n},~P\sim \pi\in \cD(A)\}_{P\in A}$, since it can be reduced to the Pauli shadow tomography task. Assume there is a $(c, \cM)$ protocol that solves this distinguishing problem in $T$ rounds. The algorithm can be represented by a learning tree $\cT$ of depth $T$.

Let $u$ be an internal node in the learning tree, and $M_u = \{F_s^u\}_s$ be the POVM used in the node. Then the probability of outcome $s$ is $p^{\rho}(s|u)=\tr(F_s^u\rho^{\otimes c})$. The likelihood ratio is $L_P(u, s)=\tr(F_s^u \rho_P^{\otimes c})/\tr(F_s^u \rho_m^{\otimes c})$. 
\begin{align}
    \MoveEqLeft\E_{P\sim \pi}\E_{s\sim \{p^{\rho_m}(s|u)\}}[(L_P(u, s)-1)^2]\nonumber\\
    =&\E_{P\sim \pi}\E_{s\sim \{p^{\rho_m}(s|u)\}}\left[\left(\frac{\tr(F_s^u\rho_P^{\otimes c})}{\tr(F_s^u\rho_m^{\otimes c})}-1\right)^2\right]\nonumber\\
    =&\E_{P\sim \pi} \chi^2_{M_u}(\rho_P^{\otimes c}\| \rho_m^{\otimes c})\nonumber\\
    \leq &\max_{M\in \cM}\E_{P\sim \pi}\chi^2_M(\rho_P^{\otimes c}\| \rho_m^{\otimes c}).\nonumber
\end{align}
We choose $\pi$ that minimizes the last line so that the last line is equal to $\delta_{c, \cM}(A)$ by definition \eqref{eq:general_delta}. 
By \Cref{lem:MartingaleTrick}, the depth $T$ of the learning tree $\cT$ is lower bounded by $\Omega(1/\delta_{c, \cM}(A))$.
Therefore, the sample complexity of the Pauli shadow tomography is lower bounded by $\Omega(c/\delta_{c, \cM}(A))$. This completes the proof of the lower bound in the master theorem. We remark that the lower bound still holds if $A$ is a set of traceless operator $O$ such that $\tr(O^2)\ge 2^n$. 

\subsection{Upper bound}
\label{sec:upperbound}
To prove the upper bound, we rewrite \eqref{eq:general_delta} using the minimax theorem.
By \Cref{lem:minimax}, there exist a finite ensemble $\{M_l\text{ w.p. }q_l\}_{l\in \N}$ such that 
\begin{equation}
    \min_{P\in A}\E_{l\sim \{q_l\}}[\chi^2_{M_l}(\rho_P^{\otimes c}\| \rho_m^{\otimes c})]\ge \delta_{c, \cM}(A)/2. \label{eq:delta_cM_maxmin1}
\end{equation}
Denote $M_l = \{F_{l, s}\}_s$. For every $cn$-qubit Pauli string $Q\in \cP_{cn}$, define a POVM $M_{l, Q}$ as $QM_lQ\coloneqq \{QF_{l, s}Q\}_s$. Since $\cM$ is Pauli-closed, $QM_lQ\in \cM$. We now design an ensemble of POVMs in $\cM$: $\{M_{l, Q}\}_{(l, Q)}$ with probability $\{q_{l, Q}\coloneqq q_l/4^{nc}\}_{(l, Q)}$. 

In each round, we measure the state $\rho^{\otimes c}$ using the ensemble. In other words, we first sample $M_{l, Q}$ from the distribution $\{q_{l, Q}\}_{(l, Q)}$, then measure the state using $M_{l, Q}$. We will repeat the measurement for $T$ times and denote the outcome of round $t$ by $(l_t, Q_t, s_t)$. In the following, we construct our learning algorithm step by step.

The probability distribution of outcome $(l, Q, t)$ is $\Pr[l, Q, t]=q_l\tr(QF_{l, s}Q\rho^{\otimes c})/4^{nc}$. All calculations of probabilities and expectations below are respective to this distribution. By \Cref{lem:SumPauliExp2}, we can calculate the marginal probability $\Pr[l, s]$ and the conditional probability $\Pr[Q|l, s]$:
\begin{align*}
    \Pr[l, s]&=\sum_{Q}\Pr[l, Q, s]=\sum_{Q}\frac{q_l\tr(QF_{l, s}Q\rho^{\otimes c})}{4^{nc}}=\frac{q_l\tr(F_{l, s})}{2^{nc}},\\
    \Pr[Q|l, s]&= \Pr[l, Q, s]/\Pr[l, s]=\frac{\tr(QF_{l, s}Q\rho^{\otimes c})}{2^{nc}tr(F_{l, s})}.
\end{align*}
Fix a $P\in A$ and $\emptyset\neq S\subseteq [c]$. We have (using \Cref{lem:SumPauliExp2} and $\tr(P^S\rho^{\otimes c})=\tr(P\rho)^{\abs{S}}$)
\begin{align}
    \E_{Q}[\epsilon_{Q, P^S}|l, s] &= \sum_{Q}\Pr[Q|l, s]\epsilon_{Q, P^S} = \sum_{Q}\frac{\tr(QF_{l, s}Q\rho^{\otimes c})}{2^{nc}tr(F_{l, s})}\epsilon_{Q, P^S}\nonumber\\
    &= \sum_{Q}\frac{\tr(QF_{l, s}P^SQP^S\rho^{\otimes c})}{2^{nc}tr(F_{l, s})}\nonumber\\
    &= \frac{\tr(F_{l, s}P^{S})\tr({P}\rho)^{\abs{S}}}{\tr(F_{l, s})}.\nonumber
\end{align}
Therefore, the indicator of commutation $\epsilon_{Q, P^S}$ can serve as an estimator of $\tr({P}\rho)^{\abs{S}}$. We want to average the estimators over all rounds. However, the coefficients $\tr(F_{l_t, s_t}P^{S})/\tr(F_{l_t, s_t})$ in different rounds may have different signs and cancel. To address this issue, we square the coefficients by considering
\begin{equation*}
    \E_{Q}\left[\frac{\tr(F_{l, s}P^S)\epsilon_{Q, P^S}}{\tr(F_{l, s})}|l, s\right]=\frac{\tr(F_{l, s}P^{S})^2\tr({P}\rho)^{\abs{S}}}{\tr(F_{l, s}P^{S})^2}.
\end{equation*}
Then we average the modified estimators:
\begin{equation}
    \hat{E}_{{P}, S}\coloneqq \frac{\sum_{t=1}^T \tr(F_{l_t, s_t}P^S)\epsilon_{Q_t, P^S}/\tr(F_{l_t, s_t})}{\sum_{t=1}^T\tr(F_{l_t, s_t}P^{S})^2/\tr(F_{l_t, s_t})^2}. \label{eq:general_theorem_estimator}
\end{equation}
We already know that $\hat{E}_{{P}, S}$ is an unbiased estimator of $\tr({P}\rho)^{\abs{S}}$ conditioned on $(l_t, s_t)_{t=1}^T$.
Using Hoeffding's inequality, we can calculate the performance of this estimator:
\begin{align}
    &\Pr_{(Q_t)_{t=1}^T}\left[\abs{\hat{E}_{{P}, S}-\tr({P}\rho)^{\abs{S}}}\ge \epsilon^{\abs{S}}|(l_t, s_t)_{t=1}^T\right]\nonumber\\
    =&\Pr_{(Q_t)_{t=1}^T}\left[\abs{\sum_{i=1}^T\frac{\tr(F_{l_t, s_t}P^S)\epsilon_{Q, P^S}}{\tr(F_{l_t, s_t})}-\sum_{i=1}^T\E[\frac{\tr(F_{l_t, s_t}P^S)\epsilon_{Q, P^S}}{\tr(F_{l_t, s_t})}]}\ge \epsilon^{\abs{S}}\sum_{t=1}^T\frac{\tr(F_{l_t, s_t}P^{S})^2}{\tr(F_{l_t, s_t})^2}~|~(l_t, s_t)_{t=1}^T\right]\nonumber\\
    \leq & 2\exp\left(-\frac{1}{2}\epsilon^{2\abs{S}}\sum_{t=1}^T\frac{\tr(F_{l_t, s_t}P^S)^2}{\tr(F_{l_t, s_t})^2}\right),\label{eq:estimator_condition_on_lq}
\end{align}
where we set error as $\epsilon^\abs{S}$ because eventually we will use $\abs{\hat{E}_{{P}, S}}^{1/\abs{S}}$ as the estimator of $\abs{\tr({P}\rho)}$.
When the exponent $\frac{1}{2}\epsilon^{2\abs{S}}\sum_{t=1}^T\tr(F_{l_t, s_t}P^{S})^2/\tr(F_{l_t, s_t})^2$ is large, $\hat{E}_{{P}, S}$ is a good estimator.  

The expectation of a single term is
\begin{align}
    \E_{(l, s)}\left[\frac{\tr(F_{l, s}P^S)^2}{\tr(F_{l, s})^2}\right]&=\sum_{l, s}\Pr[l, s]\frac{\tr(F_{l, s}P^S)^2}{\tr(F_{l, s})^2}=\E_{l}\sum_{s}\frac{\tr(F_{l, s}P^S)^2}{2^{cn}\tr(F_{l, s})}. \label{eq:general_theorem_lower_bound_proof_0}
\end{align}
Denote the right-hand side by $\mu(P, S)$. By Chernoff bound (\Cref{lem:Chernoff}), we have 
\begin{equation}
    \Pr_{(l_t, s_t)_{t=1}^T}\left[\sum_{t=1}^T\frac{\tr(F_{l, s}P^{S})^2}{\tr(F_{l, s})^2}\leq \frac{1}{2}T\mu(P, S)\right]\leq \exp(-\frac{1}{8}T\mu(P, S)).\label{eq:general_theorem_lower_bound_proof_1}
\end{equation}
We say $(l_t, s_t)_{t=1}^T$ is $(P, S)$-good if $\sum_{t=1}^T\tr(F_{l, s}P^{S})^2/\tr(F_{l, s})^2> \mu(P, S)T/2$. From \eqref{eq:estimator_condition_on_lq} and \eqref{eq:general_theorem_lower_bound_proof_1}, we know
\begin{align}
    &\Pr[\abs{\hat{E}_{P, S}-\tr(P\rho)^{\abs{S}}}\ge \epsilon^{\abs{S}}]\nonumber\\
    \leq & \Pr[\abs{\hat{E}_{P, S}-\tr(P\rho)^{\abs{S}}}\ge \epsilon^{\abs{S}}~|~(l_t, s_t)_{t=1}^T \text{ is }(P, S)\text{-good}] + \Pr[(l_t, s_t)_{t=1}^T \text{ is not }(P, S)\text{-good}]\nonumber\\
    \leq & 2\exp(-\frac{1}{4}\epsilon^{2\abs{S}}\mu(P, S)T)+\exp(-\frac{1}{8}\mu(P, S)T)\nonumber\\
    \leq & 3\exp(-\frac{1}{8}\epsilon^{2\abs{S}}\mu(P, S)T). \label{eq:general_theorem_lower_bound_proof_2}
\end{align}
Therefore, the performance of $\hat{E}_{P, S}$ is determined by the quantity $\mu(P, S)$ defined in \eqref{eq:general_theorem_lower_bound_proof_0}. By \eqref{eq:delta_cM_maxmin1} and Cauchy-Schwarz inequality, we have
\begin{align}
    \frac{\delta_{c, \cM}(A)}{2}&\leq \E_{l}[\chi^2_{M_l}(\rho_P^{\otimes c}\| \rho_m^{\otimes c})]\nonumber\\
    &= \E_{l}\sum_{s}\frac{\tr(F_{l, s})}{2^{cn}}\left(\frac{\tr(F_{l, s}(I+3\epsilon P)^{\otimes c})}{\tr(F_{l, s})}-1\right)^2\nonumber\\
    &= \E_{l}\sum_{s}\frac{1}{2^{cn}\tr(F_{l, s})}\left(\sum_{\emptyset\neq S\subseteq [c]}(3\epsilon)^{\abs{S}}\tr(F_{l, s}P^{S})\right)^2\nonumber\\
    &\leq \E_{l}\sum_{s}\frac{1}{2^{cn}\tr(F_{l, s})}\sum_{\emptyset\neq S\subseteq [c]}\mu(P, S)^{-1/2}(3\epsilon)^{\abs{S}}\tr(F_{l, s}P^S)^2\sum_{\emptyset\neq S\subseteq [c]}\mu(P, S)^{1/2}(3\epsilon)^{\abs{S}}\nonumber\\
    &= \left(\sum_{\emptyset\neq S\subseteq [c]}\mu(P, S)^{1/2}(3\epsilon)^{\abs{S}}\right)^2\leq \left(\sum_{\emptyset\neq S\subseteq[c]}3^{\abs{S}}\right)^2\max_{\emptyset\neq S\subseteq [c]}\epsilon^{2\abs{S}}\mu(P, S)\leq 16^c\max_{\emptyset\neq S\subseteq [c]}\epsilon^{2\abs{S}}\mu(P, S). \label{eq:general_theorem_lower_bound_proof_3}
\end{align}
Let $\emptyset\neq S_P\subseteq [c]$ be the set that maximize $\epsilon^{2\abs{S}}\mu(P, S)$. Then \eqref{eq:general_theorem_lower_bound_proof_2} implies that $\epsilon^{2\abs{S_P}}\mu(P, S_P)\ge \delta_{c, \cM}/(2\times 16^c)$. Plugging $S=S_P$ into \eqref{eq:general_theorem_lower_bound_proof_2}, we obtain
\begin{equation*}
    \Pr\left[\abs{\hat{E}_{P, S_P}-\tr(P\rho)^{\abs{S_P}}}\ge \epsilon^{\abs{S_P}}\right]\leq 3\exp(-\delta_{c, \cM}(A)T/16^{c+1}).
\end{equation*}
Now we set the $T=16^{c+1}\log(30\abs{A})/\delta_{c, \cM}(A)$. We arrive at
\begin{equation*}
    \Pr[\abs{\hat{E}_{{P}, S_P}-\tr({P}\rho)^{\abs{S_P}}}\ge \epsilon^{\abs{S_P}}]\leq 1/(10\abs{A}).
\end{equation*}

When $c=1$, we have $\abs{S_P}=1$, so $\hat{E}_{P, S_P}$ is already a good estimator of $\tr(P\rho)$. By union bound, the algorithm succeeds with probability at least $9/10$ as required. 

When $c\ge 2$, we can estimate $\abs{\tr(P\rho)}$ within error $\epsilon$, according to following simple fact:
\begin{lemma}
    Let $x, y\in \R$ and $r$ be a positive integer. If $\abs{x-y}\leq \epsilon^r$, then $\abs{\abs{x}^{1/r}-\abs{y}^{1/r}}\leq \epsilon$.
\end{lemma}
\begin{proof}
    $\abs{\abs{x}-\abs{y}}\leq \abs{x-y}\leq \epsilon^r$.
    Without loss of generality, assume $\abs{x}\geq \abs{y}$. Then $\abs{x}-\abs{y}\leq \epsilon^r$, so $\abs{x}^{1/r}\leq (\abs{y}+\epsilon)^{1/r}\leq \abs{y}^{1/r}+\epsilon^{1/r}$, where the last inequality is obvious by raising both sides to the power of $r$.
\end{proof}
Using this lemma and the union bound, we finally arrive at
\begin{equation*}
    \Pr[\forall {P}\in A,~\abs{\abs{\hat{E}_{{P}, S_P}}^{1/\abs{S_P}}-\abs{\tr(P\rho)}}< \epsilon]\ge 9/10.
\end{equation*}
In other words, we can learn $\abs{\tr(P\rho)}$ for every $P\in A$ within error $\epsilon$ with probability at least $9/10$ in $O(16^c\log(\abs{A})/\delta_{c, \cM}(A))$ rounds. There is still a gap between the absolute value $\abs{\tr(P\rho)}$ and the true value $\tr(P\rho)$. However, we show that as long as we can apply arbitrary single-copy measurements, learning absolute values is equivalent to learning the true values, in the following \Cref{thm:learn_sign_from_absolute}:
\begingroup
\def\thetheorem{\ref{thm:learn_sign_from_absolute}}
\begin{theorem}
    Let $\rho$ be an unknown state. Suppose we are given estimates $\{f_P\}_{P\in A}$ of the absolute values $\{\abs{\tr({P}\rho)}\}_{P \in A}$ which are accurate to additive error $\epsilon$ for every  $P\in A$. Then there is a protocol that takes $\{f_P\}_{P\in A}$, performs single-copy measurements on $O(\log(\abs{A})/\epsilon^4)$ copies of $\rho$, and with probability $9/10$ over the randomness of the measurement outcomes, estimates the expectation values $\{\tr(P\rho)\}_{P\in A}$ to additive error $3\epsilon$.
\end{theorem}
\addtocounter{theorem}{-1}
\endgroup

\noindent The proof of \Cref{thm:learn_sign_from_absolute} is at the end of this section. Combining the estimators of $\abs{\tr(P\rho)}$ and \Cref{thm:learn_sign_from_absolute}, we can learn $\tr(P\rho)$ within error $3\epsilon$ for every $P\in A$ using $O(16^cc\log(\abs{A})/\delta_{c, \cM}(A)+\log(\abs{A})/\epsilon^4)$ samples with probability $(9/10)^2>2/3$. A subtle gap is that \Cref{thm:learn_sign_from_absolute} will amplify the error by a factor of $3$. Therefore, we need to replace the $\epsilon$ in \eqref{eq:general_theorem_lower_bound_proof_2} by $\epsilon/3$. A slight modification of \eqref{eq:general_theorem_lower_bound_proof_3} yields $\max_{\emptyset\neq S\subseteq [c]}(\epsilon/3)^{2\abs{S}}\mu(P, S)\ge \delta_{c, \cM}/(2\times 100^c)$. So the final sample complexity is $O(100^cc\log(\abs{A})/\delta_{c, \cM}(A)+\log(\abs{A})/\epsilon^4)$.

Another subtle gap is that, in our definition of $(c, \cM)$ protocols, the POVM in each round is deterministic, while in the learning algorithm here, the POVM in each round is sampled from an ensemble. 
This issue can be addressed. 
Since the overall success probability with respective to $(l_t, Q_t, s_t)_{t=1}^T$ is at least $2/3$, there must be a sequence $(l_t', Q_t')_{t=1}^T$ such that the success probability conditioned on $(l_t', Q_t')_{t=1}^T$ is at least $2/3$. Therefore, we can deterministically apply POVM $M_{l_t', Q_t'}$ in $t$-th round. The success probability is still at least $2/3$. This completes the proof of \Cref{thm:general_theorem}.

We prove \Cref{thm:learn_sign_from_absolute} before closing this section. 

\begin{proof}[Proof of \Cref{thm:learn_sign_from_absolute}]
    We describe our algorithm in \Cref{fig:learning_sign_from_absolute}.
    Now we show the correctness of this algorithm.

\begin{algorithm}[ht]
    \fbox{\parbox{0.92\textwidth}{
        \textbf{Require:} $\epsilon>0$, $A\subseteq \cP_n$, copies of unknown state $\rho$, $\{f_P\}_{P\in A}$ such that $\abs{f_P-\abs{\tr(P\rho)}}\leq \epsilon$ for every $P\in A$.\\
            \textbf{Ensure:} Use $O(\log(A)/\epsilon^4)$ single-copy measurements. Output $\hat{E}_P$ for every $P\in A$ such that $\abs{\hat{E}_P-\tr(P\rho)}\leq \epsilon$. Succeed with probability at least $9/10$.
        \begin{enumerate}
            \item Find a state $\sigma$ such that $\abs{f_P-\abs{\tr(P\sigma)}}\leq \epsilon$ for every $P\in A$.
            \item Do Bell measurement on $\sigma\otimes \rho$ for $O(\log(A)/\epsilon^4)$ times, and estimate $\tr(P\otimes P~\sigma\otimes \rho)=\tr(P\sigma)\tr(P\rho)$ for every $P\in A$ within additive error $\epsilon^2$ with probability at least $9/10$. Denote the estimator of $\tr(P\sigma)\tr(P\rho)$ by $g_P$.
            \item If $f_P<2\epsilon$, set $\hat{E}_P=0$. Otherwise if $f_P>2\epsilon$, set $\hat{E}_P=g_P/\tr(\sigma P)$. Output $\{\hat{E}_P\}_{P\in A}$.
        \end{enumerate}
    }
    }
    \caption{Algorithm for learning signs from absolute values in \Cref{thm:learn_sign_from_absolute}. Note that the Bell measurement on $\sigma\otimes \rho$ is essentially a single-copy POVM on $\rho$, see \Cref{remark:Bell_measurement_single_copy}.}
    \label{fig:learning_sign_from_absolute}
\end{algorithm}
    In the first step, we can always find such a $\sigma$ since $\rho$ itself satisfies the condition (though it might be computationally hard). We know full information of the $\sigma$ we find, especially the exact value of $\tr(P \sigma)$. 

    Since $\{P\otimes P\}_{P\in \cP_n}$ is commuting, we can measure them simultaneously. 
    Indeed, the common eigenstates $\{P\otimes P\}_{P\in \cP_n}$ form a orthonormal basis (called Bell basis) $\{\ket{\Psi_Q}\coloneqq (I\otimes Q)\ket{\Psi_I}\}_{Q\in \cP_n}$, where $\ket{\Psi_I}=\frac{1}{2^{n/2}}\sum_{x\in \{0, 1\}^n}\ket{xx}$ is the standard Bell state. 
    So the spectrum decomposition of $P\otimes P$ has form $\sum_{Q\in \cP_n}\mu(P, Q)\ketbra{\Psi_Q}$ for some $\mu(P, Q)\in\{\pm 1\}$. 
    Now we measure a $2n$-qubit state $\sigma\otimes \rho$ in the Bell basis (this is called a Bell measurement). 
    Then $\E_Q[\mu(P, Q)]=\sum_Q \tr(\ketbra{\Psi_Q}~\sigma\otimes \rho)\mu(P, Q)=\tr(P\otimes P ~\sigma\otimes \rho)$.
    By Hoeffding's inequality and union bound, we can estimate $\tr(P\otimes P~\sigma\otimes \rho)$ within additive error $\epsilon^2$ with probability $9/10$ for every $P\in A$ using $O(\log(\abs{A})/\epsilon^4)$ Bell measurements. 

    Finally, suppose we have obtained $g_P$ such that $\abs{g_P-\tr(P\sigma)\tr(P\rho)}\leq \epsilon^2$ for every $P\in A$. If $f_P<2\epsilon$, we set $\hat{E}_P=0$, then 
    \begin{equation*}
        \abs{\hat{E}_P-\tr(P\rho)}=\abs{\tr(P\rho)-f_P+f_P}\leq \abs{\tr(P\rho)-f_P}+\abs{f_P}\leq \epsilon+2\epsilon=3\epsilon.
    \end{equation*}
    If $f_P>2\epsilon$, we set $\hat{E}_P=g_P/\tr(P\sigma)$, then 
    \begin{equation*}
        \abs{\hat{E}_P-\tr(P\rho)}=\frac{\abs{g_P-\tr(P\sigma)\tr(P\rho)}}{\abs{\tr(P\sigma)}}\leq \frac{\epsilon^2}{f_P-\epsilon}\leq \epsilon.
    \end{equation*}
    Therefore, the output $\{\hat{E}_P\}_{P\in A}$ estimates $\{\tr(P\rho)\}_{P\in A}$ to additive error $3\epsilon$. The success probability is at least $9/10$. 
\end{proof}

Our learning protocol for \Cref{thm:general_theorem} is summarized in \Cref{fig:main_theorem_algortihm}.

\begin{algorithm}[ht]
    \fbox{\parbox{0.92\textwidth}{
    \textbf{Require:} $c\in \N$, $\cM$ a closed and Pauli-closed (and contains all single-copy measurements if $c\ge 2$) set of $cn$-qubit POVMs, $\epsilon>0$, $A\subseteq \cP_n$, $T$ the number of rounds, copies of an unknown state $\rho$.\\
    \textbf{Ensure:} Use $(c, \cM)$ protocol of depth $T=16\times 100^{c}\log(30\abs{A})/\delta_{c, \cM}$. Output $\hat{E}_P$ for every $P\in A$ such that $\abs{\hat{E}_P-\tr(P\rho)}\leq \epsilon$. Succeed with probability at least $2/3$.
    \begin{enumerate}
        \item Find a finite ensemble $\{M_l=\{F_{l, s}\}_s\in \cM \text{ w.p. }q_l\}$ such that \[\min_{P\in A}\E_{l\sim \{q_l\}}[\chi^2_{M_l}(\rho_P^{\otimes c}\| \rho_m^{\otimes c})]\ge \delta_{c, \cM}(A)/2.\]  
        \item In round $t\in [T]$, sample a $l_t$ from $\{q_l\}$ and uniformly sample a $cn$-qubit Pauli string $Q_t$ from $\cP_{cn}$. Measure $\rho^{\otimes c}$ using the POVM $Q_tM_{l_t}Q_t=\{Q_tF_{l_t, s}Q_t\}_s$. Denote the outcome by $s_t$.
        \item For every ${P}\in A$, find $S_P\coloneqq \arg \max_{\emptyset\neq S\subseteq [c]}(\epsilon/3)^{2\abs{S}}\E_l\sum_s\frac{\tr(F_{l, s}P^S)^2}{2^{cn}\tr(F_{l, s})}$.
        \item Calculate $\hat{E}_{P, S_P}\coloneqq \frac{\sum_{t=1}^T \tr(F_{l_t, s_t}P^{S_P})\epsilon_{Q_t, P^{S_P}}/\tr(F_{l_t, s_t})}{\sum_{t=1}^T\tr(F_{l_t, s_t}P^{S_P})^2/\tr(F_{l_t, s_t})^2}$, the estimator of $\tr({P}\rho)^\abs{S_{P}}$.
        \item If $c=1$, output $\hat{E}_{P, S_P}$ as the estimator of $\tr(P\rho)$. If $c\ge 2$, do the following two steps.
        \item Calculate $\abs{\hat{E}_{P, S_P}}^{1/\abs{S_{P}}}$, the estimator of $\abs{\tr(P\rho)}$.
        \item Learn $\tr({P}\rho)$ using the algorithm in \Cref{thm:learn_sign_from_absolute}.
    \end{enumerate}
    }}
    \caption{The $(c, \cM)$ protocols for Pauli shadow tomography in \Cref{thm:general_theorem}.}
    \label{fig:main_theorem_algortihm}
\end{algorithm}

\begin{remark}\label{remark:Bell_measurement_single_copy}
    $\tr(\ketbra{\Psi_Q}~\sigma\otimes \rho)=\tr(\tr_{[1:n]}(\ketbra{\Psi_Q}~\sigma\otimes I)\rho)$. One can verify that $\tr_{[1:n]}(\ketbra{\Psi_Q}~\sigma\otimes I)=\frac{1}{2^n}Q\sigma^T Q$. So the Bell measurement on $\sigma\otimes \rho$ is equivalent to the single-copy POVM $\{\frac{1}{2^n}Q\sigma^T Q\}_{Q\in \cP_n}$ on $\rho$.
\end{remark}

\section{Pauli shadow tomography without quantum memory}
In this section, we consider Pauli shadow tomography without quantum memory (i.e. ancilla-free protocols), in which one is only allowed to perform a quantum measurement on a single copy of the unknown quantum state and collect the classical outcomes. An example protocol in this setting is the classical shadows protocol~\cite{huangPredictingManyProperties2020}, which performs random measurements and constructs a classical description of the unknown quantum state $\rho$ (the classical shadow). Given $M$ Pauli strings, it is proved that $\log M 2^n/\epsilon^2$ samples are sufficient to estimate the expectation values within additive error $\epsilon$. This protocol is proven to be optimal in the worst-case setting without quantum memory~\cite{chen2022exponential}, i.e., for any $M$, there exists a set $A\subseteq\cP_n$ with $\abs{A}=M$ such that any algorithm without quantum memory requires $\Omega(\min\{M,2^n\}/\epsilon^2)$ copies of $\rho$. 

However, the above lower bound only holds for the worst-case choice of $A$. It is natural to ask if we can characterize the fine-grained sample complexity given any choice of $A\subseteq\cP_n$ with $\abs{A}$. In the following part, we provide an affirmative response to the above question by proposing a tight bound for Pauli shadow tomography of arbitrary $A\subseteq\cP_n$ without quantum memory.

We will first prove \Cref{thm:PauliShadowNoMem}.
This is a simplification of \Cref{thm:general_theorem} in the case of memory-free algorithms (i.e., $(1, \cM_n)$ protocols).
It is unclear yet how to solve the optimization problem in \eqref{eq:DeltaNoMem} for general $A$. Here we exhibit the calculation for some special cases, including $\{X, Y, Z\}^{\otimes n}$, unions of Pauli families, and noncommuting Pauli strings. 
We also find that if we restricted the measurements to Clifford measurements, the sample complexity has a clear combinatorial meaning called fractional coloring number.

\subsection{Pauli shadow tomography of any subset without quantum memory}\label{sec:PauliShadowNo}
Recall the definition of $\delta_A$:
\begin{equation*}
    \delta_A\coloneqq \min_{\pi\in \cD(A)}\max_{\ket{\psi}}\E_{P\sim \pi}\braket{\psi|P|\psi}^2.
\end{equation*}
\Cref{thm:PauliShadowNoMem} states that the sample complexity of Pauli shadow tomography without quantum memory is lower bounded by $\Omega(1/(\epsilon^2\delta_A))$ and upper bounded by $O(\log\abs{A}/(\epsilon^2\delta_A))$. 

\begin{proof}
    Applying \Cref{thm:general_theorem} to $(1, \cM_{n})$ protocols, we know that the sample complexity is lower bounded by $\Omega(1/\delta_{1, \cM_n}(A))$ and upper bounded by $O(\log\abs{A}/\delta_{1, \cM_n}(A))$. Therefore, we only need to prove that $\delta_{1, \cM_n}(A)=9\epsilon^2\delta_A$. By definition of $\delta_{1, \cM_n}(A)$ and $\delta_A$, we only need to prove that for every distribution $\pi$ on $A$,
    \begin{equation}
        \max_{M=\{F_s\}_s\in \cM_n}\E_{P\sim \pi} \sum_s\frac{\tr(F_{s}P)^2}{2^{n}\tr(F_s)} = \max_{\ket{\psi}}\E_{P\sim \pi}\braket{\psi|P|\psi}^2. \label{eq:delta_no_memory_equal_to_(c, cM_n)}
    \end{equation}
    On the one hand, for $M'=\{F_s'\}_s$ in $\cM_n$ that maximizes the left-hand side, we write the spectrum decomposition of $F_s'$ as $F_s' = \sum_i \lambda_{s, i}\ketbra{\psi_{s, i}}$, where $\lambda_{s, i}\ge 0$. Then $\sum_{s, i}\lambda_i=\sum_s \tr(F_s')=2^n$ and
    \begin{align}
        \max_{M=\{F_s\}_s\in \cM_n}\E_{P\sim \pi}\sum_s \frac{\tr(F_{s}P)^2}{2^{n}\tr(F_s)}&=\E_{P\sim \pi}\sum_s \frac{\tr(F_{s}'P)^2}{2^{n}\tr(F_s')} \nonumber\\
        &= \frac{1}{2^n}\E_{P\sim \pi}\sum_{s}\frac{(\sum_i \lambda_{s, i}\braket{\psi_{s, i}|P|\psi_{s, i}})^2}{\sum_i \lambda_{s, i}}\nonumber \\
        &\leq \frac{1}{2^n}\E_{P\sim \pi}\sum_{s}\sum_i \lambda_{s, i}\braket{\psi_{s, i}|P|\psi_{s, i}}^2\nonumber \\
        &\leq \frac{1}{2^n}\sum_{s, i}\lambda_{s, i}\max_{\ket{\psi}}\E_{P\sim \pi}\braket{\psi|P|\psi}^2\nonumber \\
        &= \max_{\ket{\psi}}\E_{P\sim \pi}\braket{\psi|P|\psi}^2, \label{eq:delta_no_memory_equal_to_(c, cM_n)1}
    \end{align}
    where we use the Cauchy-Schwarz inequality in the third line. 
    On the other hand, let $\ket{\psi}$ be any state that maximizes the right-hand side, define a POVM $M_{\psi}\coloneqq \{\frac{1}{2^n}Q\ketbra{\psi}Q\}_{Q\in \cP_n}\in \cM_n$. We can check that $M_{\psi}$ is a valid POVM using \Cref{lem:SumPauliExp2}. Then
    \begin{align}
        \max_{M=\{F_s\}_s\in \cM_n}\E_{P\sim \pi}\sum_{s}\frac{\tr(F_sP)^2}{2^n\tr(F_s)}&\ge \E_{P\sim \pi}\delta(M_\psi, P) \nonumber\\
        &= \E_{P\sim \pi}\sum_{Q\in \cP_n}\frac{1}{2^{2n}}\tr(Q\ketbra{\psi}Q P)^2\nonumber \\
        &= \E_{P\sim \pi}\sum_{Q\in \cP_n}\frac{1}{2^{2n}}\braket{\psi| P |\psi}^2\nonumber \\
        &= \E_{P\sim \pi}\braket{\psi| P |\psi}^2\nonumber\\
        &= \max_{\ket{\psi}}\E_{P\sim \pi}\braket{\psi|P|\psi}^2. \label{eq:delta_no_memory_equal_to_(c, cM_n)2}
    \end{align}
    Combining \eqref{eq:delta_no_memory_equal_to_(c, cM_n)1} and \eqref{eq:delta_no_memory_equal_to_(c, cM_n)2}, we have \eqref{eq:delta_no_memory_equal_to_(c, cM_n)}. Therefore, $\delta_{1, \cM_n}(A)=9\epsilon^2\delta_A$ and the theorem is proved by \Cref{thm:general_theorem}.

    However, from the proof above, it is still unclear how to optimally learn $\tr(P\rho)$ without memory. Here we give an explicit algorithm using $O(\log\abs{A}/(\epsilon^2\delta_A))$ samples in \Cref{fig:algorithm_no_memory}. The proof of correctness is similar to the upper bound of \Cref{thm:general_theorem}. Indeed, by \Cref{lem:SumPauliExp2} we can calculate that 
    \begin{equation*}
        \E_{Q_t}[\epsilon_{P, Q_t}|l_t]=\sum_{Q}\frac{1}{2^n}\tr(Q\ketbra{\psi_{l_t}}Q\rho)\epsilon_{P, Q_t}=\sum_{Q}\frac{1}{2^n}\tr(Q\ketbra{\psi_{l_t}}QP\rho P)=\braket{\psi_{l_t}|P|\psi_{l_t}}\tr(P\rho).
    \end{equation*}
    Therefore, $\hat{E}_P$ is an unbiased estimator of $\tr(P\rho)$. The Hoeffding's inequality gives
    \begin{equation*}
        \Pr_{(Q_t)_{t=1}^T}\left[\abs{\hat{E}_P-\tr(P\rho)}\ge \epsilon~|~(l_t)_{t=1}^T\right]\leq 2\exp(-\frac{1}{2}\epsilon^2 \sum_{t=1}^T\braket{\psi_{l_t}|P|\psi_{l_t}}^2). 
    \end{equation*}
    Since $\E_{l\sim \{q_l\}}\braket{\psi_l|P|\psi_l}^2\ge \delta_A/2$, by \Cref{lem:Chernoff}, we have
    \begin{equation*}
        \Pr_{(l_t)_{t=1}^T}\left[\sum_{t=1}^T\braket{\psi_{l_t}|P|\psi_{l_t}}^2\leq \frac{\delta_A T}{4}\right]\leq \exp(-\frac{T\delta_A}{16}).
    \end{equation*}
    We say $(l_t)_{t=1}^T$ is good if $\sum_{t=1}^T\braket{\psi_{l_t}|P|\psi_{l_t}}^2\ge \delta_A T/4$. 
    \begin{align}
        \Pr[\abs{\hat{E}_P-\tr(P\rho)}\ge \epsilon]&\leq \Pr[\abs{\hat{E}_P-\tr(P\rho)}\ge \epsilon~|~(l_t)_{t=1}^T\text{ is good}] + \Pr[(l_t)_{t=1}^T\text{ is not good}]\nonumber\\
        &\leq 2\exp(-\frac{\epsilon^2 T\delta_A}{8})+\exp(-\frac{T\delta_A}{16})\leq 3\exp(-\frac{T\delta_A\epsilon^2}{16}).\nonumber
    \end{align}
    Setting $T=16\log(30\abs{A})/(\epsilon^2\delta_A)$ and using union bound, the success probability is at least $2/3$.
    \begin{algorithm}[ht]
        \centering
        \fbox{\parbox{0.92\textwidth}{
            \textbf{Require:} $\epsilon>0$, $A\subseteq \cP_n$, copies of unknown state $\rho$.\\
            \textbf{Ensure:} Use $T=O(\log\abs{A}/(\epsilon^2\delta_A))$ single-copy measurements. Output $\hat{E}_P$ for every $P\in A$ such that $\abs{\hat{E}_P-\tr(P\rho)}\leq \epsilon$. Succeed with probability at least $2/3$.
            \begin{enumerate}
                \item Find a finite ensemble $\{\ket{\psi_l}\text{ w.p. }q_l\}$ such that $\min_{P\in A}\E_{l\sim \{q_l\}}\braket{\psi_l|P|\psi_l}^2\ge \delta_A/2$.
                \item In round $t\in [T]$, sample a $l_t$ from $\{q_l\}$. Measure $\rho$ using the POVM $\{\frac{1}{2^n}Q\ketbra{\psi_{l_t}}Q\}_{Q\in \cP_n}$. Denote the outcome by $Q_t$.
                \item For every $P\in A$, output $\hat{E}_P=\frac{\sum_{t=1}^T \braket{\psi_{l_t}|P|\psi_{l_t}}\epsilon_{P, Q_t}}{\sum_{t=1}^T \braket{\psi_{l_t}|P|\psi_{l_t}}^2}$.
            \end{enumerate}
        }}
        \caption{Algorithm for Pauli shadow tomography without quantum memory in \Cref{thm:PauliShadowNoMem}.}
        \label{fig:algorithm_no_memory}
    \end{algorithm}
\end{proof}

It is sometimes useful to rewrite $\delta_A$ using the minimax theorem
\begin{theorem}\label{thm:maximin}
    Let $\Sigma_{\rm sep}\coloneqq \{\sum_{l=1}^{r}w_l \ket{\psi_l}\bra{\psi_l}\otimes \ket{\phi_l}\bra{\phi_l}: r\in \N, w_l\ge 0, \sum_{l}w_l=1, \ket{\psi_l}, \ket{\phi_l}\in \cH_n\}$ be the set of separable states and $\Sigma_{\rm sep, sym}\coloneqq \{\sum_{l=1}^{r}w_l \ket{\psi_l}\bra{\psi_l}\otimes \ket{\psi_l}\bra{\psi_l}: r\in \N, w_l\ge 0, \sum_{l}w_l=1, \ket{\psi_l}\in \cH_n\}$ be the set of symmetric separable states. Then
    \begin{equation*}
        \delta_A =  \max_{\rho\in \Sigma_{\rm sep, sym}}\min_{P\in A}\tr(P\otimes P~\rho) = \max_{\rho\in \Sigma_{\rm sep}}\min_{P\in A}\tr(P\otimes P~\rho).
    \end{equation*}
\end{theorem}
\begin{proof}
    The first equality is from \Cref{lem:minimax}, where we set $B=\{\ketbra{\psi}\otimes \ketbra{\psi}\}$. Then $B$ is compact and $\cD(B)$ is equivalent to $\Sigma_{\rm sep, sym}$. Since $\Sigma_{\rm sep, sym}$ is compact, the ``$\sup$'' in \Cref{lem:minimax} is achieved and can be replaced by ``$\max$''. 
    
    Now we prove the second equality. Since $\Sigma_{\rm sep, sym}\subseteq \Sigma_{\rm sep}$, the middle quantity is at most the rightmost quantity. On the other hand, let $\rho'\in \Sigma_{\rm sep}$ be the state that maximizes the third quantity. Write $\rho'=\sum_{l=1}^r w_l \ketbra{\psi_l}\otimes \ketbra{\phi_l}$ and define $\rho''=\sum_{l=1}^r w_l/2 (\ketbra{\psi_l}\otimes \ketbra{\psi_l}+\ketbra{\phi_l}\otimes \ketbra{\phi_l})$. For any $P\in A$, $\braket{\psi|P|\psi}\braket{\phi|P|\phi}\leq (\braket{\psi|P|\psi}^2+\braket{\phi|P|\phi}^2)/2$, so $\tr(P\otimes P~\rho')\leq \tr(P\otimes P~\rho'')$.
    \begin{align*}
        \max_{\rho\in \Sigma_{\rm sep}} \min_{P\in A}\tr(P\otimes P~\rho) &= \min_{P\in A}\tr(P\otimes P~\rho')\leq \min_{P\in A}\tr(P\otimes P~\rho'')\leq \max_{\rho\in \Sigma_{\rm sep, sym}}\min_{P\in A}\tr(P\otimes P~\rho).
    \end{align*}
    Therefore, the three quantities are equal.
\end{proof}

\noindent We remark that there is also a simple bound on $\delta_A$ that holds in general.

\begin{corollary}\label{thm:size_of_A}
    The sample complexity of Pauli shadow tomography for $A$ without quantum memory is at least $\Omega(\abs{A}/(\epsilon^2 2^n))$. The factor $1/2^n$ is tight.
\end{corollary}
\begin{proof}
    The lower bound is from \Cref{thm:PauliShadowNoMem} and the fact that
    \begin{equation*}
        \delta_A\leq \max_{\ket{\psi}\in \cH_n}\frac{1}{\abs{A}}\sum_{P\in A}\braket{\psi|P|\psi}^2\leq \max_{\ket{\psi}\in \cH_n}\frac{1}{\abs{A}}\sum_{P\in \cP_n}\braket{\psi|P|\psi}^2=\max_{\ket{\psi}\in \cH_n}\frac{2^n}{\abs{A}}\braket{\psi\psi|\SWAP|\psi\psi}= \frac{2^n}{\abs{A}},
    \end{equation*}
    where we have used \Cref{lem:SumPauli}.

    When $A$ is a union of $m$ disjoint Pauli families (see \Cref{sec:PauliShadowNoFamily} for details), $\abs{A}=m(2^n-1)+1$ and $\delta_A=1/m$, so the factor $1/2^n$ cannot be improved.
\end{proof}

\subsection{Pauli shadow tomography with Clifford measurements without quantum memory}\label{sec:PauliShadowNoClifford}
In the NISQ era, it is difficult to implement arbitrary measurements in $\cM_n$. 
One advantage of~\Cref{thm:general_theorem} is that it allows us to reason about arbitrary closed POVM sets.
In this section, we will use this to our advantage to get a result in a more experimentally friendly setting, where the measurements are restricted to Clifford measurements. 
We find that in this setting, the sample complexity has a clear combinatorial meaning in terms of the so-called \emph{fractional coloring number} of a certain graph depending on $A$, which we defined shortly.
Before we state the result, let us briefly introduce stabilizer states, Clifford measurements, graphs, and the fractional coloring number of a graph.

It can be shown that the largest set of commuting Pauli strings is of size $2^n$, e.g., $\{I, Z\}^{\otimes n}$. A set of $2^n$ commuting Pauli strings is called a \textit{Pauli family} or a \textit{maximal stabilizer group}. Every set of commuting Pauli strings is contained in a Pauli family. The common eigenstates of a Pauli family form an orthonormal basis, called a \textit{stabilizer basis}, e.g., the common eigenbasis of $\{I, Z\}^{\otimes n}$ is the computational basis. A state is called a \textit{stabilizer state} if it belongs to a stabilizer basis. A POVM is called a \textit{Clifford measurement} if it is $\{\ketbra{\psi_i}|i\in [2^n]\}$
where $\{\ket{\psi_i}\}$ is a stabilizer basis. It gets its name because a Clifford measurement is equivalent to a Clifford gate followed by a computational basis measurement, but we will not introduce the concept of Clifford gates here. Denote the set of all Clifford measurements on $n$ qubits by $\cM_{n, \text{Clifford}}$.

A \textit{graph} $G=(V, E)$ consists of a set of vertices $V$ and a set of edges $E\subseteq V\times V$. In this paper, we only consider undirected graphs so that $(u, v)$ and $(v, u)$ are the same edge. A coloring of $G$ is an assignment of colors to the vertices such that adjacent vertices have different colors. The \textit{coloring number} or \textit{chromatic number} of $G$, denoted by $\zeta(G)$, is the smallest number of colors needed to color $G$. We can also define the coloring number through independent sets. An \textit{independent set} of $G$ is a subset of vertices such that no two vertices are adjacent. Let $\cI_G$ be the set of all independent sets. Then the coloring number of $G$ is the minimal $k$ such that there exists a partition of $V$ into $k$ independent sets $(I_i\in \cI_G)_{i\in [k]}$. In fractional coloring, we allow the partition to be fractional. In other words, we partition $V$ into $(I_i\in \cI_G\text{ with weight }w_i)$ such that the total appearance of every $v$ (i.e., $\sum_{i:v\in I_i}w_i$) is at least $1$. The \textit{fractional coloring number} is the smallest $\sum_{i}w_i$ such that the fractional coloring exists. It is clearer to define it via linear programming.
The fractional coloring number of $G$, denoted by $\zeta_f(G)$, is the solution of the following linear programming:
\begin{align*}
    \text{minimize }&\sum_{B\in \cI_G}x_B,\\
    \text{subject to }& \sum_{B\in \cI_G: v\in B}x_B\ge 1, \forall v\in V,\\
    & x_B\ge 0, \forall B\in \cI_G.
\end{align*}
Notice that if we restrict $x_B$ to be integers, then the solution of this integer linear programming is just the coloring number. Therefore, the fractional coloring number is a relaxation of the coloring number.

Now we are ready to state our result.
\begingroup
\def\thetheorem{\ref{thm:Clifford}}
\begin{theorem}
    Let $A\subseteq \cP_n$ be a set of Pauli strings. Let $G_A$ be the graph of commutation, i.e., two vertices $P, P'\in A$ are connected by an edge if and only if $P$ and $P'$ are anticommuting.
    Then $\delta_{1, \cM_{n, \text{Clifford}}}(A)=9\epsilon^2/\zeta_f(G_A)$. Therefore, the sample complexity of Pauli shadow tomography for $A$ without quantum memory and with Clifford measurements is lower bounded by $\Omega(\zeta_f(G_A)/\epsilon^2)$ and upper bounded by $O(\log\abs{A}\zeta_f(G_A)/\epsilon^2)$.
\end{theorem}
\addtocounter{theorem}{-1}
\endgroup
\begin{proof}
    Let $\mathrm{Stab}$ be the set of all stabilizer states. For $\ket{s}\in \mathrm{Stab}$, let $F$ be the Pauli family such that $\ket{s}$ is a common eigenstate of $F$. Then $\braket{s|P|s}^2=1$ for all $P\in F$. An important fact is that $\braket{s|Q|s}^2=0$ for all $Q\not\in F$. Indeed, since a Pauli family is a maximal set of commuting Pauli strings, there exists a $P\in F$ anticommuting with $Q$. Then $\braket{s|Q|s}=-\braket{s|PQP|s}=-\braket{s|Q|s}$, implying $\braket{s|Q|s}=0$. Therefore, for any Clifford measurement $M=\{\ketbra{\psi_i}|i\in [2^n]\}$ with the corresponding Pauli family $F$, and any Pauli string $Q$, $\braket{\psi_i|Q|\psi_i}^2=0$ if $Q\not\in F$ and $\braket{\psi_i|P|\psi_i}^2=1$ if $P\in F$. Hence $\chi^2_M(\rho_P\|\rho_m)=9\epsilon^2\times \mathbbm1[\text{$P\in F$}]$ and
    \begin{equation}
        \delta_{1, \cM_{n, \text{Clifford}}}(A) = 9\epsilon^2\min_{\pi\in \cD(A)}\max_{B\in \cI_{G_A}}\E_{P\sim \pi}\mathbbm1[\text{$P\in B$}]. \label{eq:clifford2}
    \end{equation}
    Let $\delta$ be the solution of the min-max in \eqref{eq:clifford2}.Then $\delta_{1,\cM_{n, \text{Clifford}}}(A) = 9\epsilon^2 \delta$ and $\delta$ is the solution of the following linear programming:
    \begin{align*}
        \text{minimize }&\delta\\
        \text{subject to }& \sum_{P\in B}\pi_P\leq \delta, \forall B\in \cI_{G_A},\\
        & \pi_P\geq 0, \forall P\in A.\\
        & \sum_{P\in A}\pi_P=1.
    \end{align*}
    Rescaling $\pi_P$, we can see that $1/\delta$ is the solution of the following linear programming:
    \begin{align*}
        \text{maximize }&\sum_{P\in A}\pi_P\\
        \text{subject to }&\sum_{P\in B}\pi_P\leq 1, \forall B\in \cI_{G_A},\\
        & \pi_P\geq 0, \forall P\in A. 
    \end{align*}
    The dual of this linear programming is exactly the definition of fractional coloring number. Therefore, $1/\delta=\zeta_f(G_A)$ and $1/\delta_{1, \cM_{n, \text{Clifford}}}(A)=\zeta_f(G_A)/(9\epsilon^2)$. The sample complexity is then proved by \Cref{thm:general_theorem}.
\end{proof}

The optimal protocol using $O(\log{\abs{A}}\zeta_f(G_A)/\epsilon^2)$ is similar to \Cref{fig:algorithm_no_memory}, but can be explained in a more combinatorial way.
It first finds an optimal fractional coloring of $G_A$, i.e., a set of independent sets $I_i\in\cI(G_A)$ with weights $w_i$ such that $\sum_i w_i=\zeta_f(G)$ and $\sum_{i:P\in I_i}w_i\ge 1,~\forall P\in A$. Each $I_i$ is a commuting set and can be expanded into a Pauli family $F_i$. The protocol performs the Clifford measurement corresponding to $F_i$ with frequency proportional to $w_i$. 

When $w_i\in \{0, 1\}$, this protocol reduces to the folklore protocol that partitions $A$ into several commuting sets and measures each set separately. 
In the optimal partition, the number of parts is the coloring number $\zeta(G_A)$ by definition. 
Therefore, our algorithm reduces the sample complexity from the coloring number $\zeta(G_A)$ to the fractional coloring number $\zeta_f(G_A)$. This saves an at most $O(\log\abs{A})$ factor since $\zeta(G)/(1+\ln\abs{G})\leq \zeta_f(G)\leq \zeta(G)$ \cite{lovaszRatioOptimalIntegral1975}. Our result shows that this protocol (as well as the folklore protocol) is indeed optimal up to $\polylog(\abs{A})$ without magic.

\subsection{Calculation for some specific Pauli sets}\label{sec:ApplicationNo}
In this section, we calculate or bound the sample complexity for some specific Pauli sets.

\subsubsection{Unions of disjoint Pauli families}\label{sec:PauliShadowNoFamily}
As a warm-up, we consider the case when $A=\cup_{i=1}^m S_i$ is a union of $m$ disjoint Pauli families $S_1, \cdots, S_m$. This is the simplest example and we can calculate $\delta_A$ exactly.

Recall that a Pauli family is a set of $2^n$ commuting Pauli strings. 
We say two Pauli families are disjoint if they only share the trivial Pauli string $I^{\otimes n}$. Therefore, $A$ contains a trivial Pauli string and $m(2^n-1)$ nontrivial Pauli strings. 
Since all Pauli strings in each Pauli family are simultaneously measurable, there is a simple algorithm that learns $A$ using roughly $m$ samples.
By \Cref{thm:PauliShadowNoMem}, we can prove that this is nearly optimal.

\begin{theorem}\label{thm:UnionDisjoint}
    Let $A=\cup_{i=1}^m S_i$ be a union of $m$ disjoint Pauli families $S_1, \cdots, S_m$. Then $\delta_A=\frac{1}{m}$, so the sample complexity of Pauli shadow tomography for $A$ without quantum memory is between $\Omega(m/\epsilon^2)$ and $O(mn\log(m)/\epsilon^2)$. 
\end{theorem}
\begin{proof}
    On the one hand, for any distribution $\pi\in \cD(A)$, there must be an $i$ such that $\sum_{P\in S_i}\pi_P$ is at least $1/m$. Let $\ket{\psi_i}$ be any common eigenstate of $S_i$, then $\braket{\psi|P|\psi}^2=1$ for $P\in S_i$. So
    \begin{align*}
        \max_{\ket{\psi}}\sum_{P\in A}p_a\braket{\psi|P|\psi}^2\ge \sum_{P\in A}\pi_P\braket{\psi_i|P|\psi_i}^2\ge \sum_{P\in S_i}\pi_P\ge \frac{1}{m}.
    \end{align*}
    Since this is true for any distribution $\pi\in\cD(A)$, we have $\delta_A\ge 1/m$.

    On the other hand, Let $\pi'\in \cD(A)$ be the uniform distribution on nontrivial Pauli strings in $A$, then
    \begin{align*}
        \delta_A&\leq \max_{\ket{\psi}}\sum_{P\in A}\pi'_P\braket{\psi|P|\psi}^2\\
        &\leq\max_{\ket{\psi}}\frac{1}{m(2^n-1)}\sum_{P\in \cP_n,P\neq I^{\otimes n}}\braket{\psi|P|\psi}^2\\
        &=\max_{\ket{\psi}}\frac{1}{m(2^n-1)}\braket{\psi\psi|(2^n\SWAP_n-I^{\otimes n})|\psi\psi}=\frac{1}{m},
    \end{align*}
    where we used \Cref{lem:SumPauli} in the third line.
    Therefore, $\delta_A=1/m$.
    By \Cref{thm:PauliShadowNoMem}, the sample complexity is between $\Omega(m/\epsilon^2)$ and $O(m\log(m(2^n-1)+1)/\epsilon^2)=O(mn\log(m)/\epsilon^2)$.
\end{proof}

\subsubsection{Noncommuting Pauli strings}\label{sec:PauliShadowNoNonCommute}
The next example is the set of noncommuting Pauli strings.
Because we cannot simultaneously measure any two of them, it seems that we cannot do better than just measuring each Pauli string one by one.
Using \Cref{thm:PauliShadowNoMem}, we can prove that this trivial algorithm is nearly optimal.

\begin{theorem}\label{thm:Noncommuting}
    Let $A=\{P_1, \cdots, P_m\}$ be a set of $m$ noncommuting Pauli strings. Then the sample complexity of Pauli shadow tomography for $A$ without quantum memory is $\tilde{\Theta}(m/\epsilon^2)$.
\end{theorem}
\begin{proof}
    On the one hand, for any distribution $\pi\in\cD(A)$, there must be an $P'\in A$ such that $\pi_{P'} \ge 1/m$. Let $\ket{\psi}$ be any $+1$ eigenstate of $P'$, then $\braket{\psi|P'|\psi}^2=1$. So
    \begin{align*}
        \max_{\ket{\psi}}\sum_{P\in A}\pi_P\braket{\psi|P|\psi}^2\ge \pi_{P'}\ge \frac{1}{m}.
    \end{align*}
    Since this is true for any distribution $\pi\in\cD(A)$, we have $\delta_A\ge 1/m$.

    On the other hand, we can prove that for any $\ket{\psi}$, $\sum_{i=1}^m\braket{\psi|P_i|\psi}^2\leq O(\log m)$. Take arbitrary $\epsilon_1, \cdots, \epsilon_m\in \{-1, 1\}$. Using the fact that $P_iP_j+P_jP_i=0$, we can calculate $(\sum_{i}\epsilon_i P_i)^2=mI$. So all eigenvalues of $\sum_{i}\epsilon_iP_i$ are $\pm \sqrt{m}$, implying $\sum_{i=1}^{m}\epsilon_i\braket{\psi|P_i|\psi}\leq \sqrt{m}$.
    As it holds for arbitrary $\epsilon$'s, we obtain $\sum_{i=1}^{m}\abs{\braket{\psi|P_i|\psi}}\leq \sqrt{m}$. Without loss of generalization, assume $\abs{\braket{\psi|P_i|\psi}}$, $i=1, 2, \cdots, m,$ are in the descending order. Then $\abs{\braket{\psi|P_m|\psi}}\leq 1/\sqrt{m}$. Similarly we obtain that $\abs{\braket{\psi|P_i|\psi}}\leq 1/\sqrt{i}$. Hence $\sum_{i=1}^m\braket{\psi|P_i|\psi}^2\leq \sum_{i=1}^m 1/i =O(\log m)$ and 
    \begin{align*}
        \delta_A\leq \max_{\ket{\psi}}\sum_{P\in A}\frac{1}{m}\braket{\psi|P|\psi}^2\leq O\left(\frac{\log m}{m}\right).
    \end{align*}
    Therefore $\delta_A=\tilde{\Theta}(1/m)$. By \Cref{thm:PauliShadowNoMem}, the sample complexity is $\tilde{\Theta}(m/\epsilon^2)$.
\end{proof}

Note that our $\tilde{\Omega}(m)$ lower bound is not contradicting with the $\tilde{O}(2^n\log(m))$ upper bound in the classical shadow paper \cite{huangPredictingManyProperties2020} because the largest sets of $n$-qubit noncommuting Pauli strings only have size $m=2n+1$, see, e.g., \cite[Proposition 9]{hrubesFamiliesAnticommutingMatrices2016}.

\subsubsection{\texorpdfstring{$\{X,Y,Z\}^{\otimes n}$}{Lg}}\label{sec:PauliShadowNoXYZ}
Now we consider a concrete Pauli string set $A=\{X, Y, Z\}^{\otimes n}$, i.e., the set of Pauli strings that are nontrivial on every qubit. Previously, \cite{ippoliti2024classical} proved an upper bound on the sample complexity for this task of $O((3/2)^n / \epsilon^2)$. Here we prove that this tight and give an alternative estimator as a corollary of Theorem~\ref{thm:Clifford}.

\begin{theorem}\label{thm:PauliShadowXYZ}
    The sample complexity of Pauli shadow tomography for $\{X, Y, Z\}^{\otimes n}$ without quantum memory is $\tilde{\Theta}((3/2)^n/\epsilon^2)$
\end{theorem}
\begin{proof}
    The lower bound is from \Cref{thm:size_of_A}. We will prove the upper bound by \Cref{thm:Clifford}. Let $G$ be the graph of commutation of $\{X, Y, Z\}^{\otimes n}$. Define $B = \{(X, Y), (Y, Z), (Z, X)\}$. For $b\in B$, we use $b(1), b(2)$ to denote the first and the second element in $b$, respectively.
    For any $b_1, \cdots, b_n\in B$, define two Pauli string sets
    \begin{equation*}
        C_{\text{even}}(b_{[1:n]})\coloneqq \{\otimes_{i=1}^n b_i(w_i)|w_i\in\{1, 2\}, 2|w_1+\cdots+w_n\}, 
    \end{equation*}
    \begin{equation*}
        C_{\text{odd}}(b_{[1:n]})\coloneqq \{\otimes_{i=1}^n b_i(w_i)|w_i\in\{1, 2\}, 2\nmid w_1+\cdots+w_n\}.
    \end{equation*}
    It is easy to verify that all elements in $C_{\text{even}}(b_{[1:n]})$ (and $C_{\text{odd}}(b_{[1:n]})$) are commuting, so they are independent sets of $G$. When $b_{[1:n]}$ goes though $B^{n}$, we will obtain $2\times 3^n$ independent sets and every Pauli string in $\{X, Y, Z\}^{\otimes n}$ is in $2^n$ of them. We assign weight $1/2^n$ to each independent set and they form a fractional coloring of $G$. Therefore, $\zeta_f(G)\leq 2(3/2)^n$ and the sample complexity is upper bounded by $\Omega(n(3/2)^n/\epsilon^2)$ according to \Cref{thm:Clifford}.
\end{proof}

\section{Pauli shadow tomography with bounded quantum memory}\label{sec:PauliShadowBoundedMemory}
\subsection{Pauli shadow tomography with bounded quantum memory}\label{sec:Protocol_bounded_memory}
In this section, we consider Pauli shadow tomography with $k$ qubits of quantum memory. To begin with, we show how to identify learning protocols with bounded quantum memory with a special case of $(c, \cM)$ protocols.

The class of $c$-copy learning algorithms with $k$ qubits of quantum memory is equivalent to $(c, \cM_{c, n}^k)$ protocols for a certain POVM set $\cM_{c,n}^k$. To see what $\cM_{c, n}^k$ is, let us delve into each round of the algorithm carefully.

The algorithm first implements a $(n\to k)$-POVM $M=\{N_{s_1}^\dagger N_{s_1}\}_{s_1}$ to $\rho$. Upon the outcomes $s_1$, it stores the post-measurement state $\sigma_{s_1}\coloneqq N_{s_1}\rho N_{s_1}^\dagger/ \tr(N_{s_1}\rho N_{s_1}^\dagger)$ into quantum memory. Then the algorithm performs a $(n+k\to k)$ POVM $M_{s_1}=\{N_{s_1, s_2}^\dagger N_{s_1, s_2}\}_{s_2}$ to $\sigma_{s_1}\otimes \rho$ and stores the $k$-qubit post-measurement state to the memory $\sigma_{s_1, s_2}=N_{s_1, s_2}(\sigma_{s_1}\otimes \rho)N_{s_1, s_2}^\dagger/\tr(N_{s_1, s_2}(\sigma_{s_1}\otimes \rho)N_{s_1, s_2}^\dagger)$, and so forth. Finally, given previous outcomes $s_1, \cdots s_{c-1}$, the algorithm performs a $(n+k\to 0)$ POVM $M_{s_{[1:c-1]}}=\{N_{s_{[1:c-1]}, s_c}^\dagger N_{s_{[1:c-1]}, s_c}\}_{s_c}$ to $\sigma_{s_{[1:c-1]}}\otimes \rho$. Using the tensor network diagram, we can represent the probability of outcome $s_1,\cdots, s_c$ (here we draw the $c=4$ case for simplicity):
\begin{equation}
    \begin{quantikz}
        &\setwiretype{n}&\gate[2]{N_{s_1, s_2}}&\gate[1]{N_{s_1}\vphantom{N_{s_1}^\dagger}}\setwiretype{c}&\gate[1]{\rho}\setwiretype{q}&\gate[1]{N_{s_1}^\dagger}\gategroup[4,steps=4,style={dashed,rounded corners,inner sep=6pt},label style={label position=below, yshift=0.1cm}]{$\ket{\tilde{L}_{s_{[1:c]}}}$}&\gate[2]{N_{s_1, s_2}^\dagger}\setwiretype{c}\\
        &\gate[2]{N_{s_1,s_2,s_3}}\setwiretype{n}&\setwiretype{c}&\setwiretype{q}&\gate[1]{\rho}&&&\gate[2]{N_{s_1,s_2,s_3}^\dagger}\setwiretype{c}\\
        \gate[2]{N_{s_1,s_2,s_3,s_4}}&\setwiretype{c}&\setwiretype{q}&&\gate[1]{\rho}&&&&\gate[2]{N_{s_1,s_2,s_3,s_4}^\dagger}\setwiretype{c}\\
        &&&&\gate[1]{\rho}&&&&
    \end{quantikz} \label{eq:tensor_network_MPS}
\end{equation}
The part in the dashed box is an unnormalized $(2^n, 2^k, c)$ MPS, denoted by $\ket{\tilde{L}_{s_{[1:c]}}}$. Here we use tilde to emphasize that the MPS is not normalized. The probability of outcome $s_1,s_2\cdots, s_c$ is 
\begin{equation*}
    p_{s_1,\cdots, s_c} = \braket{\tilde{L}_{s_1,\cdots, s_c}|\rho^{\otimes c}|\tilde{L}_{s_1,\cdots, s_c}}.
\end{equation*}
Therefore, all $c$ measurements in one round can be represented by a single POVM $L\coloneqq \{\ketbra{\tilde{L}_{s_1,\cdots, s_{c}}}\}$. We denote the set of all such POVMs by $\cM_{c, n}^k$.
Then the class of $(c, \cM_{c, n}^k)$ learning protocols is the class of $c$-copy learning algorithms with $k$ qubits of quantum memory.

\noindent The POVM set $\cM_{c, n}^k$ is obviously closed, Pauli-closed, and contains all single-copy measurements. Hence, \Cref{thm:general_theorem} implies
\begin{theorem}\label{thm:PauliShadowK}
    The sample complexity of Pauli shadow tomography for $A\subseteq \cP_n$ using $c$-copy measurements with $k$ qubits of quantum memory is
    \begin{itemize}
        \item lower bounded by $\Omega(c/\delta_{c, \cM_{c, n}^k}(A))$,
        \item upper bounded by $O(2^{O(c)}c\log(\abs{A})/\delta_{c, \cM_{c, n}^k}(A)+\log(\abs{A})/\epsilon^4)$.
    \end{itemize}
\end{theorem}

\subsection{Pauli shadow tomography of all Pauli strings with bounded quantum memory}\label{sec:PauliShadowKAll}
We now consider Pauli shadow tomography in \Cref{prob:PauliShadow} for the whole set $\cP_n=\{I,X,Y,Z\}^n$. It is proved that $\tilde{\Theta}(2^n/\epsilon^2)$ copies of $\rho$ are necessary and sufficient any algorithm without quantum memory~\cite{huang2021information,huangPredictingManyProperties2020,chen2022exponential}. When we have bounded quantum quantum memory of $k$ qubits, $\Omega(2^{(n-k)/3})$ copies are necessary for Pauli shadow tomography of $\cP_n$ to constant accuracy~\cite{chen2022exponential}. When we have $n$ qubits of quantum memory,~\cite{huang2021information,chen2022exponential} provides a two-copy protocol that can learn the absolute value $\abs{\tr(P_a\rho)}_{a=1}^{4^n}$ within $\epsilon$ using $O(n/\epsilon^4)$ copies. However, learning the signs such that we can fully address Pauli shadow tomography for $\cP_n$ requires protocol $c=O(n/\epsilon^2)$-copy measurements according to the previous literature. There also exists a gap between the upper bound of $O(1/\epsilon^4)$ and the lower bound of $\Omega(1/\epsilon^2)$ concerning the dependence on $\epsilon$. In addition, the previous results provide no algorithms for a fine-grained quantum memory trade-off for Pauli shadow tomography of $\cP_n$ with bounded quantum memory of $k$ qubits. 

In this section, we provide the proof for \Cref{thm:PauliShadowAllK}, which provide affirmative answers for all the questions above. For convenience, we restate and reorganize these results in \Cref{thm:TwoCopyKMemAllUpper} and \Cref{thm:TwoCopyKMemAllLower}, and provide the proof for each part sequentially. 

\subsubsection{The protocol using two-copy measurements}
As a warmup, we first consider the case when $k=n$ and $c=2$, where we can learn all Pauli strings efficiently using two-copy measurements, (i.e., two-copy measurements with $n$ qubits of quantum memory, see \Cref{def:ModelCCopyKMem} and \Cref{def:ModelCCopy} for details). We give the following theorem which directly follows from Bell measurements and \Cref{thm:learn_sign_from_absolute}. We provide a proof for concreteness.

\begin{theorem}\label{thm:TwoCopyAllUpper}
Consider the Pauli shadow tomography task defined in \Cref{prob:PauliShadow}. Given an arbitrary set $A\subseteq\cP_n$, there exists a protocol solving Pauli shadow tomography of $A$ with a high probability using $n$ qubits of quantum memory (two-copy measurements) and $O(\log(\abs{A})/\epsilon^4)$ copies of $\rho$.
\end{theorem}

\noindent In particular, when $A=\cP_n$, the sample complexity required to estimate all Pauli strings is $O(n/\epsilon^4)$. We remark that increasing the value of $c$ for protocols using $c$-copy measurements at $c>2$ will not reduce the $n$ dependence of the sample complexity as the number of copies required for Pauli shadow tomography of any $A\subseteq\cP_n$ is $\Omega(\log(|A|))$ within constant error from an information-theoretic perspective~\cite{huangPredictingManyProperties2020,aaronson2018shadow}.

\begin{proof}
We consider the two-stage process to solve this task. We first learn the absolute values $\abs{\tr(P_a\rho)}$ for all of Pauli strings $P_a\in A$~\cite{huang2021information,chen2022exponential}. We set the quantum memory with a state $\rho$ using teleportation and set the working qubits in $\rho$. We then measure the two copies $\rho\otimes\rho$ using Bell basis measurements (see theorem 6.7,~\cite{chen2022exponential}) given that $P_a\otimes P_a$ are simultaneously diagonalized under this basis. It is proved that $O(\log(\abs{A})/\epsilon^4)$ copies are enough to estimate $\tr(P_a\rho)^2$ for all $P_a\in A$ within $\epsilon^2$, which estimates $\abs{\tr(P_a\rho)}$ within $\epsilon$, with probability at least $9/10$.

For learning signs, we exploit \Cref{thm:learn_sign_from_absolute} at set $\epsilon$ in \Cref{thm:learn_sign_from_absolute} to be $\epsilon/3$. We can thus learn all signs with a probability of at least $9/10$. The total success probability over the two stages is bounded below by $(9/10)^2\geq 2/3$, which finishes the proof for this theorem.
\end{proof}

\subsubsection{The protocol using \texorpdfstring{$2$}{Lg}-copy measurements with \texorpdfstring{$k$}{Lg} qubits of quantum memory}

Now, we establish a fine-grained quantum memory trade-off for Pauli shadow tomography of $\cP_n$ using $2$-copy measurements with $k\leq n$ qubits of quantum memory by the following theorem.

In this section, we will index $\cP_n$ by $\Z_2^{2n}$. We regard every \emph{index} $a\in\Z_2^{2n}$ as a $2n$-bit classical string $a=a_{x,1}a_{z,1}\cdots a_{x,n}a_{z,n}$. The corresponding Pauli operator reads $P_a=\bigotimes_{k=1}^ni^{a_{x,k}a_{z,k}}X^{a_{x,k}}Z^{a_{z,k}}$ with the phase chosen to ensure Hermiticity. We also define a symplectic inner product $\langle\cdot,\cdot\rangle$ via $\langle a,b\rangle \coloneqq \sum_{k=1}^n(a_{x,k}b_{z,k}+a_{z,k}b_{x,k}) \mod{2}$. 

\begin{theorem}\label{thm:TwoCopyKMemAllUpper}
Consider the Pauli shadow tomography task defined in \Cref{prob:PauliShadow}. There exists a protocol solving Pauli shadow tomography of $\cP_n$ with high probability using $2$-copy measurements with $k\leq n$ qubits of quantum memory and $O(n2^{n-k}/\epsilon^4)$ copies of $\rho$. 
\end{theorem}

\begin{proof}
We first consider learning the absolute value $\abs{\tr(P_a\rho)}$ for $P_a\in\cP_n$. We extend the two-qubit Bell basis defined in \Cref{sec:Prelim_QI} to a multi-qubit Bell basis. We denote a Bell state $\ket{\Psi_v}=P_v\otimes I\ket{\Psi^+}$ on the $2k$-qubit subsystem with $\ket{\Psi^+}=\frac{1}{2^k}\sum_{i=1}^{2^k}\ket{i}\ket{i}$. A stabilizer group $S$ on the $(n-k)$-qubit subsystem is a set of $2^{n-k}$ commuting Pauli operators. $S$ can be regarded as a $(n-k)$-dimensional subspace of a $\mathbb{Z}_2^{2(n-k)}$. The stabilizer state of this group is a simultaneous eigenstate of all Pauli strings in $S$. Mathematically, we can write the Bell state and the stabilizer state as follows:
\begin{align*}
\ket{\Psi_v}\bra{\Psi_v}&=\frac{1}{4^k}\sum_{u\in\mathbb{Z}_2^{2k}}(-1)^{\expval{u,v}}P_u\otimes P_u^\top,\forall v\in \mathbb{Z}_2^{2k}\\
\ket{\phi_e^S}\bra{\phi_e^S}&=\frac{1}{2^{n-k}}\sum_{s\in S}(-1)^{\expval{e,s}}P_s,\forall e\in S^\perp.
\end{align*}
The stabilizer covering is a set of stabilizer groups such that every Pauli string belongs to at least one of them. In particular, we can always find a stabilizer covering of size $2^{n-k}+1$ for $(n-k)$-qubit Pauli strings~\cite{lawrence2002mutually,wootters1989optimal,jena2019partitioning}. 

Given a $n$-qubit quantum state $\rho$, we can always write it in the form $\rho=\frac{1}{2^n}\sum_{a\in\mathbb{Z}_2^{2n}}\lambda_aP_a$. We now consider the following procedure in \Cref{fig:AllPauliKAlgo} that learns $\abs{\lambda_a}$, which is the absolute values of expectations for all Pauli strings, using $2$-copy measurements with $k\leq n$ qubits of quantum memory: 
\begin{algorithm}[t]
    \centering
    \fbox{\parbox{0.92\textwidth}{
        \textbf{Require:} $\epsilon>0$, copies of an unknown state $\rho$.\\
        \textbf{Ensure:} Use $O(n2^{n-k}/\epsilon^4)$ copies of $\rho$ and output $\hat{E}_a$ for every $P_a\in\cP_n$ such that $\abs{\hat{E}_a-\tr(P_a\rho)}\leq \epsilon$. Succeed with probability at least $2/3$.
        \begin{enumerate}
            \item Pick two copies of $\rho$ and pick a stabilizer group $S$ of $(n-k)$ qubits.
            \item Measure the first $k$ qubits of each copy together on the Bell basis.
            \item Measure the remaining $(n-k)$ qubits for each copy on the stabilizer basis separately.
            \item Repeat the above three steps for $O(n/\epsilon^4)$ times. Obtain estimators of $\abs{\tr(P\rho)}$ for all $P\in \cP_k\otimes S$ within error $\epsilon/3$ with probability at least $1-1/(10\times (2^{n-k}+1))$.
            \item Do the above four steps for all $2^{n-k}+1$ stabilizer group such that any non-identity $(n-k)$-qubit Pauli string belongs to exactly one of it. Obtain estimators of $\abs{\tr(P_a\rho)}$ for all $P_a\in \cP_n$ within error $\epsilon/3$ with probability at least $9/10$.
            \item Apply \Cref{thm:learn_sign_from_absolute} to learn the true values $\tr(P_a\rho)$.
        \end{enumerate}
    }}
    \caption{The protocol using two-copy measurements and $k\leq n$ qubits of quantum memory for learning all Pauli strings in \Cref{thm:TwoCopyKMemAllUpper}.}
    \label{fig:AllPauliKAlgo}
\end{algorithm}

To see that the above process can be implemented using at most $k$ ancilla qubits, we can first implement the $(n-k)$ qubit stabilizer state measurement on the last $(n-k)$ qubits of the first copy. The remaining $k$ qubits of this copy, together with the second copy, are measured using $k$ ancilla qubits. The probability distribution is unchanged as we can always delay the measurement to the end if we don't perform any further operation on the last $(n-k)$ qubits of the first copy afterward.

We now consider the probability of obtaining the measurement outcome $v\in\mathbb{Z}_2^{2k}$ and $e,e'\in S^\perp$. Notice that the two copies can be written as
\begin{align*}
\rho\otimes\rho=\frac{1}{4^n}\sum_{a,a'\in\mathbb{Z}_2^{2n}}\lambda_a\lambda_{a'}P_a\otimes P_{a'}.
\end{align*}
The measurement outcome distribution can be calculated as
\begin{align*}
p(v,e,e')&=\frac{1}{8^n}\tr\left\{\sum_{a,a'\in\mathbb{Z}_2^{2n}}\lambda_a\lambda_{a'}P_a\otimes P_{a'}\cdot \sum_{u\in\mathbb{Z}_2^{2k}}(-1)^{\expval{u,v}}P_u\otimes P_u^\top\otimes \sum_{s\in S}(-1)^{\expval{e,s}}P_s\otimes \sum_{s'\in S}(-1)^{\expval{e',s'}}P_{s'}\right\}\nonumber\\
&=\frac{1}{4^n}\sum_{u\in\mathbb{Z}_2^{2k}}\sum_{s\in S}\sum_{s'\in S}\lambda_{u\oplus s}\lambda_{u\oplus s'}(-1)^{\expval{u,v}}(-1)^{\expval{e,s}}(-1)^{\expval{e',s'}}.
\end{align*}
From the above equation, we can observe that $p(v,e,e')$ and $\lambda_{u\oplus s}\lambda_{u\oplus s'}$ are related by the Hadamard-Walsh transformation, which also indicates that
\begin{align*}
\lambda_{u\oplus s}\lambda_{u\oplus s'}=\sum_{v\in\mathbb{Z}_2^{2k}}\sum_{e,e'\in S^\perp}p(v,e,e')(-1)^{\expval{u,v}+\expval{s,e}+\expval{s',e'}}.
\end{align*}
Thus, $(-1)^{\expval{u,v}+\expval{s,e}+\expval{s,e'}}$ is an unbiased estimation for $\lambda_{u\oplus s}^2$. According to the Hoeffding's inequality, we only need $O(\log(1/\delta)/\epsilon^2)$ to estimate $\lambda_{u\oplus s}^2$ within accuracy $\epsilon$ with high probability $1-\delta$.

In order to learn all the $\abs{\lambda_a}$'s, we have to enumerate all $2^{n-k}+1$ stabilizer groups. As we only get $\abs{\lambda_a}^2$ in the above process, we have to estimate it within accuracy $\epsilon^2$ to guarantee that each $\abs{\lambda_a}$ is estimated within error $\epsilon$. Thus, the total number of copies we require to learn absolute values of all Pauli strings is $\tilde{O}(2^{n-k}/\epsilon^4)$, where $\tilde{O}$ comes from the $O(n)$ overhead to guarantee the success probability.

Finally, we apply \Cref{thm:learn_sign_from_absolute} to learn true values $\tr(P_a\rho)$ from the estimators of absolute values.
\end{proof}

\subsubsection{The lower bound for \texorpdfstring{$c$}{Lg}-copy measurement protocols with \texorpdfstring{$k$}{Lg} qubits of quantum memory}\label{sec:lowerboundck}

\Cref{thm:TwoCopyAllUpper} and \Cref{thm:TwoCopyKMemAllUpper}, together with the fact that ancilla-free protocols can solve the Pauli shadow tomography of $\cP_n$ using $O(2^n/\epsilon^2)$ samples, finish the proof for the upper bound parts of \Cref{thm:PauliShadowAllK}. We next turn to the lower bound part. 

A useful upper bound of the $\delta_{c, \cM}(A)$ (thus a lower bound of the sample complexity) is the following:
\begin{theorem}\label{thm:lowerbound_c_M_simpler}
    Let $c\in \N$, $A\subseteq \cP_n$, and $\cM$ be a closed set of $cn$-qubit POVMs. For any $P\in A$ and $S\subseteq [c]$, we use $P^S$ to denote the $cn$-qubit Pauli string obtained by applying $P$ to each block of $n$ qubits indexed by $S$ and identity to the rest. Then for any $\pi\in \cD(A)$,
    \begin{equation}
        \delta_{c, \cM}(A)\leq \left(\sum_{\emptyset\neq S\subseteq [c]}(3\epsilon)^{\abs{S}}\sqrt{\max_{M=\{F_s\}_s\in \cM}\E_{a\sim \pi}\sum_{s}\frac{\tr(F_sP_a^{S})^2}{2^{cn}\tr(F_s)}}\right)^2 .\label{eq:lowerbound_c_M_simpler}
    \end{equation}
\end{theorem}
\begin{proof}
    Let $w_S$ be the term corresponding to $S$ in \eqref{eq:lowerbound_c_M_simpler}. We have
    \allowdisplaybreaks
    \begin{align}\allowdisplaybreaks
        \delta_{c, \cM}(A)&\leq \max_{M=\{F_s\}_s\in \cM}\E_{a\sim \pi}\chi_M^2(\rho_a^{\otimes c}\|\rho_m^{\otimes c})\nonumber\\
        &= \max_{M=\{F_s\}_s\in \cM}\E_{a\sim \pi}\sum_{s}\frac{\tr(F_s)}{2^{cn}}\left(\frac{\tr(F_s(I+3\epsilon P_a)^{\otimes c})}{\tr(F_s)}-1\right)^2\nonumber\\
        &= \max_{M=\{F_s\}_s\in \cM}\E_{a\sim \pi}\sum_{s}\frac{1}{2^{cn}\tr(F_s)}\left(\sum_{\emptyset\neq S\subseteq [c]}(3\epsilon)^{\abs{S}}\tr(F_sP_a^{S})\right)^2\nonumber\\
        &\leq \max_{M=\{F_s\}_s\in \cM}\E_{a\sim \pi}\sum_{s}\frac{1}{2^{cn}\tr(F_s)}\sum_{\emptyset\neq S\subseteq [c]}\frac{1}{w_S}(3\epsilon)^{2\abs{S}}\tr(F_sP_a^{S})^2\sum_{\emptyset\neq S\subseteq [c]}w_S\nonumber\\
        &= \max_{M=\{F_s\}_s\in \cM}\sum_{\emptyset\neq S\subseteq [c]}w_S\sum_{\emptyset\neq S\subseteq [c]}\frac{1}{w_S}(3\epsilon)^{2\abs{S}}\E_{a\sim \pi}\sum_{s}\frac{\tr(F_sP_a^{S})^2}{2^{cn}\tr(F_s)}\nonumber\\
        &\leq \sum_{\emptyset\neq S\subseteq [c]}w_S\sum_{\emptyset\neq S\subseteq [c]}\frac{1}{w_S}(3\epsilon)^{2\abs{S}}\max_{M=\{F_s\}_s\in \cM}\E_{a\sim \pi}\sum_{s}\frac{\tr(F_sP_a^{S})^2}{2^{cn}\tr(F_s)}\nonumber\\
        &= \left(\sum_{\emptyset\neq S\subseteq [c]}w_S\right)^2.\nonumber
    \end{align}
    Here we use the Cauchy-Schwarz in the fourth line. 
\end{proof}

Therefore, to obtain a lower bound, one only needs to choose a distribution $\pi$ over $A$ and upper bound $\max_{M=\{F_s\}_s\in \cM}\E_{a\sim \pi}\sum_{s}\tr(F_sP_a^{S})^2/(2^{cn}\tr(F_s))$ separately for all $\emptyset\neq S\subseteq [c]$. In our case, we choose $\pi_u$ to be the uniform distribution over $\cP_n$ and prove the following lemma:

\begin{lemma} \label{lem:Bounded_memory_lower_bound}Let $c\in \N$, $S$ be a non-empty subset of $[c]$, we have
    \begin{enumerate}[label=(\arabic*)]
        \item If $\abs{S}=1$, then $\max_{M=\{F_s\}_s\in \cM_{cn}}\E_{P\sim \pi_u}\sum_{s}\tr(F_sP^{S})^2/(2^{cn}\tr(F_s))\leq \frac{1}{2^{n}}$.
        \item If $\abs{S}\ge 2$, then $\max_{M=\{F_s\}_s\in \cM_{cn}}\E_{P\sim \pi_u}\sum_{s}\tr(F_sP^{S})^2/(2^{cn}\tr(F_s))\leq 1$.
        \item If $\abs{S}\ge 2$, then $\max_{M=\{F_s\}_s\in \cM_{c,n}^k}\E_{P\sim \pi_u}\sum_{s}\tr(F_sP^{S})^2/(2^{cn}\tr(F_s))\leq \frac{1}{2^{n-k}}$.
    \end{enumerate}
\end{lemma}
\begin{proof}
    (1) Suppose $S=\{j\}$. For any $M=\{F_s\}_s\in \cM_{cn}$, 
    \begin{align}
        \E_{P\sim \pi_u}\sum_{s}\frac{\tr(F_s P^S)^2}{2^{cn}\tr(F_s)}&=\sum_{s}\frac{\tr(F_s\otimes F_s~ \E_{P\sim \pi_u}P^{\{j\}}\otimes P^{\{j\}})}{2^{cn}\tr(F_s)}\nonumber\\
        &= \frac{1}{2^{n}}\sum_s \frac{\tr(F_s\otimes F_s~\SWAP_{j, j+c})}{2^{cn}\tr(F_s)}\nonumber\\
        &\leq \frac{1}{2^n}\sum_s \frac{\tr(F_s\otimes F_s)}{2^{cn}\tr(F_s)}=\frac{1}{2^n}.\nonumber
    \end{align}
    Here we use \Cref{lem:SumPauli} in the second line. $\SWAP_{j, j+c}$ means the swap gate on the $j$-th qudit and the $(j+c)$-th qudit (each qudit is a $n$-qubit system). The third line is due to the fact that $\tr(AB)\leq \tr(A)\norm{B}$ for any positive semi-definite operator $A$ and Hermitian matrix $B$.\\
    (2) This is trivial since $\tr(F_sP^S)^2\leq \tr(F_s)^2$.\\
    (3) Any $M\in \cM_{c, n}^k$ can be written as $\{2^{cn}w_s\ketbra{L_s}\}_s$, where each $\ket{L_s}\in \mathrm{MPS}_c(2^n, 2^k)$ is a normalized MPS, and $\sum_s w_s=1$. Then $\E_{P\sim \pi_u}\sum_{s}\tr(F_sP^{S})^2/(2^{cn}\tr(F_s))=\sum_s \E_{P\sim \pi_u} w_s\braket{L_s|P|L_s}^2$. So we only need to prove that $\E_{P\sim \pi_u}\braket{\psi|P|\psi}^2\leq 2^{k-n}$ for any $\ket{\psi}\in \mathrm{MPS}_c(2^n, 2^k)$.  

    Let $l$ be the smallest element in $S$. By \Cref{lem:MPS_Schmidt}, we can write its Schmidt decomposition with respect to the cut $[1:l]\cup [l+1:c]$ as
    \begin{equation*}
        \ket{\psi}=\sum_{i=1}^{2^k}\sqrt{\lambda_i}\ket{\alpha_i}\otimes\ket{\beta_i}, 
    \end{equation*}
    where $\{\ket{\alpha_i}\}_{i\in [2^k]}$ ($\{\ket{\beta_i}\}_{i\in [2^k]}$) are orthonormal states defined on the first $l$ (last $c-l$) qudits, and $\sum_{i=1}^{2^k}\lambda_i=1$. Then 
    \begin{align}
        \E_{P\sim \pi_u}\braket{\psi|P^{S}|\psi}^2&=\frac{1}{4^n}\sum_{P\in \cP_n}\braket{\psi|P^{S}|\psi}^2\nonumber\\
        &\leq \frac{1}{4^n}\sum_{P\in \cP_n}\left(\sum_{i,j\in[2^k]}\abs{\sqrt{\lambda_i\lambda_j}\braket{\alpha_i|P^{\{l\}}|\alpha_j}\braket{\beta_i|P^{S\backslash\{l\}}|\beta_j}}\right)^2\nonumber\\
        &\leq \frac{1}{4^n}\sum_{P\in \cP_n}\left(\sum_{i,j\in [2^k]}\lambda_i\lambda_j\abs{\braket{\alpha_i|P^{\{l\}}|\alpha_j}}^2\right)\left(\sum_{i,j\in [2^k]}\abs{\braket{\beta_i|P^{S\backslash\{l\}}|\beta_j}}^2\right)\label{eq:PauliShadowKMemAllLower2}.
    \end{align}
    We can expand $\{\ket{\beta}\}_{i\in [2^k]}$ to a orthonormal basis $\{\ket{\beta}\}_{i\in [2^{n(c-l)}]}$ on the last $c-l$ qudits. Then
    \begin{align}
        \sum_{i,j\in [2^k]}\abs{\braket{\beta_i|P^{S\backslash\{l\}}|\beta_j}}^2\leq \sum_{i\in [2^k]}\sum_{j\in [2^{n(c-l)}]} \abs{\braket{\beta_i|P^{S\backslash\{l\}}|\beta_j}}^2=\sum_{i\in [2^k]}\norm{P^{S\backslash \{l\}}\ket{\beta_i}}^2=\sum_{i\in [2^k]}1=2^k. \label{eq:PauliShadowKMemAllLower3}
    \end{align}
    By \eqref{eq:PauliShadowKMemAllLower2} and \eqref{eq:PauliShadowKMemAllLower3} and \Cref{lem:SumPauli}, we have
    \begin{align}
        \E_{P\sim \pi_u}\braket{\psi|P^{S}|\psi}^2&\leq \frac{2^k}{4^n}\sum_{P\in \cP_n}\sum_{i,j\in [2^k]}\lambda_i\lambda_j\abs{\braket{\alpha_i|P^{\{l\}}|\alpha_j}}^2 \nonumber\\
        &= \frac{2^k}{4^n}\sum_{i,j\in [2^k]}\lambda_i\lambda_j \sum_{P\in \cP_n}\braket{\alpha_i\alpha_j|P^{\{l\}}\otimes P^{\{2l\}}|\alpha_j\alpha_i}\nonumber\\
        &= \frac{2^k}{2^n}\sum_{i,j\in [2^k]}\lambda_i\lambda_j \braket{\alpha_i\alpha_j|\SWAP_{l, 2l}|\alpha_j\alpha_i}\nonumber\\
        &\leq \frac{2^k}{2^n}\sum_{i,j\in [2^k]}\lambda_i\lambda_j=2^{k-n}. \nonumber
    \end{align}
\end{proof}

We are ready to prove the lower bound of Pauli shadow tomography with bounded memory. 

\begin{theorem}\label{thm:TwoCopyKMemAllLower}
Consider the Pauli shadow tomography task defined in \Cref{prob:PauliShadow}, any $c$-copy protocol with $k$ qubits of quantum memory requires $\Omega(\min\{2^n/(c\epsilon^2),2^{n-k}e^{-6c\epsilon}/(c^3\epsilon^4)\})$ copies of $\rho$ to solve the Pauli shadow tomography of $\cP_n$ with a high probability.
\end{theorem}
\begin{proof}
    By \Cref{thm:lowerbound_c_M_simpler} and \Cref{lem:Bounded_memory_lower_bound}, we have
    \begin{align}
        \delta_{c, \cM_{c, n}^k}(\cP_n)&\leq \left(\sum_{\emptyset\neq S\subseteq [c]}(3\epsilon)^{\abs{S}}\sqrt{\max_{M=\{F_s\}_s\in \cM_{c, n}^k}\E_{P\sim \pi_u}\sum_{s}\frac{\tr(F_sP^{S})^2}{2^{cn}\tr(F_s)}}\right)^2\nonumber\\
        &\leq \left(3c\epsilon  2^{-n/2}+((1+3\epsilon)^c-1-3c\epsilon)2^{-(n-k)/2}\right)^2\nonumber\\
        &\leq \left(3c\epsilon 2^{-n/2}+9\epsilon^2c^2e^{3c\epsilon}2^{-(n-k)/2}\right)^2.\nonumber\\
        &\leq 18c^2\epsilon^22^{-n}+162\epsilon^4c^4 e^{6c\epsilon}2^{k-n}.\nonumber
    \end{align}
    In the third line, we use $(1+x)^c-1-cx=\sum_{i=2}^c \binom{c}{i}x^i\leq c^2\sum_{i=2}^c\binom{c-2}{i-2}x^i=c^2x^2(1+x)^{c-2}\leq c^2x^2e^{cx}$. The last line is the Cauchy-Schwarz inequality $(a+b)^2\leq 2a^2+2b^2$. By \Cref{thm:PauliShadowK}, the sample complexity is lower bounded by
    \begin{equation*}
        \Omega\left(\min\left\{\frac{2^n}{c\epsilon^2}, \frac{2^{n-k}e^{-6c\epsilon}}{c^3\epsilon^4}\right\}\right).
    \end{equation*}
    When $c=O(1)$, the lower bound is $\Omega(\min\{2^n/\epsilon^2,2^{n-k}/\epsilon^4\})$. When $c=O(\log n)$, $e^{-6c\epsilon}$ will give some $1/\poly(n)$ factor. We remark that the lower bound is still meaningful when $c=\poly(n)$ and $\epsilon=O(1/c)$, then $e^{-6c\epsilon}=O(1)$ and the lower bound becomes $\Omega(\min\{2^n/\epsilon^2,2^{n-k}/\epsilon^4\}/\poly(n))$. 
\end{proof}
In \Cref{thm:TwoCopyKMemAllLower}, $k$ is at most $n$ or slightly larger than $n$. Here we turn our attention to the case when $k$ is much larger. Indeed, when $k\ge (c-1)n$, $\cM_{c, n}^k$ is just $\cM_{cn}$. We show that the $\Omega(1/\epsilon^4)$ lower bound still holds in this case.
\begin{theorem}\label{thm:ccopyLower}
    Consider the Pauli shadow tomography take defined in \Cref{prob:PauliShadow}. Any $c$-copy protocol requires $\Omega(\min\{2^n/(c\epsilon^2),e^{-6c\epsilon}/(c^3\epsilon^4)\})$ copies of $\rho$ to solve the Pauli shadow tomography of $\cP_n$ with a high probability.
\end{theorem}
\begin{proof}
    Similar to the proof of \Cref{thm:TwoCopyKMemAllLower}, by \Cref{thm:lowerbound_c_M_simpler} and \Cref{lem:Bounded_memory_lower_bound}, \begin{align}
        \delta_{c, \cM_{cn}}(\cP_n)&\leq \left(\sum_{\emptyset\neq S\subseteq [c]}(3\epsilon)^{\abs{S}}\sqrt{\max_{M=\{F_s\}_s\in \cM_{cn}}\E_{P\sim \pi_u}\sum_{s}\frac{\tr(F_sP^{S})^2}{2^{cn}\tr(F_s)}}\right)^2\nonumber\\
        &\leq \left(3c\epsilon  2^{-n/2}+(1+3\epsilon)^c-1-3c\epsilon\right)^2\nonumber\\
        &\leq 18c^2\epsilon^22^{-n}+162\epsilon^4c^4 e^{6c\epsilon}.\nonumber
    \end{align}
    By \Cref{thm:general_theorem}, the sample complexity is lower bounded by $\Omega(\min\{2^n/(c\epsilon^2),e^{-6c\epsilon}/(c^3\epsilon^4)\})$.
\end{proof}

\section{Oblivious Pauli shadow tomography is hard}\label{sec:Oblivious}
In this section, we consider obliviously Pauli shadow tomography and provide the proof for \Cref{thm:Oblivious}. Different from the above results where the set of Pauli strings $A$ is known before we measure the unknown state, we are now required to make measurements to get a classical description of the unknown quantum state (also known as the classical shadow~\cite{huangPredictingManyProperties2020}). Then we are required to make predictions for given Pauli strings using the classical description without access to any quantum data. In this paper, we focus on the oblivious Pauli shadow tomography when we are required to predict the expectation value of one Pauli string using the classical description. For this problem, we show below it is as hard as the Pauli shadow tomography $\cP_n$ containing all Pauli strings:

\begin{theorem}\label{thm:Oblivious}
    Suppose there is an oblivious protocol for estimating a single Pauli observable. That is, the protocol takes as input $T$ copies of $\rho$, performs measurements on these copies, and outputs a classical description which with probability $2/3$ satisfies the following: given one arbitrary observable $P\in\cP_n$, one can estimate $\tr(P\rho)$ to within additive error $\epsilon$ using the classical description.

    Then for any $A\subseteq\cP_n$, there is an algorithm that repeats this protocol for $\log(\abs{A})$ times (thus using $O(\log(\abs{A})T)$ copies of $\rho$) and uses the resulting classical outputs to estimate $\{\tr(P\rho)\}_{P\in A}$ to within additive error $\epsilon$, with probability $2/3$.
\end{theorem}

\begin{proof}[Proof of \Cref{thm:Oblivious}]
Assume there exists an algorithm that uses $T$ samples of $\rho$ to construct a classical function $A^{\rho}$ for the oblivious Pauli shadow tomography task. Then we have
\begin{align}\label{eq:ObliviousGoal}
    \Pr[\abs{A^{\rho}(P)-\tr(\rho P)}\ge \epsilon] \leq \frac{1}{3},\ \forall P\in\cP_n,\  \rho\in\mathbb{C}^{2^n\times 2^n}
\end{align}
as we require this classical function to be able to estimate the expectation value of an arbitrarily chosen Pauli string with high probability. 

On the contrary, suppose we have an algorithm that outputs a function $B^{\rho}$ for the non-oblivious Pauli shadow tomography of $\cP_n$ that succeeds with a high probability. We have  
\begin{align}\label{eq:NonObliviousGoal}
\Pr[\abs{B^{\rho}(P)-\tr(\rho P)}\ge \epsilon,\ \forall P\in\cP_n] \leq \frac{1}{3},\ \forall \rho\in\mathbb{C}^{2^n\times 2^n}.
\end{align}

To fill the gap between \eqref{eq:NonObliviousGoal} and \eqref{eq:ObliviousGoal}, we need to amplify the success probability of $A^{\rho}$ from $2/3$ to $1-\frac{1}{3(4^n-1)}$ and then use the union bound. We provide the following lemma to achieve the amplification.
\begin{lemma}\label{lem:success_probability_amplification}
Let $x$ be an unknown real number and $A$ be a black box. Each call to $A$ returns a sample $y$ such that $\abs{y-x}\leq \epsilon$ with probability at least $2/3$. Then there exists an algorithm $B$ that calls $A$ at most $10n$ times and outputs an estimate $\hat{x}$ such that $\abs{\hat{x}-x}\leq \epsilon$ with probability at least $1-\frac{1}{3(4^n-1)}$.
\end{lemma}
\begin{proof}
The algorithm $B$ is simple: it calls $A$ for $10n$ times and outputs the median of the samples. If the median is not in $[x-\epsilon, x+\epsilon]$, then there are at least half of the samples outside $[x-\epsilon, x+\epsilon]$. By the Chernoff bound, the failure probability is at most $1-\frac{1}{3(4^n-1)}$.
\end{proof}
\noindent Now for every $P\in\cP_n$, we repeat $A^{\rho}$ for $10n$ times and take the median as the estimate of $\tr(\rho P)$. The success probability is at least $1-\frac{1}{3(4^n-1)}$. By union bound, the success probability of learning all Pauli strings is at least $\frac{2}{3}$. To provide a proof for \Cref{thm:Oblivious}, we don't even need to amplify the success probability $1-\frac{1}{3(4^n-1)}$. Instead, the amplification to $1-\frac{1}{3|A|}$ using $5\log|A|$ calls to $A^\rho$ suffices.
\end{proof}

On the one hand, the above result indicates the hardness of oblivious shadow tomography in the ancilla-free setting. In particular, we consider the task of learning \textbf{one} Pauli string $P$ without memory. A non-oblivious algorithm is aware of $P$ so it can easily learn it with $O(1/\epsilon^2)$ samples. However, an oblivious algorithm does not know $P$. It can only apply some measurements and collect the outcomes beforehand. Then it has to store in the classical description the information to predict every possible $P$ using the outcomes of the measurements. Intuitively, it is as hard as learning all Pauli strings, which requires $\Omega(2^n/\epsilon^2)$ samples. 

On the other hand, \Cref{thm:Oblivious} also indicates that the Pauli shadow tomography task is much easier than the general shadow tomography problem. As the sample complexity for Pauli shadow tomography with unbounded quantum memory is $O(n/\epsilon^4)$, the bits of information required to store the information regarding expectation values of all Pauli strings for any quantum state is linear in $n$ and polynomial in $1/\epsilon$. However, information required to store expectation values of general observables even for pure states is $\exp(\Theta(n))$ bits~\cite{grier2022sample}. 

\section{Purity testing with bounded quantum memory}\label{sec:Purity}
In this section, we turn attention to the purity testing problem. In this task, we are given access to an unknown quantum state $\rho$ and we need to distinguish whether $\rho$ is a pure state or a mixed state. 
We will first give a $2$-copy protocol with $k$ qubits of quantum memory for purity testing with $O(\min(2^{n-k}, 2^{n/2}))$ samples. Then we prove this is optimal using martingale arguments, similar to the lower bound proof of Pauli shadow tomography. This shows a different transition from the Pauli shadow tomography task: when $k\leq n/2$, the quantum memory is not helpful, and when $k$ increases from $n/2$ to $n$, the sample complexity smoothly reduces from $O(2^{n/2})$ to $O(1)$.

\subsection{Upper bound}\label{sec:PurityUpper}
If $k<n/2$, we can directly use the memory-free algorithm in \cite{chen2022exponential} that takes $O(2^{n/2})$ samples:
\begin{lemma}[See \cite{chen2022exponential}, Theorem 5.13]
    There is a learning algorithm without quantum memory that takes $T=O(2^{n/2})$ copies of $\rho$ to distinguish between whether $\rho$ is a pure state or maximally mixed. 
\end{lemma}
Therefore, to prove the upper bound of \Cref{thm:Purity}, we provide a $2$-copy protocol with $k$ qubits of quantum memory using $O(2^{n-k})$ samples when $k\ge n/2$ in \Cref{fig:algorithm_purity_k}.

\begin{algorithm}[ht]
        \centering
        \fbox{\parbox{0.92\textwidth}{
            \textbf{Require:} copies of unknown state $\rho$.\\
            \textbf{Ensure:} Use $O(2^{n-k})$ copies to output whether $\rho$ is a maximal mixed state or a pure state
            \begin{enumerate}
    \item Take one copy $\rho$, measure the first $n-k$ qubits in the computational basis. Denote the outcome as $x_1$. Store the remaining $k$ qubits in the memory.
    \item Take another copy $\rho$, measure the first $n-k$ qubits in the computational basis. Denote the outcome as $x_2$.
    \item Apply the swap test to the remaining two $k$-qubit systems. In other words, we measure the $2k$-qubit system with POVM $\{\frac{I+\SWAP}{2}, \frac{I-\SWAP}{2}\}$. Denote the two outcomes as ``symmetry'' and ``anti-symmetry'', respectively.
    \item Repeat the above steps for $10\times 2^{n-k}$ times. If there exists a round in which $x_1=x_2$ and the swap test gives ``anti-symmetry'', then output ``mixed''. Otherwise, output ``pure''.
\end{enumerate}
        }}
        \caption{Algorithm for purity testing with $k\geq n/2$ qubits of quantum memory.}
        \label{fig:algorithm_purity_k}
    \end{algorithm}

When $k=n$, this algorithm reduces to the canonical swap test. We now prove the correctness of this protocol.
\begin{proof}
    We denote the first $n-k$ qubits of the first copy as system $A_1$, the remaining $k$ qubits as system $B_1$. Similarly, we write $A_2, B_2$ for the second copy. The probability of $x_1=x_2$ and ``anti-symmetry'' is
    \begin{equation*}
        p^{\rho} \coloneqq\sum_{x\in \{0, 1\}^{n-k}}\tr(\ketbra{x}_{A_1}\otimes \ketbra{x}_{A_2}\otimes \frac{I-\SWAP_{B_1, B_2}}{2}~\cdot~ \rho\otimes \rho ).
    \end{equation*}
    It is well known that the eigenvalues of $\SWAP_{B_1, B_2}$ are $\pm 1$. The $-1$ eigenspace has dimension $\binom{2^k}{2}$ with a canonical basis $\{(\ket{\psi_{x, y}}\coloneqq \ket{xy}-\ket{yx})/\sqrt{2}~|x,y\in \{0, 1\}^{k}, x<y\}$ (here the order $x<y$ is the lexicographic order). Therefore,
    \begin{equation*}
        p^{\rho} = \sum_{x\in \{0, 1\}^{n-k}}\sum_{y, z\in \{0, 1\}^k, y<z}\tr(\ketbra{x}_{A_1}\otimes \ketbra{x}_{A_2}\otimes \ketbra{\psi_{y, z}}_{B_1, B_2}~\cdot~ \rho\otimes \rho ).
    \end{equation*}
    When $\rho=\rho_m =I/2^n$ is the maximally mixed state, every term is $1/2^{2n}$, so $p^{\rho_m}=2^{n-k}\binom{2^k}{2}/2^{2n}=(2^k-1)/2^{n+1}$. When $\rho=\ketbra{\phi}$ is a pure state, every term is 0, since
    \begin{align*}
        \MoveEqLeft \tr(\ketbra{x}_{A_1}\otimes \ketbra{x}_{A_2}\otimes \ketbra{\psi_{y, z}}_{B_1, B_2}~\cdot~ \ketbra{\phi}\otimes \ketbra{\phi} )\nonumber\\
        &=\braket{\phi|xy}\braket{\phi|xz}\braket{xy|\phi}\braket{xz|\phi}+\braket{\phi|xz}\braket{\phi|xy}\braket{xz|\phi}\braket{xy|\phi}\nonumber\\
        &\qquad\braket{\phi|xy}\braket{\phi|xz}\braket{xz|\phi}\braket{xy|\phi}-\braket{\phi|xz}\braket{\phi|xy}\braket{xy|\phi}\braket{xz|\phi}\nonumber\\
        &=0\,.
    \end{align*}
    Therefore, if $\rho$ is a pure state, our protocol always outputs the correct answer ``pure''. If $\rho$ is mixed, the probability of the incorrect output ``pure'' is at most $(1-(2^k-1)/2^{n+1})^{10\times 2^{n-k}}\leq e^{-5(1-2^{-k})}\leq e^{-5(1-2^{-1/2})}\leq 1/3$, where $k\ge n/2 \ge 1/2$. Therefore, the protocol is correct with probability at least $2/3$, and it takes only $O(2^{n-k})$ samples.
\end{proof}
\subsection{Lower bound}
We now prove the lower bound for purity testing with bounded quantum memory.
We reduce the purity testing task to the following distinguishing task:

\begin{problem}[Many-versus-one distinguishing problem for purity testing] \label{prob:Distinguish_purity}
    We are given access to copies of an unknown quantum state $\rho$. And we need to distinguish the following two cases
        \begin{itemize}
        \item $\rho$ is the maximal mixed state $\rho_m=I/2^n$.
        \item $\rho=\ketbra{\psi}$, where $\ket{\psi}$ is sampled from the Haar measure over pure states.
    \end{itemize}
\end{problem}

If we can solve the purity testing task, then we can solve this distinguishing task by testing the purity of $\rho$ (with the sample number of samples). Therefore, we only need to prove the lower bound for this distinguishing problem. 

Now assume that there is a $2$-copy algorithm with $k$ qubits of quantum memory that solves the many-versus-one distinguishing task with high probability using $T$ rounds (i.e., $2T$ samples). Let $\cT$ be the tree representation of the algorithm. We will show that $T\ge \Omega(\min(2^{n-k}, 2^{n/2}))$.

Fix a $\ell\in \mathrm{leaf}(\cT)$, let $\{e_{u_t, s_t}\}_{t=1}^T$ be the path from root $r$ to leaf $\ell$, where $u_1=r$. The likelihood ratio is 
\begin{equation*}
    L_{\psi}(\ell) = \prod_{t=1}^T \frac{\tr(F_{s_t}^{u_t}\ketbra{\psi}^{\otimes 2})}{\tr(F_{s_t}^{u_t}\rho_m^{\otimes 2})}=\frac{\tr(\otimes_{t=1}^T F_{s_t}^{u_t}~\ketbra{\psi}^{\otimes 2T})}{\tr(\otimes_{t=1}^T F_{s_t}^{u_t}~\rho_m^{\otimes 2T})}.
\end{equation*}

Using Haar integral \Cref{lem:haar_integral} and the Bernoulli's inequality $(1-x)^r\ge 1-rx$ for $x\in [0, 1]$ and $r\ge 1$.
\begin{align}
    L(\ell) &= \E_{\ket{\psi}\sim \mu_H}[L_{\psi}(\ell)] \nonumber\\
    &= \frac{(2^{n})^{2T}}{2^n(2^n+1)\cdots (2^n+2T-1)}\frac{\tr(\otimes_{t=1}^T F_{s_t}^{u_t}~S_{2T})}{\tr(\otimes_{t=1}^T F_{s_t}^{u_t})}\nonumber\\
    &\ge \left(1-\frac{2T}{2^n}\right)^{2T}\frac{\tr(\otimes_{t=1}^T F_{s_t}^{u_t}~S_{2T})}{\prod_{t=1}^T \tr(F_{s_t}^{u_t})}\nonumber\\
    &\ge \left(1-\frac{4T^2}{2^n}\right)\frac{\tr(\otimes_{t=1}^T F_{s_t}^{u_t}~S_{2T})}{\prod_{t=1}^T \tr(F_{s_t}^{u_t})}\label{eq:likelihood_ratio_purity}
\end{align}\\
A key lemma to bound~\eqref{eq:likelihood_ratio_purity} is the following:
\begin{lemma}\label{lem:permutation}
    Let $x, y$ be two positive integers, $\rho_x$ be a mixed state on $x$ qudits (each qudit has local dimension $d$), and $\rho_y$ be a mixed state on $y$ qudits. Then 
    \begin{equation*}
        \tr(\rho_x\otimes \rho_y S_{x+y})\ge \tr(\rho_x S_x)\tr(\rho_y S_y).
    \end{equation*}
\end{lemma}
\noindent We will prove the lemma at the end of this section. Recursively applying the lemma to~\eqref{eq:likelihood_ratio_purity}, we have
\begin{align}
    L(\ell) &\ge \left(1-\frac{4T^2}{2^n}\right)\prod_{t=1}^{T}\frac{\tr(F_{s_t}^{u_t}S_2)}{\tr(F_{s_t}^{u_t})} \label{eq:likelihood_ratio_purity2}
\end{align}

According to \Cref{lem:permutation_operator}, $S_2=I+\SWAP$ is positive semi-definite and $\tr(S_2)=2^n(2^n+1)$. Therefore, $\rho_s\coloneqq S_2/(2^n(2^n+1))$ is a mixed state. Let $\rho_0 = I_{2n}/2^{2n}=\rho_m^{\otimes 2}$ be the maximally mixed state on $2n$ qudits and denote $\alpha=2^{n}/(2^n+1)$. Then we can write \eqref{eq:likelihood_ratio_purity2} as
\begin{align}
    L(\ell) \ge \left(1-\frac{4T^2}{2^n}\right)\alpha^{-T}\prod_{t=1}^{T}\frac{\tr(F_{s_t}^{u_t}\rho_s)}{\tr(F_{s_t}^{u_t}\rho_0)}\ge \left(1-\frac{4T^2}{2^n}\right)\prod_{t=1}^{T}\frac{\tr(F_{s_t}^{u_t}\rho_s)}{\tr(F_{s_t}^{u_t}\rho_0)}.\label{eq:likelihood_ratio_purity2.5}
\end{align}
If we regard the learning tree $\cT$ as a $(1, \cM_{2, n}^k)$ protocol on $2n$ qubits, then $\prod_{t=1}^T \frac{\tr(F_{s_t}^{u_t}\rho_{s})}{\tr(F_{s_t}^{u_t}\rho_0)}$ is the likelihood ratio of distinguishing $\rho_s$ and $\rho_0$. Therefore, we can apply \Cref{lem:MartingaleTrick} to lower bound it. According to \Cref{lem:MartingaleTrick}, we need to bound the concentration of the likelihood ratio in each step, that is,
\begin{align}
    \max_{\{F_s\}_s\in \cM_{2, n}^k}\E_{s\sim \tr(F_s\rho_2)}[(\frac{\tr(F_s\rho_s)}{\tr(F_s\rho_2)}-1)^2]&=\max_{\{F_s\}_s\in \cM_{2, n}^k}\sum_{s}\frac{\tr(F_s)}{2^{2n}}\left(\frac{\tr(F_s\rho_s)}{\tr(F_s\rho_2)}-1\right)^2\nonumber\\
    &\leq \max_{\{F_s\}_s\in \cM_{2, n}^k}\sum_{s}\frac{\tr(F_s)}{2^{2n}}\left(\alpha\frac{\tr(F_sS_2)}{\tr(F_s)}-1\right)^2\nonumber\\
    &\leq 2(\alpha-1)^2 +2\max_{\{F_s\}_s\in \cM_{2, n}^k}\sum_{s}\frac{\tr(F_s(S_2-I))^2}{2^{2n}\tr(F_s)}\nonumber\\
    &= 2\times 2^{-2n}+2\max_{\{F_s\}_s\in \cM_{2, n}^k}\sum_{s}\frac{\tr(F_s\SWAP)^2}{2^{2n}\tr(F_s)}\label{eq:likelihood_ratio_purity3}
\end{align}
where the third line is $(\alpha x-1)^2 = (\alpha(x-1)+(\alpha-1))^2\leq 2(\alpha-1)^2+2\alpha^2(x-1)^2\leq 2(\alpha-1)^2+2(x-1)^2$.
We bound the right-hand side of \eqref{eq:likelihood_ratio_purity3} in the following lemma, and leave the proof to the end of this section.
\begin{lemma}\label{lem:permutation_SWAP}
    \begin{equation}
        \max_{\{F_s\}_s\in \cM_{2, n}^k}\sum_{s}\frac{\tr(F_s\SWAP)^2}{\tr(F_s)}\leq 2^{k+n}. \label{eq:permutation_SWAP}
    \end{equation}
\end{lemma}
\noindent With this lemma, we further bound \eqref{eq:likelihood_ratio_purity3} as
\begin{align*}
    \max_{\{F_s\}_s\in \cM_{c, n}^k}\E_{s\sim \tr(F_s\rho_2)}\left[\left(\frac{\tr(F_s\rho_s)}{\tr(F_s\rho_c)}-1\right)^2\right]&\leq 2\times 2^{-2n} +2\times 2^{k-n}\nonumber
\end{align*}
By \Cref{lem:MartingaleTrick}, there exists a constant $\gamma$ such that
\begin{equation*}
    \Pr_{\ell\in \mathrm{leaf}(\cT), \ell\sim p^{\rho_m}(\ell)}\left[\prod_{t=1}^{T}\frac{\tr(F_{s_t}^{u_t}\rho_s)}{\tr(F_{s_t}^{u_t}\rho_c)}>0.9\right] \ge 0.9-\gamma T \left(2\times 2^{-2n} +2\times 2^{k-n}\right).
\end{equation*}
Combining this with \eqref{eq:likelihood_ratio_purity2.5}, we have
\begin{equation}
    \Pr_{\ell\in \mathrm{leaf}(\cT), \ell\sim p^{\rho_m}(\ell)}\left[L(\ell)>0.9\left(1-\frac{4T^2}{2^n}\right)\right]\ge 0.9-\gamma T \left(2\times 2^{-2n} +2\times 2^{k-n}\right). \label{eq:likelihood_ratio_purity5}
\end{equation}
Assume $T\leq \min\left\{\frac{2^{n/2}}{20}, \frac{2^{n-k}}{400\gamma}\right\}$, we have
\begin{equation*}
    0.9\left(1-\frac{4T^2}{2^n}\right)\ge 0.9\times 0.99\ge 0.89,
\end{equation*}
\begin{equation*}
    \gamma T(2\times 2^{-2n}+2\times 2^{k-n})\leq 0.01.
\end{equation*}
Plugging these into \eqref{eq:likelihood_ratio_purity5}, we finally arrive at
\begin{equation*}
    \Pr_{\ell\in \mathrm{leaf}(\cT), \ell\sim p^{\rho_m}(\ell)}[L(\ell)>0.89]\ge 0.89.
\end{equation*}
By \Cref{lem:one_side_likelihood}, the success probability of $\cT$ is upper bounded by $0.11+0.11=0.22$, less than $2/3$. Therefore, for the success probability to be at least $2/3$, we must have
\begin{equation*}
    T\ge \min\left\{\frac{2^{n/2}}{20}, \frac{2^{n-k}}{400\gamma}\right\} = \Omega(\min\{2^{n/2}, 2^{n-k}\}).
\end{equation*}
This completes the proof of the lower bound for purity testing with bounded quantum memory.

We now prove the two lemmas used in the proof.
\paragraph{Proofs of \Cref{lem:permutation} and \Cref{lem:permutation_SWAP}}
\begin{proof}[Proof of \Cref{lem:permutation}]
    For $0\leq z\leq \min(x, y)$, let $S_{x, y}^z$ be the set of $\pi\in S_{x, y}$ such that $\pi$ crosses $\{1, \cdots, x\}$ and $\{x+1, \cdots, x+z\}$ for exactly $z$ times. In other words, there exist $z$ index $i_1, \cdots, i_z\in \{1, \cdots, x\}$ such that $\pi(i_j)\in \{x+1, \cdots, x+y\}$. We also use $S_{x, y}^z$ to denote the operator $\sum_{\pi \in S_{x, y}^z}\pi$. Since $(S_{x, y}^z)_{z=0}^{\min(x, y)}$ forms a partition of $S_{x, y}$ and $\tr(\rho_x S_x)\tr(\rho_y S_y)=\tr(\rho_x\otimes \rho_y S_{x, y}^0)$, we only need to prove that for every $1\leq z\leq \min(x, y)$, $\tr(\rho_x\otimes \rho_y S_{x, y}^z)\ge 0$.

    The key observation is that every permutation $\pi$ in $S_{x, y}^z$ can be decomposition into three steps: we first apply some permutation on $[1:x]$ and some permutation on $[x+1:x+y]$. Then we swap $[1:z]$ and $[x+1:x+z]$ (i.e., do $z$ swaps $(i, x+i)$ for $1\leq i\leq z$). Finally, we apply a permutation on $[1:x]$ and a permutation on $[x+1:x+y]$. This yields 
    \begin{equation*}
        S_{x, y}^{z} = \frac{1}{(z!)^3(x-z)!(y-z)!}(S_x\otimes S_y) \SWAP_{[1:z], [x+1, x+z]} (S_x\otimes S_y).
    \end{equation*}
    Here the factor $1/(z!)^3(x-z)!(y-z)!$ is due to the repeated counting of the same permutation. The number of repetitions can be calculated from the total number of permutations. Indeed, the number of permutations in $S_{x, y}^z$ is $\binom{x}{z}^2\binom{y}{z}^2(x-z)!z!(y-z)!=(x!)^2(y!)^2/((z!)^3(x-z)!(y-z)!)$. On the other hand, the right hand side has $(x!)^2(y!)^2$ terms, so every permutation in $S_{x, y}^z$ is counted for $(z!)^3(x-z)!(y-z)!$ times.

    For clarity, denote the factor by $f$. Then we have
    \begin{align*}
        \tr(\rho_x\otimes \rho_y S_{x, y}^z) &= f\tr((S_x\rho_x S_x)\otimes (S_y\rho_yS_y)~ \SWAP_{[1:z], [x+1, x+z]})\\
        &= f\tr((\tr_{[z+1:x]}(S_x\rho_x S_x)\otimes \tr_{[x+z+1:x+y]}(S_y\rho_y S_y))~\SWAP_{[1:z], [x+1, x+z]}),
    \end{align*}
    where $\tr_{[z+1, x]}$ ($\tr_{[x+z+1:x+y]}$) means tracing out the qudits in $[z+1, x]$ ($[x+z+1, x+y]$).

    The swap operator is non-negative on the product states. To see this, let $\ket{\psi_x}, \ket{\psi_y}$ be two pure states on qudit $[1:z], [x+1, x+z]$, respectively. Then 
    \begin{equation*}
        \tr(\ketbra{\psi_x}\otimes \ketbra{\psi_y}~\SWAP_{[1:z], [x+1, x+z]}) = \braket{\psi_x|\psi_y}\braket{\psi_y|\psi_x}\ge 0.
    \end{equation*}
    By linearity, the non-negativity also holds for mixed states. Therefore, $\tr(\rho_x\otimes \rho_y S_{x, y}^z)\ge 0$ and the lemma is proved.
\end{proof}

\begin{proof}[Proof of \Cref{lem:permutation_SWAP}]
    For $M\in \cM_{2, n}^k$, we explicitly write it as $\{\ketbra{\tilde{L}_{s_1,s_2}}\}$ according to \eqref{eq:tensor_network_MPS}, where
    \begin{equation*}
        \ket{\tilde{L}_{s_1, s_2}} = \begin{quantikz}
            &\gate[1]{N_{s_1}^\dagger} &\gate[2]{N_{s_1,s_2}^\dagger}\setwiretype{c}\\
            &&
        \end{quantikz},
    \end{equation*}
    $\{N_{s_1}^\dagger N_{s_1}\}_{s_1}$ is a $(n\to k)$ POVM, and $\{N_{s_1, s_2}^\dagger N_{s_1, s_2}\}_{s_2}$ is a $(n+k\to 0)$ POVM for every $s_1$.
    \begin{align*}
        \sum_{s_2}\ketbra{\tilde{L}_{s_1, s_2}} &= \sum_{s_2}\begin{quantikz}
            &\gate[1]{N_{s_1}^\dagger} &\gate[2]{N_{s_1,s_2}^\dagger N_{s_1,s_2}}\setwiretype{c}&\gate[1]{N_{s_1}\vphantom{N_{s_1}^\dagger}}&\setwiretype{q}\\
            &&&&
        \end{quantikz}\nonumber\\
        &=\begin{quantikz}
            &\gate[1]{N_{s_1}^\dagger} &\gate[1]{N_{s_1}\vphantom{N_{s_1}^\dagger}}\setwiretype{c}&\setwiretype{q}\\
            &&&
        \end{quantikz} = (N_{s_1}^\dagger N_{s_1})\otimes I_{2^n}.
    \end{align*}
    Denote $F_{s_1,s_2}\coloneqq\ketbra{\tilde{L}_{s_1, s_2}}$ and $A_{s_1}\coloneqq N_{s_1}^\dagger N_{s_1}$. Then the last equation becomes $\sum_{s_2}F_{s_1, s_2} = A_{s_1}\otimes I_{2^n}$. Since $N_{s_1}$ is $2^k\times 2^n$, the rank of $A_{s_1}$ is at most $2^k$. Let $\Pi_{s_1}$ be the projector onto the subspace spanned by the eigenstates of $A_{s_1}$ with positive eigenvalues. Then $\tr(\Pi_{s_1})\leq 2^k$ and $A=\Pi_{s_1}A_{s_1}=A_{s_1}\Pi_{s_1}$. We will use the following simple fact without proof:
    \begin{lemma}
        Let $A, B\ge 0$ be two positive semi-definite operators. Then $\tr(AB)\ge 0$. Furthermore, if $\tr(AB)=0$, then $AB=0$. 
    \end{lemma}
    According to the lemma and
    \begin{equation*}
        \sum_{s_2}\tr(F_{s_1,s_2}~((I_n-\Pi_{s_1})\otimes I_{n}))=\tr((A_{s_1}\otimes I_{2^n})~((I_n-\Pi_{s_1})\otimes I_{2^n}))=0,
    \end{equation*}
    we have $F_{s_1,s_2}~((I_n-\Pi_{s_1})\otimes I_{n})=0$. So $F_{s_1,s_2}~(\Pi_{s_1}\otimes I_{n})=F_{s_1, s_2}$. Similarly $(\Pi_{s_1}\otimes I_{n})~F_{s_1, s_2}=F_{s_1, s_2}$. 
    We also need a property of the swap operator: $(X\otimes I)\SWAP=\SWAP(I\otimes X)$ for any $X$.
    Now we are ready to bound \eqref{eq:permutation_SWAP}.
    \begin{align}
        \sum_{s_1, s_2}\frac{\tr(F_{s_1, s_2}\SWAP)^2}{\tr(F_{s_1, s_2})} &= \sum_{s_1, s_2}\frac{\tr((\Pi_{s_1}\otimes I_n)F_{s_1, s_2}(\Pi_{s_1}\otimes I_n)\SWAP)^2}{\tr(F_{s_1, s_2})}\nonumber\\
        &= \sum_{s_1, s_2}\frac{\tr((\Pi_{s_1}\otimes I_n)F_{s_1, s_2}(\Pi_{s_1}\otimes I_n)~(\Pi_{s_1}\otimes I_n)\SWAP(\Pi_{s_1}\otimes I_n))^2}{\tr(F_{s_1, s_2})}\nonumber\\
        &= \sum_{s_1, s_2}\frac{\tr((\Pi_{s_1}\otimes I_n)F_{s_1, s_2}(\Pi_{s_1}\otimes I_n)~(I_n\otimes \Pi_{s_1})\SWAP(I_n\otimes \Pi_{s_1}))^2}{\tr(F_{s_1, s_2})}\nonumber\\
        &= \sum_{s_1, s_2}\frac{\tr((\Pi_{s_1}\otimes \Pi_{s_1})F_{s_1, s_2}(\Pi_{s_1}\otimes \Pi_{s_1})\SWAP)^2}{\tr(F_{s_1, s_2})}\nonumber\\
        &\leq \sum_{s_1, s_2}\frac{\tr((\Pi_{s_1}\otimes \Pi_{s_1})F_{s_1, s_2}(\Pi_{s_1}\otimes \Pi_{s_1}))^2}{\tr(F_{s_1, s_2})}\nonumber\\
        &\leq \sum_{s_1, s_2}\tr((\Pi_{s_1}\otimes \Pi_{s_1})F_{s_1, s_2}(\Pi_{s_1}\otimes \Pi_{s_1}))\nonumber\\
        &= \sum_{s_1}\tr((\Pi_{s_1}\otimes \Pi_{s_1})(A_{s_1}\otimes I_n)(\Pi_{s_1}\otimes \Pi_{s_1}))\nonumber\\
        &= \sum_{s_1}\tr(A_{s_1}\otimes \Pi_{s_1})\nonumber\\
        &\leq 2^k\sum_{s_1}\tr(A_{s_1})=2^k\sum_{s_1}\tr(N_{s_1}^\dagger N_{s_1})=2^{k}\tr(I_n)=2^{k+n}\,,\nonumber
    \end{align}
where the third line follows from using $(X\otimes I)\SWAP=\SWAP(I\otimes X)$ (and $(I\otimes X)\SWAP=\SWAP(X\otimes I)$) twice as:
\begin{align*}
(\Pi_{s_1}\otimes I_n)\SWAP(\Pi_{s_1}\otimes I_n)&=\SWAP(I_n\otimes\Pi_{s_1})(\Pi_{s_1}\otimes I_n)\\&=\SWAP(\Pi_{s_1}\otimes I_n)(I_n\otimes\Pi_{s_1})\\&=(I_n\otimes\Pi_{s_1})\SWAP(I_n\otimes\Pi_{s_1}).
\end{align*}
\end{proof}
%%%%%%%%%%%%%%%%%%%%%%%%%%%%%%%%%%%%%%%%%%%%%%%%%%%%%%%%%%%%%

\section*{Acknowledgments}
We thank Scott Aaronson, Dong-Ling Deng, Alexey Gorshkov, Sabee Grewal, Hong-Ye Hu, Hsin-Yuan Huang, Vishnu Iyer, Tongyang Li, Weikang Li, Luke Schaeffer, Sheng-Tao Wang, and Zhihan Zhang for helpful discussions. We also thank Robbie King, David Gosset, Robin Kothari, and Ryan Babbush for kindly agreeing to coordinate submission of their manuscript with ours.

\newpage

\bibliographystyle{MyRefFont}
\bibliography{PauliStringSamp}
\clearpage
\appendix

\section{Martingales and tree representation}\label{sec:Martingale}
In this appendix, we prove the key~\Cref{lem:MartingaleTrick}. We rewrite the lemma here for convenience.
\begin{lemma}
    There exists a constant $c$ such that the following statement holds. Suppose $\cT$ is a tree representation of a learning protocol for the many-versus-one distinguishing problem defined in \Cref{prob:Distinguish}. 
    If there is a $\delta>0$ such that for every node $u$ we have
    \begin{equation}
        \E_{a\sim \{p_a
        \}}\E_{s\sim \{p^{\rho_m}(s|u)\}}[(L_a(u, s)-1)^2]\leq \delta, \label{eq:condition_learning_tree}
    \end{equation}
    then 
    \begin{equation*}
        \Pr_{a, \ell\sim p^{\rho_m}(\ell)}[L_a(\ell)\leq 0.9] \leq 0.1+c\delta T.
    \end{equation*}
    As a conclusion, the success probability of $\cT$ is at most $0.2+c\delta T$.
    In particular, if $\cT$ solves the many-versus-one distinguishing problem with success probability at least $1/3$, then $T\ge \Omega(1/\delta)$. 
\end{lemma}

For a $\ell\in \mathrm{leaf}(\cT)$, let $\{e_{u_t, s_t}\}_{t=1}^T$ be the path from root $r$ to leaf $\ell$, where $u_1=r$. This lemma is a concentration of $L_a(\ell)=\prod_{t=1}^T L_{a}(u_t, s_t)$ given a good concentration of every $L_a(u, s)$. However, trivially applying the usual concentration inequality (e.g., Hoeffding's inequality) to $\ln(L_{a}(u_t, s_t))$ does not work, since the random variables $\ln(L_{a}(u_t, s_t))$ are not independent. Thus, it is more suitable to model the learning process as a stochastic process and use the martingale concentration inequality. Here we provide a minimal introduction to the preliminary.

In this paper, we simply define a stochastic process as a sequence of random variables $X = (X_0, X_1, \cdots)$, where the distribution of $X_t$ depends on the previous random variables $X_0, \cdots, X_{t-1}$. 
Denote $X_{0:t}\coloneqq (X_0, \cdots, X_t)$ and $X_{1:t}\coloneqq (X_1, \cdots, X_t)$. In the following, when we say $\Pr$ or $\E$, we mean the probability or expectation with respect to the randomness of the stochastic process. In this context, a random variable is indeed a function from $X$ to $\R$. A sequence of random variable $Z_t=Z_t(X_{0:t}) (t=1, 2, \cdots)$ is called a submartingale difference sequence with respective to $X$ if $\E[Z_t|X_{0:t-1}]\ge 0$ for all $t$ and $X_{0:t-1}$. The Freedman's inequality gives a strong concentration for the submartingale difference sequence.

\begin{lemma}[Freedman's Inequality]
    Let $Z_t=Z_t(X_{0:t})(t=1,2, \cdots, T)$ be a submartingale difference sequence, i.e., $\E[Z_t~|X_{0:t-1}]\geq 0$. Assume $Z_t\ge -R$ almost surely for some constant $R\ge 0$. Then for any $\eta,\nu>0$ and positive integer $T$,
    \begin{equation*}
        \Pr[\sum_{t=1}^T Z_t \leq -\eta, \sum_{t=1}^{T}\E[Z_t^2|X_{0:t-1}]\leq \nu] \leq \exp(-\frac{\eta^2}{2\nu+2R\eta/3}).
    \end{equation*}
\end{lemma}

We will model the learning tree $\cT$ as a stochastic process. \Cref{lem:MartingaleTrick} is a direct consequence of the following lemma.

\begin{lemma}\label{lem:stochastic_process}
    For every constant $c_1\in (0, 1/3)$ and $c_2>0$, there are constant $c_3, c_4>0$ such that the following statement holds. Let $(X_0, X_1, \cdots)$ be a stochastic process. $Y_t=Y_t(X_{0:t})\in [-1, +\infty) (t=1, 2, \cdots, T)$ are random variables such that
    \begin{equation}
        \E[Y_t|X_{0:t-1}]\ge -\mu,\quad \E[Y_t^2|X_{1:t-1}]\leq \delta, \label{eq:condition_stochastic_process}
    \end{equation}
    for some constant $\delta\ge 0, \mu\in [0, 1]$. For $T\in \N$, define random variable $L_T=(1+Y_1)(1+Y_2)\cdots (1+Y_T)$, then
    \begin{equation*}
        \Pr[L\leq e^{-T\mu}(1-c_1)] \leq c_2e^{c_4T\mu^2}+c_3\delta T.
    \end{equation*}
\end{lemma}

\noindent Assume the lemma is correct. To prove \Cref{lem:MartingaleTrick}, we specify the stochastic processes $X_t$ and $Y_t$:
\begin{itemize}
    \item $X_0$ is $a$, conforming the distribution $\{p_a\}$ in the many-versus-one distinguishing problem.
    \item Running the learning protocol $\cT$ to state $\rho_m$, we will obtain a sequence of edges $\{e_{u_t, s_t}\}_{t=1}^T$ from the root $u_1=r$ to a leaf $\ell$. We define $X_t\coloneqq (u_t, s_t) $ for $t=1, 2, \cdots, T$ and $X_t=0$ for $t>T$.
    \item $Y_t(X_{0:t})\coloneqq L_a(u_t, s_t)-1$. Then $L_T=(1+Y_1)(1+Y_2)\cdots (1+Y_T)=L_a(\ell)$.
\end{itemize}

In this definition, 
\begin{equation*}
    \E[Y_t|X_{0:t-1}] = \sum_{s} p^{\rho_m}(s|u_t)(L_a(u_t, s)-1) = \sum_{s} (p^{\rho_a}(s|u_t)-p^{\rho_m}(s|u_t))=1-1=0,
\end{equation*}
\begin{equation*}
    \E[Y_t^2|X_{1:t-1}] = \E_{a\sim \{p_a
        \}}\E_{s\sim \{p^{\rho_m}(s|u)\}}[(L_a(u, s)-1)^2].
\end{equation*}
Therefore, if $\cT$ satisfies \eqref{eq:condition_learning_tree} with $\mu=0$, then the stochastic processes satisfy \eqref{eq:condition_stochastic_process}. \Cref{lem:stochastic_process} implies
\begin{equation*}
    \Pr[L_T\leq 1-c_1] \leq c_2+c_3\delta T.
\end{equation*}
Notice that $L_T=L_a(\ell)$, so the \Cref{lem:MartingaleTrick} follows by setting $c_1=c_2=0.1$. 

We now prove \Cref{lem:stochastic_process}.
\begin{proof}[Proof of \Cref{lem:stochastic_process}]
    Let $W_{t} \coloneqq \E[Y_t^2|X_{0: t-1}]$ be a random variable depending on $X_0, \cdots, X_{t-1}$. The condition $\E[Y_t^2|X_{1:t-1}]\leq \delta$ implies that $\E_{X_0}[W_{t}]\leq \delta$.

    Without loss of generality, we assume $\E[Y_t|X_{0:t-1}]\leq 0$ for any $X_{0:t-1}$. Otherwise if $\E[Y_t|X_{0:t-1}]>0$ for some $X_{0:t-1}$, we pick all $X_t$ such that $Y_t(X_{0:t})>0$ and replaces the values of $Y_t(X_{0:t})$ by smaller non-negative values such that $\E[Y_t|X_{0:t-1}]$ becomes 0. During this process, $\E[Y_t^2|X_{1:t-1}]$ and $L_T$ will not increase. So the condition \eqref{eq:condition_stochastic_process} still holds and the conclusion becomes stronger. Therefore, we only need to consider the case $\E[Y_t|X_{0:t-1}]\leq 0$.

    We want to apply Freedman's inequality to $\ln(L)=\sum_{t}\ln(1+Y_t)$. However, $\ln(1+Y_t)$ is not well-behaved when $Y_t$ is close to $-1$. So we will consider random variables $\ln(1+Y_t)\mathbbm{1}[Y_t\ge -\epsilon]$ for some constant $\epsilon$ determined later. 
    We can bound the expectation of $\ln(1+Y_t)\mathbbm{1}[Y_t\ge -\epsilon]$ by
    \begin{align*}
        &\E[\ln(1+Y_t)\mathbbm{1}[Y_t\ge -\epsilon]~|X_{0:t-1}]\\
        \ge&\E[(Y_t-Y_t^2)\mathbbm{1}[Y_t\ge -\epsilon]~|X_{0:t-1}]\\
        \ge&\E[(Y_t-Y_t^2)~|X_{0:t-1}]-\E[Y_t\mathbbm{1}[Y_t< -\epsilon]~|X_{0:t-1}]\\
        \ge& -\mu-W_{t}-\E[Y_t^2~|X_{0:t-1}]^{1/2}\cdot \Pr[Y_t< -\epsilon~|X_{0:t-1}]^{1/2}\\
        \ge& -\mu-W_t - W_t^{1/2}\cdot \Pr[Y_t^2>\epsilon^2|X_{0:t-1}]^{1/2}\\
        \ge& -\mu-(1+\epsilon^{-2})W_t,
    \end{align*} 
    where in the first step we used the fact that $\ln(1+z)\ge z-z^2$ for $z\ge -1/2$, in the third step we used Cauchy-Schwarz inequality, and in the last step we used the Markov inequality. Define $Z_t\coloneqq \ln(1+Y_t)\mathbbm{1}[Y_t\ge -\epsilon]+(1+\epsilon^{-2})W_t+\mu$. Then $Z_t$ is a submartingale difference sequence. $Z_t\ge \log(1-\epsilon)\ge -2\epsilon$ given $0\leq \epsilon < 1/2$. By Cauchy-Schwarz inequality
    \begin{align*}
        \E[Z_t^2~|X_{0:t-1}] &\leq 3\E[\ln^2(1+Y_t)\mathbbm{1}[Y_t\ge -\epsilon]+(1+\epsilon^{-2})^{2}W_t^2+\mu^2~|X_{0:t-1}]\\
        &\leq 6W_t+3(1+\epsilon^{-2})^2W_{t}^2+3\mu^2\\
    \end{align*}
    where we used $\ln^2(1+z)\leq 2z^2$ for $z\ge -1/2$.
    By Freedman's inequality, for every $\eta, \nu>0$ and positive integer $T$, we have
    \begin{align*}
        \Pr[\sum_{t=1}^{T}\ln(1+Y_t)\mathbbm{1}[Y_t\ge -\epsilon] + \sum_{t=1}^T(1+\epsilon^{-2})W_t + T\mu\leq -\eta, \sum_{t=1}^{T}(6W_t+3(1+\epsilon^{-2})^2W_t^2)\leq \nu]\\
        \leq \exp(-\frac{\eta^2}{2(\nu+3T\mu^2)+4\epsilon\eta/3}).
    \end{align*}
    We want to upper bound the $\Pr[\ln L_T=\sum_{t=1}^{T}\ln(1+Y_t)\leq -2\eta-T\mu]$. Consider the following events:
    \begin{enumerate}
        \item $Y_t<-\epsilon$ for some $t$.
        \item $\sum_{t=1}^{T}(6W_t+3(1+\epsilon^{-2})^2W_t^2)> \nu$.
        \item $\sum_{t=1}^T(1+\epsilon^{-2})W_t > \eta$.
        \item $\sum_{t=1}^{T}\ln(1+Y_t)\mathbbm{1}[Y_t\ge -\epsilon]+\sum_{t=1}^T(1+\epsilon^{-2})W_t + T\mu\leq -\eta$ and $\sum_{t=1}^{T}(6W_t+3(1+\epsilon^{-2})^2W_t^2)\leq \nu$.
    \end{enumerate}
    We can verify that if all of the four events do not happen, then $\sum_{t=1}^T\ln(1+Y_t)>-2\eta -T\mu$. By union bound, $\Pr[\sum_{t=1}^T\ln(1+Y_t)\leq -2\eta-T\mu]$ is upper bounded by the sum of the probabilities of the four events. We now bound the probabilities of the four events one by one.
    \begin{enumerate}
        \item Fixed any $t$ and realization of $X_1, X_2,\cdots, X_{t-1}$, we have $\Pr[Y_t<-\epsilon|X_{1:t-1}]\leq \Pr[Y_t^2\ge \epsilon^2|X_{1:t-1}]\leq \epsilon^{-2}\delta$. Since this hold for any $X_1,\cdots, X_{t-1}$, we have $\Pr[Y_t<-\epsilon]\leq \epsilon^{-2}\delta$. By union bound, the first event happens with probability at most $\epsilon^{-2}\delta T$.
        \item $\E[W_t]=\E_{X_{1:t-1}}[\E_{X_0}[W_t]]\leq \E_{X_{1:t-1}}[\delta]=\delta $. So $\E[\sum_{t=1}^TW_t]\leq \delta T$. By Markov inequality, with probability $1-\frac{(12+3\epsilon^{-2})\delta T}{v}$, we have $\sum_{t=1}^{T}(12+3\epsilon^{-2})W_t\leq \nu$. Assuming $\nu<1$, this implies $\sum_{t=1}^{T}6W_t\leq \nu/2, \sum_{t=1}^{T}3(1+\epsilon^{-2})^2W_t^2\leq \nu^2/2\leq \nu$. So the second event happens with probability at most $\frac{(12+3\epsilon^{-2})\delta T}{\nu}$.
        \item Similarly, the third event happens with probability at most $\frac{(1+\epsilon^{-2})\delta T}{\eta}$.
        \item We already know from the Freedman's inequality that the fourth event happens with probability at most $\exp(-\frac{\eta^2}{2\nu+6T\mu^2+4\epsilon\eta/3})$.
    \end{enumerate}
    Therefore, 
    \begin{equation}
        \Pr[L\leq e^{-2\eta-T\mu}] \leq \epsilon^{-2}\delta T+\frac{(12+3\epsilon^{-2})\delta T}{\nu}+\frac{(1+\epsilon^{-2})\delta T}{\eta}+\exp(-\frac{\eta^2}{2\nu+6T\mu^2+4\epsilon\eta/3}). \label{eq:appendix_1}
    \end{equation}
    
    We choose $\mu$ such that $e^{-2\eta}=1-c_1$. Then we choose sufficiently small $\epsilon$ and $\nu$ such that 
    \begin{equation*}
        -\frac{\eta^2}{2\nu+6T\mu^2+4\epsilon\eta/3}=-\frac{\eta^2}{2\nu+4\epsilon\eta /3} + \frac{6\eta^2T\mu^2}{(2\nu+6T\mu^2+4\epsilon\eta/3)(2\nu+4\epsilon\eta/3)}\leq \ln(c_2)+c_4T\mu^2,
    \end{equation*}
    where $c_4\coloneqq 6\eta^2/(2\nu+4\epsilon\eta/3)^2$. Choose $c_3 \coloneqq \epsilon^{-2}+(12+3\epsilon^{-2})/\nu+(1+\epsilon^{-2})/\eta$. \eqref{eq:appendix_1} becomes
    \begin{equation*}
        \Pr[L\leq e^{-T\mu}(1-c_1)] \leq c_2e^{c_4T\mu^2}+c_3\delta T.
    \end{equation*} 
    So the lemma is proved.
\end{proof}

\begin{remark}
    In the proof of \Cref{lem:MartingaleTrick}, we set $Y_t(X_{0:t})=L_a(u_t, s_t)-1$, so $\E[Y_t|X_{0:t-1}]=0$ is always 0. However, it is sometimes useful to set $Y_t(X_{0:t})$ to some quantity smaller than $L_a(u_t, s_t)-1$ to reduce the variance $\E[Y_t^2|X_{1:t-1}]$. In this case, we need to consider $\mu>0$ in \Cref{lem:stochastic_process}. For example, see \cite[Lemma D.5]{chen2022complexity}.
\end{remark}

%%%%%%%%%%%%%%%%%%%%%%%%%%%%%%%%%%%%%%%%%%%%%%%%%%%%%%%%%%%%%

\end{document}